\documentclass[english]{article}
\usepackage[T1]{fontenc}
\usepackage[a4paper]{geometry}
\usepackage{dsfont}
\usepackage{amsmath}
\usepackage{amsthm}
\usepackage{makecell}
\usepackage{amssymb}
\usepackage{stmaryrd}
\usepackage{graphicx}
\usepackage{setspace}
\usepackage{esint}
\usepackage{caption}

\usepackage[authoryear]{natbib}
\usepackage[unicode=true,bookmarks=true,bookmarksnumbered=false,bookmarksopen=false,breaklinks=false,pdfborder={0 0 0},pdfborderstyle={},backref=page,colorlinks=true]{hyperref}
\hypersetup{pdftitle={Adversarial Transform Particle Filters},pdfauthor={Chengxin Gong, Wei Lin, Cheng Zhang},linkcolor=RoyalBlue,citecolor=RoyalBlue}
\usepackage{float}
\usepackage[dvipsnames,svgnames,x11names,hyperref]{xcolor}
\usepackage{nicefrac}

\usepackage{epstopdf}
\usepackage{url} 
\usepackage{caption}
\usepackage{multirow}
\usepackage{amssymb}
\usepackage{multicol}
\usepackage{booktabs}
\usepackage{bbm}
\usepackage{bm}
\usepackage{tikz}
\usetikzlibrary{shapes.geometric, arrows}
\usetikzlibrary{positioning}
\tikzstyle{circle_node} = [circle, draw, minimum size=2.5em, inner sep=-0.1em]
\tikzstyle{diamond_node} = [diamond, draw, minimum size=3.1em, inner sep=-0.1em]
\tikzstyle{rounded_rectangle} = [rectangle, draw, rounded corners=11pt, minimum width=2.5em, minimum height=2.5em, text centered]
\tikzstyle{arrow} = [->, thick]
\usepackage{algorithm}
\usepackage{algorithmic}

\usepackage{subfigure}
\usepackage{mathrsfs}
\usepackage{enumerate}

\usepackage{amsmath,amsfonts,bm}

\def\eqref#1{equation~\ref{#1}}

\def\1{\bm{1}}

\DeclareMathAlphabet{\mathsfit}{\encodingdefault}{\sfdefault}{m}{sl}
\SetMathAlphabet{\mathsfit}{bold}{\encodingdefault}{\sfdefault}{bx}{n}

\newcommand{\Var}{\mathrm{Var}}

\theoremstyle{plain}
\newtheorem{theorem}{\textbf{Theorem}}

\newtheorem{definition}{\textbf{Definition}}
\newtheorem{lemma}{\textbf{Lemma}}

\newtheorem{corollary}{\textbf{Corollary}}

\newtheorem{proposition}{\textbf{Proposition}}
\newtheorem{assumption}{\textbf{Assumption}}

\DeclareMathOperator{\trace}{trace}

\newcommand{\EE}{\mathbb{E}}

\usepackage[textsize=tiny]{todonotes}

\makeatletter
\newcommand{\myfnsymbol}[1]{%
  \expandafter\@myfnsymbol\csname c@#1\endcsname
}
\newcommand{\@myfnsymbol}[1]{%
  \ifcase #1
  \or 1
  \or 2
  \or 3
  \or \TextOrMath{\textasteriskcentered}{*}
  \fi
}
\newcommand{\affiliationA}{\@myfnsymbol{1}}
\newcommand{\affiliationB}{\@myfnsymbol{2}}
\newcommand{\affiliationC}{\@myfnsymbol{3}}
\newcommand{\correspondingA}{\@myfnsymbol{4}}
\makeatother

\begin{document}

\title{Adversarial Transform Particle Filters}

\author{
Chengxin Gong\textsuperscript{\affiliationA},
Wei Lin\textsuperscript{\affiliationB},
Cheng Zhang\textsuperscript{\affiliationC,\correspondingA}
}

\date{
}

\renewcommand{\thefootnote}{\myfnsymbol{footnote}}
\maketitle
\footnotetext[1]{School of Mathematical Sciences and Center for Statistical Science, Peking University,
   Beijing, 100871, China. Email: gongchengxin@pku.edu.cn}%
\footnotetext[2]{School of Mathematical Sciences and Center for Statistical Science, Peking University,
Beijing, 100871, China. Email: weilin@math.pku.edu.cn}%
\footnotetext[3]{School of Mathematical Sciences and Center for Statistical Science, Peking University, Beijing, 100871, China. Email: chengzhang@math.pku.edu.cn}%
\footnotetext[4]{Corresponding author}%

\setcounter{footnote}{0}
\renewcommand{\thefootnote}{\fnsymbol{footnote}}

\begin{abstract}
  The particle filter (PF) and the ensemble Kalman filter (EnKF) are widely used for approximate inference in state-space models. From a Bayesian perspective, these algorithms represent the prior by an ensemble of particles and update it to the posterior with new observations over time.
  However, the PF often suffers from weight degeneracy in high-dimensional settings, whereas the EnKF relies on linear Gaussian assumptions that can introduce significant approximation errors.
  In this paper, we propose the Adversarial Transform Particle Filter (ATPF), a novel filtering framework that combines the strengths of the PF and the EnKF through adversarial learning.
  Specifically, importance sampling is used to ensure statistical consistency as in the PF, while adversarially learned transformations, such as neural networks, allow accurate posterior matching for nonlinear and non-Gaussian systems.
  In addition, we incorporate kernel methods to ease optimization and leverage regularization techniques based on optimal transport for better statistical properties and numerical stability.
  We provide theoretical guarantees, including generalization bounds for both the analysis and forecast steps of ATPF.
  Extensive experiments across various nonlinear and non-Gaussian scenarios demonstrate the effectiveness and practical advantages of our method.
\end{abstract}




\section{Introduction}
\label{sec:intro}
Data assimilation plays a critical role in fields such as meteorology
\citep{rabier2005overview}, oceanography \citep{bertino2003sequential}, and geophysics \citep{fletcher2022data}.
In these fields, state-space models (SSMs) are often employed to describe the evolution dynamics of the system state, and the relationship of the latent states to observed data.
The goal of data assimilation is to obtain an estimate of the true latent states of the system and the associated uncertainty of that estimate.
At each time step, the learning and inference process consists of two sequential alternating phases: analysis and forecast.
During the analysis phase, current observations are used to update prior knowledge of the latent state, resulting in a posterior distribution that provides a more accurate estimate of the current state.
In the forecast phase, the posterior distribution is propagated forward using the evolution dynamics to predict the distribution and likely value of the next state.
While numerous methods exist to perform these two steps \citep{evensen1994sequential,gordon1993novel,anderson2001ensemble,lei2011moment}, the forecast phase generally relies on predefined evolution dynamics. Therefore, in this paper we focus on improving the analysis step to efficiently transform prior particles into posterior ones for more accurate latent state estimation.
There is an enormous statistics and engineering literature on state estimation for data assimilation.
For linear, Gaussian models, the Kalman filter \citep{kalman1960} provides an efficient closed-form solution for posterior state estimation given the observations.
For more general SSMs with nonlinear, non-Gaussian models, Markov chain Monte Carlo (MCMC) methods have been developed \citep{carlin1992,Shephard1997,Gamerman1998}.
While these methods are effective, they struggle to scale to high-dimensional states in view of the rapidly growing computational cost (typically cubically) as the state dimension increases.
An alternative approach for general SSMs is sequential Monte Carlo, also known as the particle filter (PF; \citet{gordon1993novel}).
The PF approximate the state distribution using a weighted set of samples or \emph{particles}, which are propagated forward over time according to the state evolution dynamics and updated with new observations via reweighting or resampling, typically using importance sampling. 
Although PFs are exact in the limit of infinite particles, they tend to suffer from weight degeneracy (i.e., all weights but one become essentially zero), especially in high dimensions \citep{snyder2008obstacles}.
While various PF variants have been proposed to alleviate this issue \citep{arulampalam2002tutorial, nakano2007merging, xiong2006note, murphy2001rao}, scaling PFs to high-dimensional problems remains a significant challenge.
Instead of adjusting particle weights, one can update the particles directly to approximate the posterior using an appropriate transformation.
For example, the ensemble Kalman filter (EnKF; \citet{evensen1994sequential}) updates the ensemble of particles using a linear ``shift" that approximates the best linear update, assuming a linear Gaussian model.
While this approach enhances particle diversity, it sacrifices statistical consistency in the context of nonlinear and non-Gaussian models.
Furthermore, in high-dimensional scenarios, the EnKF requires additional techniques to avoid rank deficiency and underestimation of forecast-error covariance matrices \citep{sun2024high}.
Three main methods are typically used: dimension localization (e.g., \citet{hunt2007efficient,montemerlo2003fastslam}), covariance tapering (e.g., \citet{anderson2007adaptive,furrer2006covariance}) and inflation (e.g., \citet{evensen2003ensemble,lei2011moment}).
These approaches assume that correlations decrease rapidly with spatial distance, implying that the covariance matrix should be sparse.
Besides these, numerous EnKF variants, such as the ensemble transform Kalman filter (ETKF; \citet{bishop2001adaptive}) and the ensemble adjustment Kalman filter (EAKF; \citet{anderson2001ensemble}), have been developed to improve approximation accuracy and reduce computational burden.

In this paper, we propose the adversarial transform particle filter (ATPF), a new framework for data assimilation in SSMs that combines the statistical consistency of the PF with the particle diversity of the EnKF through the lens of adversarial learning.
In the field of generative models, adversarial learning has proven successful, as exemplified by generative adversarial networks (GANs; \citet{goodfellow2014generative}) and further Wasserstein GANs (W-GANs; \citet{arjovsky2017wasserstein}), which can effectively approximate complex target data distributions.
From a statistical perspective, W-GANs distinguish between two distributions, $P$ and $Q$, using the maximum mean discrepancy (MMD), also known as the integral probability metric (IPM) or adversarial loss.
Specifically, a set of test functions, denoted by $\mathcal{F}_D$, acts as a discriminator to measure the dissimilarity between $P$ and $Q$, which leads to the MMD defined as follows:
\begin{equation}
  \mathrm{MMD}(P\|Q)=\sup_{f\in\mathcal{F}_D}\left|\mathbb{E}_{x\sim P}f(x)-\mathbb{E}_{x\sim Q}f(x)\right|.
  \label{eq:MMD}
\end{equation}
Similar to the EnKF, we aim to learn a transformation function to update prior particles toward the posterior.
As shown in \eqref{eq:MMD}, only the expectations of the test functions under $P$ and $Q$ are needed.
By leveraging importance sampling, as in PFs, we can estimate these expectations under the posterior distribution and train the transformation function by minimizing the MMD between the pushforward of the prior and the posterior.
Unlike the EnKF, this adversarial framework allows us to match arbitrary posterior distributions, making it particularly well-suited for nonlinear and non-Gaussian models.
Moreover, the approach inherits key advantages from both the PF and the EnKF.
On one hand, the transformation formulation and adversarial loss in \eqref{eq:MMD} mitigates particle collapse.
On the other hand, the posterior expectations from PFs ensure statistical consistency in the estimates.

The remainder of the paper is organized as follows.
In Section \ref{sec:SSMs}, we provide a brief introduction to state-space models (SSMs) and review existing filtering methods.
Our proposed method is presented in detail in Section \ref{sec:ATPF}, followed by an exploration of its theoretical properties in Section \ref{sec:theory}.
Section \ref{sec:experi} presents the experimental results that validate our method.
We conclude the paper and offer further discussions in Section \ref{sec:conc}.
All proofs are provided in Appendix \ref{app:AppendixA}.

\section{State-Space Models and Ensemble Methods}
\label{sec:SSMs}
\subsection{State-Space Models}
\label{subsec:SSMs}
We begin with a brief introduction to state-space models (SSMs), which describe the temporal evolution of the latent state and how it relates to the observations.
The SSM is defined as
\begin{equation}
  \begin{aligned}
    x_t&=\mathcal{M}_{t-1}(x_{t-1})+\eta_t,\\
    y_t&=\mathcal{H}_t(x_t)+\varepsilon_t,
  \end{aligned}
  \label{eq:SSM}
\end{equation} where $x_t,\eta_t\in\mathbb{R}^{d_{x_t}}$ denote the latent state and the evolution noise, $y_t,\varepsilon_t\in\mathbb{R}^{d_{y_t}}$ denote the observed data and the observation noise, the evolution operator $\mathcal{M}_{t-1}:\mathbb{R}^{d_{x_{t-1}}}\mapsto\mathbb{R}^{d_{x_t}}$ and the observation operator $\mathcal{H}_t:\mathbb{R}^{d_{x_t}}\mapsto\mathbb{R}^{d_{y_t}}$ are both measurable functions, for each time $t=1,\dots,T$. Without loss of generality, we set $\EE\eta_t=\EE\varepsilon_t=0$.
In the data assimilation problem, it is commonly assumed that (1) the analytical forms of $\mathcal{M}_{t-1},\mathcal{H}_t$ are fully known; (2) $\eta_t,\varepsilon_t$ are mutually independent for all $t=1,\dots,T$; and (3) densities $p_{\eta_t}$ and $p_{\varepsilon_t}$ are simple and straightforward to sample from. At each time $t$, all filtering methods perform two sequential steps: the analysis step and the forecast step. The goal of the analysis step is to transform the prior $p_t(x_t):=p(x_t|y_{1:t-1})$ into the posterior $q_t(x_t):=p(x_t|y_{1:t})$, while the forecast step aims to estimate the prior $p_{t+1}(x_{t+1}):=p(x_{t+1}|y_{1:t})$ for the next time step.

The optimal transformation in the analysis step is naturally induced by Bayes' rule, which gives the posterior
\begin{equation}
  q_t(x_t)=\frac{p_t(x_t)p_{\varepsilon_t}(y_t-\mathcal{H}_t(x_t))}{\int p_t(x)p_{\varepsilon_t}(y_t-\mathcal{H}_t(x))\text{d}x}.
  \label{eq:Bayes}
\end{equation}
However, a closed-form solution to this equation is only tractable in a few special cases, such as when $\mathcal{M}_t,\mathcal{H}_t$ are linear functions and $\eta_t$ and $\varepsilon_t$ are Gaussian variables, or when $x_t$ is discrete. In more general scenarios, ensemble filtering methods are often employed to approximate the solution, where $p_t(x_t),q_t(x_t)$ are represented by an ensemble of particles and the updating rules operate directly on these particles. Two essential methods for this are the PF and the EnKF, which we introduce next.

\subsection{Particle Filter}
\label{subsec:pf}
Particle filters (PFs) \citep{gordon1993novel} update the prior particles by recalculating their weights according to the associated likelihood of the observed data using self-normalized importance sampling (SNIS; \citet{kong1992note}).
At time $t$, given the prior particles $\{x_t^1,\dots,x_t^n\}$ and the corresponding weights $\{w_t^1,\dots,w_t^n\}$, the vanilla PF algorithm works as follows:

\noindent\textbf{PF Step 1}: Compute the likelihood $l_t^i=p_{\varepsilon_t}(y_t-\mathcal{H}_t(x_t^i))$ for $i=1,\dots,n$.

\noindent\textbf{PF Step 2}: Update the weights $\tilde{w}_{t+1}^i=w_t^il_t^i$ for $i=1,\dots,n$ and renormalize $w_{t+1}^i \propto \tilde{w}_{t+1}^i$ such that $\sum_{i=1}^nw_{t+1}^i=1$.

\noindent\textbf{PF Step 3}: Sample $x_{t+1}^i=\mathcal{M}_{t}(x_{t}^i)+\eta_{t+1}^i$ with $\eta_{t+1}^i\stackrel{\text{i.i.d.}}{\sim}p_{\eta_{t+1}}$ for $i=1,\dots,n$ to obtain forecast particles $\{x_{t+1}^1,\dots,x_{t+1}^n\}$ with weights $\{w_{t+1}^1,\dots,w_{t+1}^n\}$.

The PF method is asymptotically consistent since the distribution induced by the weighted particles converge weakly to the true posterior as the number of particles goes to infinity \citep{kunsch2005recursive}.
When the effective sample size $\displaystyle\text{ESS}:=\frac{1}{\sum_{i=1}^n(w_t^i)^2}$ becomes relatively small, which is likely to occur in high-dimensional settings, an additional resampling step can be introduced between steps 2 and 3 to mitigate weight degeneracy:

\noindent\textbf{PF Step 2a}: Resample from the ensemble $\{x_t^1,\dots,x_t^n\}$ with weights $\{w_{t+1}^1,\dots,w_{t+1}^n\}$ for $n$ times, that is, sample from the distribution $\sum_{i=1}^nw_{t+1}^i\delta(x-x_t^i)$ for $n$ times. Then set $w_{t+1}^i\equiv\frac{1}{n}$ for $i=1,\dots,n$.

As mentioned in Section \ref{sec:intro}, the resampling step may fail to improve particle diversity when the evolution law is deterministic, i.e., $\mathbb{P}(\eta_t=0)=1$ for all $t=1,\dots,T$, as identical particles will remain identical as time evolves.
To address this issue, the regularized PF method \citep{musso2001improving,arulampalam2002tutorial} adopts the kernel smoothing trick. Let $m(\cdot)$ denote a kernel density function satisfying $\int m(x)\text{d}x=1$. After step 2 or step 2a, we can take the following step:

\noindent\textbf{PF Step 2b}: Sample from the kernelized mixture distribution $\sum_{i=1}^nw_{t+1}^im(x-x_t^i)$ for $n$ times. Then set $w_{t+1}^i\equiv\frac{1}{n}$ for $i=1,\dots,n$.

\subsection{Ensemble Kalman Filter}
\label{subsec:enkf}
For each time $t$, assume a linear Gaussian state-space model
\begin{align}
  x_t & = M_tx_{t-1} + \eta_t, \label{eq:linevo} \\
  y_t & =H_tx_t+\varepsilon_t, \label{eq:lineobs}
\end{align}
where $M_t\in\mathbb{R}^{d_{x_t}\times d_{x_{t-1}}},\eta_t\sim\mathcal{N}_{d_{x_t}}(0,Q_t), H_t\in\mathbb{R}^{d_{y_t}\times d_{x_t}},\varepsilon_t\sim\mathcal{N}_{d_{y_t}}(0,R_t)$.
When the prior distribution is Gaussian, the posterior remains Gaussian and explicit update formulas for the mean and variance parameters can be derived.
This is the well-known Kalman filter \citep{kalman1960,kim2018introduction}.
In high dimensions, however, maintaining and propagating full covariance matrices becomes computationally expensive and memory demanding.
To address these issues, \citet{evensen1994sequential} proposed the EnKF which represents the state distribution using an ensemble of particles.
This ensemble is propagated forward through time and updated as new observations become available.
Given the prior forecast ensemble $\{x_t^1,\dots,x_t^n\}$ and the perturbed observation $y_t$, the EnKF proceeds as follows:

\noindent\textbf{EnKF Step 1}: Estimate the covariance of $x_t$ by the sample covariance: $P_t:=\widehat{\text{Cov}}(x_t,x_t)=\frac{1}{n-1}\sum_{i=1}^n(x_t^i-\bar x_t)(x_t^i-\bar x_t)^T$ where $\bar x_t=\frac{1}{n}\sum_{i=1}^nx_t^i$ is the forecast mean.

\noindent\textbf{EnKF Step 2}: Estimate the Kalman gain as $K_t:=P_tH_t^T(H_tP_tH_t^T+R_t)^{-1}.$

\noindent\textbf{EnKF Step 3}: For each $i=1,\dots,n$, sample $\varepsilon_t^i\stackrel{\text{i.i.d.}}{\sim}\mathcal{N}(0,R_t)$ and generate the corresponding pseudo-observation $y_t^i$ as $ y_t^i=H_tx_t^i+\varepsilon_t^i.$

\noindent\textbf{EnKF Step 4}: For each $i=1,\dots,n$, update the particle $x_t^i$ as\begin{equation}
  x_t^i \leftarrow x_t^i + K_t(y_t-y_t^i).
  \label{eq:kalupd}
\end{equation}

\noindent\textbf{EnKF Step 5}: For each $i=1,\dots,n$, sample $x_{t+1}^i=\mathcal{M}_t(x_t^i)+\eta_{t+1}^i$ with $\eta_{t+1}^i\stackrel{\text{i.i.d.}}{\sim}p_{\eta_{t+1}}$.

In high dimensions, techniques such as covariance inflation or covariance tapering can be incorporated between steps 4 and 5, while localization can be applied during step 4.
These strategies, which will be elaborated in Appendix \ref{app:AppendixC}, enhance the performance of the EnKF in complex settings.
However, the linear transformations used in the EnKF are often suboptimal for nonlinear or non-Gaussian models.
While some non-Gaussian EnKF approaches have been proposed \citep{anderson2010nonGaussian, lei2011moment}, they still exhibit poor performance for certain classes of measurement distributions \citep{katzfuss2020ensemble}.

\section{Adversarial Transform Particle Filter}
\label{sec:ATPF}
In this section, we present the adversarial transform particle filter (ATPF), a novel adversarial learning framework for data assimilation that integrates the particle diversity of the EnKF with the statistical consistency of the PF.
The key motivations for ATPF are: (1) extending the linear update equation in the EnKF (i.e., \eqref{eq:kalupd}) and similar methods to more flexible, general transformations; (2) replacing the moment-matching criteria in the nonlinear ensemble adjustment filter (NLEAF; \citet{lei2011moment}) with a maximum mean discrepancy (MMD) objective, which offers greater effectiveness for distributional matching.

Let $\mathcal{G}_t$ and $\mathcal{F}_D$ be two pre-specified function spaces, where $\mathcal{G}_t$ represents the collection of transformation functions and $\mathcal{F}_D$ represents the collection of test functions.
For a given time $t$, let $G_t\in\mathcal{G}_t$ denote the transformation function, which takes the prior particles as inputs and outputs the estimated posterior samples. Let $\hat P_t$ and $P_t$ be the estimated and the ground truth prior, respectively, where $P_t$ may be estimated via PFs or Markov chain Monte Carlo (MCMC) methods. Similarly, let $\hat Q_t$ and $Q_t$ be the corresponding estimated and true posterior distributions.
Specifically, $\hat Q_t=G_{t\#}\hat P_t$ (``$_\#$'' denotes the pushforward operator) and $Q_t$ can be estimated with a reweighted version of $P_t$ via importance sampling.
Our goal is to ensure that $\hat Q_t$ closely approximates the true posterior $Q_t$ under the adversarial loss \eqref{eq:MMD}.
This leads to the oracle objective of minimizing the MMD:
\begin{equation}\small
  G_t=\underset{G_t\in\mathcal{G}_t}{\arg\min}\sup_{f\in\mathcal{F}_D}\left|\mathbb{E}_{x\sim \hat Q_t}f(x)-\mathbb{E}_{x'\sim Q_t}f(x')\right|=\underset{G_t\in\mathcal{G}_t}{\arg\min}\sup_{f\in\mathcal{F}_D}\left|\mathbb{E}_{x\sim \hat P_t}f(G_t(x))-\mathbb{E}_{x'\sim P_t}b_t(x')f(x')\right|,
  \label{eq:ATPFobj}
\end{equation}
where $b_t(\cdot)$ is a weighting function defined as:
\begin{equation}
  b_t(x')=\frac{\text{d}Q_t}{\text{d}P_t}=\frac{p_{\varepsilon_t}(y_t-\mathcal{H}_t(x'))}{\int p_t(x)p_{\varepsilon_t}(y_t-\mathcal{H}_t(x))\text{d}x}.
\end{equation}
Since computing the normalizing constant $\int p_t(x_t)p_{\varepsilon_t}(y_t-\mathcal{H}_t(x_t))\text{d}x_t$ is intractable, we adopt the self-normalized importance sampling (SNIS; \citet{hesterberg1988advances}) instead. By leveraging the empirical distributions of the priors, the transformation function $\hat G_t$ can be learned by minimizing the empirical MMD:

\begin{equation}
  \hat G_t=\underset{G_t\in\mathcal{G}_t}{\arg\min}\;\widehat{\mathrm{MMD}}(\hat Q_t\| Q_t)=\underset{G_t\in\mathcal{G}_t}{\arg\min}\sup_{f\in\mathcal{F}_D}\left|\frac{1}{N_1}\sum_{i=1}^{N_1}f(G_t(x_t^i))-\sum_{j=1}^{N_2}\frac{w_t(x_t'^j)}{\sum_{j=1}^{N_2}w_t(x_t'^j)}f(x_t'^j)\right|,
  \label{eq:obj}
\end{equation}
where $w_t(x)=w_{t-1}(x)p_{\varepsilon_t}(y_t-\mathcal{H}_t(x))$, $x_t^i\stackrel{\text{i.i.d.}}{\sim}\hat P_t$ for $i=1,\dots,N_1$, and $x_t'^j\stackrel{\text{i.i.d.}}{\sim}P_t$ for $j=1,\dots,N_2$. Since the estimator given by SNIS is biased, $N_2$ is often chosen to be larger than $N_1$ for better estimation of the posterior expectation.

Once $\hat G_t$ is learned, we can immediately obtain the posterior ensemble $\{\hat G_t(x_i)\}_{i=1}^{N_1}$. Then we can execute the forecast step (e.g., EnKF step 5) and obtain the estimate for the next latent state.
The detailed procedure is outlined in Algorithm \ref{ATPF algorithm}.

\begin{algorithm}[!t]
\caption{The adversarial transform particle filter algorithm for SSMs}
\label{ATPF algorithm}
\begin{algorithmic}[1]
\REQUIRE $y_{1:T}$, $\mathcal{M}_{1:T}$, $\mathcal{H}_{1:T}$, $Q_{1:T}$, $R_{1:T}$, $\alpha$, $\tau$

\ENSURE $\{x_1^i\}_{i=1}^{N_1},\dots,\{x_T^i\}_{i=1}^{N_1}$

\STATE Sample $\{x_0^i\}_{i=1}^{N_1},\{x_0'^j\}_{j=1}^{N_2}$ from the prior and set $w_j\equiv\frac{1}{N_2}$ for $j=1,\dots,N_2$;
\FOR{$t=1,\dots,T$}
\STATE Update $w_j$, $j=1,\dots,N_2$ for $\{x_t'^j\}_{j=1}^{N_2}$ according to PF Steps 1 and 2;
\STATE Train $G_t$ with learning rate $\alpha$ until convergence or reach the iteration limit $\tau$;
\STATE Update the particle $x_t^i$ as $x_t^i\leftarrow G_t(x_t^i)$ for $i=1,\dots,N_1$;
\IF{$t \neq T$}
\STATE Resample from $\{x_t'^j\}_{j=1}^{N_2}$ according to their corresponding weights $\{w_j\}_{j=1}^{N_2}$;
\STATE Execute the inflation step \eqref{eq:inflat1} for $\{x_t^i\}_{i=1}^{N_1}$ and $\{x_t'^j\}_{j=1}^{N_2}$, respectively;
\STATE Execute the forecast step (EnKF Step 5) for all particles.
\ENDIF
\ENDFOR
\end{algorithmic}
\end{algorithm}

\subsection{The MMD in Reproducing Kernel Hilbert Spaces}
\label{subsec:ker}
The MMD training objective \eqref{eq:obj} shares a similar form with that of GANs, which suggests that it may inherit challenges such as instability and difficulty in training when $G_t$ and $f$ are optimized alternately.
It turns out that we can eliminate the inner optimization of $f$ by restricting the test functions $\mathcal{F}_D$ to the unit ball of a reproducing kernel Hilbert space (RKHS).

\begin{proposition}[Kernel MMD, \citet{smola2006maximum}]
  \label{prop:KMMD}
  Let $K(\cdot,\cdot)$ denote the reproducing kernel. Under the RKHS constraint $\mathcal{F}_D=\{f:\|f\|_{\mathcal{H}}\leq 1\}$, the squared MMD takes the explicit form
  \begin{equation}
    \mathrm{MMD}^2(P\|Q)=\mathbb{E}_{x,x'\sim P}K(x,x')-2\mathbb{E}_{x\sim P,y\sim Q}K(x,y)+\mathbb{E}_{y,y'\sim Q}K(y,y').
  \end{equation}
\end{proposition}

The above proposition demonstrates that using the kernel method, the challenging bi-level optimization can be reduced to a more standard and tractable problem:
\begin{equation}
  \begin{aligned}
    G_t&=\underset{G_t\in\mathcal{G}_t}{\arg\min}\sup_{f\in\mathcal{F}_D}\left|\mathbb{E}_{x\sim \hat Q_t}f(x)-\mathbb{E}_{y\sim Q_t}f(y)\right|=\underset{G_t\in\mathcal{G}_t}{\arg\min}\,\mathrm{MMD}^2(\hat Q_t\|Q_t)\\
    &=\underset{G_t\in\mathcal{G}_t}{\arg\min}\left\{\mathbb{E}_{x,x'\sim\hat P_t}K(G_t(x),G_t(x'))-2\mathbb{E}_{x\sim \hat P_t,y\sim Q_t}K(G_t(x),y)+\mathbb{E}_{y\sim Q_t,y'\sim Q_t}K(y,y')\right\}.
  \end{aligned}
\end{equation}
We can use the V-statistics \citep{akritas1986empirical,serfling2009approximation} to estimate the first term, and use the SNIS to estimate the second term. The last term can be omitted as it is independent of $G_t$.
With this simplification, we can learn $G_t$ by minimizing the empirical kernel MMD
\begin{equation}
  \hat G_t:=\underset{G_t\in\mathcal{G}_t}{\arg\min}\frac{1}{N_1^2}\sum_{i=1}^{N_1}\sum_{j=1}^{N_1}K(G_t(x_i),G_t(x_j))-\frac{2}{N_1}\sum_{i=1}^{N_1}\sum_{j=1}^{N_2}\frac{w_t(x_j')}{\sum_{j=1}^{N_2}w_t(x_j')}K(G_t(x_i),x_j').
\end{equation}
We refer to this approach as the Kernel Adversarial Transform Particle Filter (KATPF).

\subsection{Optimal Transport Regularization}
\label{subsec:OT}
When updating the prior to the posterior in SSMs, an optimal transport (OT, \citet{reich2013nonparametric}) approach has been employed to identify a transformation that minimizes the expected distance between the two distributions and ensure statistical consistency for infinite samples \citep{reich2013nonparametric}.
Likewise, we can incorporate OT to regularize the transformation as\begin{equation}
  G_t:=\underset{G_t\in\mathcal{G}_t}{\arg\min}\text{ MMD}^2(\hat Q_t\|Q_t)+\lambda_tW_2^2(\hat Q_t,P_t),
  \label{eq:regulobj}
\end{equation} where $W_2$ is the Wasserstein-2 distance. Similar optimization objectives can also be found in other works \citep{jordan1998variational,fan2021variational}.
Intuitively, when the space of the test functions is relatively limited and insufficient (which induces a weak MMD metric), the regularization provided by the $W_2$ distance can help maintain numerical stability during iterative updates without compromising statistical consistency \citep{reich2013nonparametric}.
Furthermore, it can prevent our ATPF particles from overfitting and being too similar to the PF particles, which are known to exhibit several issues such as weight degeneracy and sample impoverishment \citep{poterjoy2016localized,snyder2008obstacles}.

In practice, we may use the coupling induced by the transformation to form a surrogate (upper bound) of the $W_2$ distance, which leads to the following regularized optimization
\begin{equation}
     \tilde{G}_t:=\underset{G_t\in\mathcal{G}_t}{\arg\min}\text{ MMD}^2( G_{t\#}P_t\|Q_t)+\lambda_t\mathbb{E}_{x\sim P_t}\|G_t(x)-x\|_2^2.
     \label{eq:surrogate}
\end{equation}
Lemma A.1 in \citet{xu2024normalizing} reveals the equivalence between \eqref{eq:regulobj} and \eqref{eq:surrogate}. Furthermore, the surrogate can be estimated as\begin{equation}
  \mathbb{E}_{x\sim P_t}\|G_t(x)-x\|_2^2\approx\frac{1}{N_1}\sum_{i=1}^{N_1}\|G_t(x_i)-x_i\|^2.
\end{equation}

\subsection{Neural Network Parameterization and Joint Optimization}
\label{subsec:ATPF-JO}
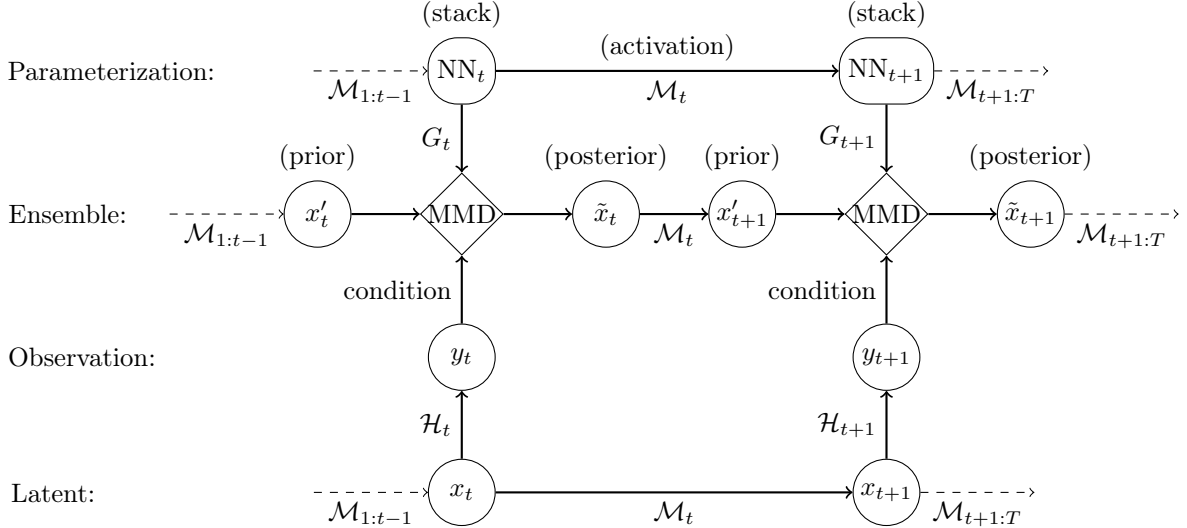
\begin{figure}[t]
  \centering
  \normalsize{
  \begin{tikzpicture}[node distance=0.9cm]

    \node[rounded_rectangle, label=above:{(stack)}] (NNt) at (0, 2) {NN$_t$};
    \node[diamond_node, below=of NNt] (MMD) {MMD};
    \node[circle_node, label=above:{(prior)}, left=of MMD] (xt_prime) {$x'_t$};
    \node[circle_node, label=above:{(posterior)}, right=of MMD] (xt_tilde) {$\tilde{x}_t$};
    \node[circle_node, label=above:{(prior)}, right=of xt_tilde] (xt1_prime) {$x_{t+1}'$};
    \node[circle_node, below=of MMD] (yt) {$y_t$};
    \node[circle_node, below=of yt] (xt) {$x_t$};
    \node[diamond_node, right=of xt1_prime] (MMD2) {MMD};
    \node[circle_node, below=of MMD2] (yt1) {$y_{t+1}$};
    \node[circle_node, label=above:{(posterior)}, right=of MMD2] (xt1_tilde) {$\tilde{x}_{t+1}$};
    \node[circle_node, below=of yt1] (xt1) {$x_{t+1}$};
    \node[rounded_rectangle, label=above:{(stack)}, above=of MMD2] (NNt1) {NN$_{t+1}$};
    
    \draw[arrow] (NNt) -- (MMD) node[midway, left] {$G_t$};
    \draw[arrow] (NNt) -- (NNt1) node[midway, below] {$\mathcal{M}_t$} node[midway, above] {(activation)};
    \draw[arrow] (xt_tilde) -- (xt1_prime) node[midway, below] {$\mathcal{M}_t$};
    \draw[arrow] (yt) -- (MMD) node[midway, left] {condition};
    \draw[arrow] (xt) -- (yt) node[midway, left] {$\mathcal{H}_t$};
    \draw[arrow] (xt_prime) -- (MMD);
    \draw[arrow] (xt1_prime) -- (MMD2);
    \draw[arrow] (MMD) -- (xt_tilde);
    \draw[arrow] (yt1) -- (MMD2) node[midway, left] {condition};
    \draw[arrow] (xt1) -- (yt1) node[midway, left] {$\mathcal{H}_{t+1}$};
    \draw[arrow] (NNt1) -- (MMD2) node[midway, left] {$G_{t+1}$};
    \draw[arrow] (xt) -- (xt1) node[midway, below] {$\mathcal{M}_t$};
    \draw[arrow] (MMD2) -- (xt1_tilde);
    \draw[dashed, ->] ([xshift=-1.5cm]NNt.west) -- (NNt.west) node[midway, below]
    {$\mathcal{M}_{1:t-1}$};
    \draw[dashed, ->] (NNt1.east) -- ([xshift=1.5cm]NNt1.east) node[midway, below]
    {$\mathcal{M}_{t+1:T}$};
    \draw[dashed, ->] ([xshift=-1.5cm]xt_prime.west) -- (xt_prime.west) node[midway, below] {$\mathcal{M}_{1:t-1}$}; 
    \draw[dashed, ->] ([xshift=-1.5cm]xt.west) -- (xt.west) node[midway, below] {$\mathcal{M}_{1:t-1}$};
    \draw[dashed, ->] (xt1.east) -- ([xshift=1.5cm]xt1.east) node[midway, below] {$\mathcal{M}_{t+1:T}$};
    \draw[dashed, ->] (xt1_tilde.east) -- ([xshift=1.5cm]xt1_tilde.east) node[midway, below] {$\mathcal{M}_{t+1:T}$};
    
    \node at ([xshift=-4.16cm]NNt.west) {Parameterization:};
    \node at ([xshift=-2.84cm]xt_prime.west) {Ensemble:};
    \node at ([xshift=-4.55cm]yt.west) {Observation:};
    \node at ([xshift=-4.93cm]xt.west) {Latent:};
  \end{tikzpicture}
  \caption{Our parameterization methods and flow chart of filtering for SSMs.}
\label{fig:NNparameterization}
}
\end{figure}

The rich family of test functions used in ATPF provides better distributional matching capability than the NLEAF methods that only use the first and second moments. Moreover, compared to the EnKF that uses linear transformations, we use more flexible nonlinear transformations based on neural networks for more accurate posterior estimation.

While it is natural to train these transformations sequentially, we can also train them jointly by minimizing the sum of the MMD and $W_2$ regularization terms across all observation times:
\begin{equation}
  \hat G_{1:T}=\underset{(G_1,\dots,G_T)\in\mathcal{G}_1\times\dots\times\mathcal{G}_T}{\arg\min}\sum_{i=1}^T\left(\widehat{\mathrm{MMD}}(\hat Q_t\|Q_t)+\frac{\lambda_t}{N_1}\sum_{i=1}^{N_1}\|G_t(x_i^t)-x_i^t\|^2\right).
\end{equation}
The difference is as follows.
Specifically, the $t$-th term $\widehat{\mathrm{MMD}}(\hat Q_t\|Q_t)+\frac{\lambda_t}{N_1}\sum_{i=1}^{N_1}\|G_t(x^t_i)-x^t_i\|^2$, which we refer to as $\widehat{\mathrm{MMD}}_t$, involves the first $t$ transformation functions $G_1,\dots,G_t$.
By viewing the forecast step as a special activation layer, training $G_{1:t}$ together under the loss $\widehat{\mathrm{MMD}}_t$ is equivalent to training a deep neural network(NN), which can be seen as a stack of $t$ small NNs (Figure \ref{fig:NNparameterization}).
This joint optimization procedure, therefore, allows the transformations at different time steps to collaborate more effectively with each other for a better fit of the target posteriors.
In contrast, in the sequential optimization approach, the transformations from earlier time steps are fixed once trained, preventing them from aiding subsequent transformations in improving their approximations.
We refer to this joint optimization version of ATPF as ATPF-JO.
\section{Theoretical Properties}
\label{sec:theory}
\subsection{Upper Bounds on the Excess Risk}
In this section, we apply techniques from empirical process theory to derive bounds on the excess risk for our method. To ensure that the model remains well-posed and avoids singularities or divergence, we first introduce some common assumptions on the regularity of SSMs and the collection of test functions \citep{agapiou2017importance,distfreebook,lange2024bellman}.

\begin{assumption}[Continuity and differentiability]\label{ass:1} For each $t=1,\dots,T$, $\mathcal{M}_t,\mathcal{H}_t\in C(\mathbb{R}^d)$ and $p_{\eta_t},p_{\varepsilon_t}\in C^1(\mathbb{R}^d)$.
\end{assumption}

\begin{assumption}[Bounded weights and test functions]\label{ass:2} $\sup_{1\le t\le T}\|w_t\|_\infty\vee\|\mathcal{F}_D\|_\infty\leq B$, where $\{w_t\}_{t=1}^T$ are the weights of the SNIS particles and $\|\mathcal{F}_D\|_\infty:=\sup_{f\in\mathcal{F}_D}\|f\|_\infty$.
\end{assumption}

\begin{assumption}[Bounded variance]\label{ass:3} $\sup_{1\le t\le T}\mathbb{E}\eta_t^2\vee\mathbb{E}\varepsilon_t^2<\infty$.
\end{assumption}

To better understand the excess risk, we first decompose it into three parts.

\begin{theorem}[Decomposition of the excess risk]
  \label{thm:boundrisk}
  The excess risk of the analysis step (also called the statistical error, i.e., $\mathrm{MMD}(\hat Q_t\|Q_t)-\widehat{\mathrm{MMD}}(\hat Q_t\|Q_t)$), can be uniformly upper bounded by the following sum of three components:\begin{equation}\label{eq:decomposition}\footnotesize\begin{aligned}
    \sup_{f\in\mathcal{F}_D}\left|\underset{y\sim {Q_t}}{\EE}f(y)-\underset{x'\sim P_t}{\EE}\left[\frac{\sum_{j=1}^{N_2}w_t(x'_j)f(x'_j)}{\sum_{j=1}^{N_2}w_t(x'_j)}\right]\right|&+\sup_{f\in\mathcal{F}_D}\left|\underset{x'\sim P_t}{\EE}\left[\frac{\sum_{j=1}^{N_2}w_t(x'_j)f(x'_j)}{\sum_{j=1}^{N_2}w_t(x'_j)}\right]-\frac{\sum_{j=1}^{N_2}w_t(x'_j)f(x'_j)}{\sum_{j=1}^{N_2}w_t(x'_j)}\right|\\
    &+\sup_{h\in\mathcal{F}_D\circ\mathcal{G}_t}\left|\frac{1}{N_1}\sum_{i=1}^{N_1}h(x_i)-\underset{x\sim \hat P_t}{\EE}h(x)\right|
  \end{aligned}
  \end{equation}
  for all $\hat Q_t=G_{t\#}\hat P_t$ where $G_t\in\mathcal{G}_t$.
\end{theorem}

The first term represents the bias of the SNIS, which has a convergence rate of $O(\frac{1}{N_2})$ (Theorem 2.1 in \citet{agapiou2017importance}, Theorem \ref{thm:SNISbias} in Appendix \ref{app:AppendixA}).
The second term reflects the randomness inherent in the stochastically weighted empirical process of the SNIS estimator. Ensuring a sharp convergence rate for this term requires additional techniques and concentration conditions (see Lemma \ref{lem:lip} in Appendix \ref{app:AppendixA}).
Specifically, the following assumptions are commonly used in related work \citep[e.g.,][]{liang2021well,oko2023diffusion,magazinov2022concentration,saumard2014log}.

\begin{assumption}[Smoothness]\label{ass:4} Both $f(x)$ and $w_t(x)$ are differentiable for all $f\in\mathcal{F}_D$ and $t=1,\dots,T$.
Moreover, they are $\beta$-Lipschitz, i.e., $\|\nabla f(x)\|_2\vee\|\nabla w_t(x)\|\leq\beta$.
\end{assumption}

\begin{assumption}[Concavity]\label{ass:5} For all $t=1,\dots,T$, the distribution $P_t$ is strongly log-concave with parameter $\gamma_t>0$.
\end{assumption}

Note that Assumptions 4 and 5 are not necessary for the convergence of our method. However, they play a crucial role to ensure a faster exponential rate of $O(\exp(-N_2))$. Without these assumptions, the rate reduces to $O(\frac{1}{N_2})$.
For further details, please refer to Theorems \ref{thm:bound2} and \ref{thm:w/oconvex} in Appendix \ref{app:AppendixA}.
With these assumptions and the decomposition of the excess risk, we derive the following nonasymptotic upper bound for our proposed method.

\begin{theorem}[Upper bound on the excess risk]\label{thm:simplefinalupperbound} Assume Assumptions \ref{ass:1}--\ref{ass:5} hold. Let $\mathscr{N}_k(\varepsilon,\cdot)$ denote the $L_k$ $\varepsilon$-covering number and $\mathscr{N}_k(\varepsilon,\cdot,X^{1:n})$ denote the empirical $l_k$ $\varepsilon$-covering number with the empirical $l_k$ metric.
Then
\begin{equation}\label{eq:directPACbound}\begin{aligned}
  &\mathbb{P}\left(\sup_{G_t\in\mathcal{G}_t}\left|\mathrm{MMD}(G_{t\#}\hat{P}_t\|Q_t)-\widehat{\mathrm{MMD}}(G_{t\#}\hat{P}_t\|Q_t)\right|>\varepsilon+\frac{12B(\chi^2(Q_t\|P_t)+1)}{N_2}\right)\\
  \leq\,& 2\mathscr{N}_\infty\left(\frac{\varepsilon}{6},\mathcal{F}_D\right)\exp\left(-\frac{\gamma_tN_2(\mathbb{E}w_t)^2\varepsilon^2}{5760B^2\beta^2}\right)+2\exp\left(-\frac{N_2(\mathbb{E}w_t)^2}{2B^2}\right)\\
  &\qquad\qquad\qquad\qquad\qquad+8\mathscr{N}_\infty\left(\frac{\varepsilon}{32},\mathcal{F}_D\right)\mathbb{E}_{X\sim\hat P_t}\mathscr{N}_1\left(\frac{\varepsilon}{32\beta},\mathcal{G}_t,X^{1:N_1}\right)\exp\left(-\frac{N_1\varepsilon^2}{2048B^2}\right),\end{aligned}
  \end{equation}
  where $\chi^2(Q_t\|P_t)=\int(\frac{\text{d}Q_t}{\text{d}P_t}-1)^2\text{d}P_t$ is the $\chi^2$ divergence between $Q_t$ and $P_t$. As a consequence, for any $\delta>0$, there exist $N_1(\delta)$ and $N_2(\delta)$, such that when $N_1\geq N_1(\delta),N_2\geq N_2(\delta)$, with probability at least $1-\delta$, the following inequality holds simultaneously and uniformly for all $\hat Q_t=G_{t\#}\hat P_t$ over $G_t\in\mathcal{G}_{t}$:
  \begin{equation}\label{eq:PACbound}\begin{aligned}
    &\mathrm{MMD}(\hat G_{t\#}\hat P_t\|Q_t)\leq \inf_{G_t\in\mathcal{G}_t}\mathrm{MMD}(G_{t\#}\hat P_t\|Q_t)+\frac{24B(\chi^2(Q_t\|P_t)+1)}{N_2}\\
    &\quad+\left(\frac{128B\beta}{\sqrt{\gamma_t}(\EE w_t)}\sqrt{\frac{\log\mathscr{N}_\infty(\frac{\varepsilon(\delta)}{6},\mathcal{F}_D)}{N_2}}+128B\sqrt{\frac{\log\mathscr{N}_\infty(\frac{\varepsilon(\delta)}{32},\mathcal{F}_D)+\log\EE\mathscr{N}_1(\frac{\varepsilon(\delta)}{32\beta},\mathcal{G}_t)}{N_1}}\right)\sqrt{\log\frac{3}{\delta}}.\end{aligned}
    \end{equation}
\end{theorem}

The term $\mathrm{MMD}(\hat G_{t\#}\hat P_t\|Q_t)$ represents the loss of our empirical estimate defined in \eqref{eq:obj} and the term $\inf_{G_t\in\mathcal{G}_t}\mathrm{MMD}(G_{t\#}\hat P_t\|Q_t)$ represents the approximation error. In fact, $\mathrm{MMD}(\hat G_{t\#}\hat P_t\|Q_t)-\inf_{G_t\in\mathcal{G}_t}\mathrm{MMD}(G_{t\#}\hat P_t\|Q_t)$ is the statistical error. The other three terms appearing on the right-hand side of \eqref{eq:PACbound} correspond to the three terms in the decomposition \eqref{eq:decomposition}. For the first term, the $\chi^2$ divergence reflects the level of weight decay and a larger value means a severer decay and thus a slower convergence rate.  For the second term, the factor $\frac{\beta}{\sqrt{\gamma_t}}$ means that stronger log-concavity and Lipschitzity help to accelerate the convergence.
This is reasonable since (1) a stronger log-concavity in $P_t$ implies that this distribution is relatively easier to estimate by finite samples, and (2) a stronger Lipschitzity in $w_t$ implies that the importance sampling in the analysis step is more effective. The factor $\EE w_t$ also acts as an indicator of the weight decay.
For the third term, it is a standard rate with a parameter $\beta$. The law of convergence is the same as before: the larger value $\beta$ takes, the slower rate the model will converge at.

Concretely, we can consider the case that $Q_t$ and $\mathcal{F}_D$ both refer to the Sobolev space $W^{\alpha,2}(1)$.
Let $N_1$ go to infinity and take expectations on $X^{'1:N_2}$.
Then the term $\mathscr{N}_\infty(\mathcal{F}_D)$ can be replaced by its Dudley entropy integral, and the convergence rate would recover the vanilla GANs' rate $\displaystyle N_2^{-\frac{\alpha}{d}}\vee\frac{\log N_2}{\sqrt{N_2}}$ (Lemma 26 and Theorem 9 in \citet{liang2021well}).
This is because the bias of the SNIS estimator goes to zero as $N_2$ goes to infinity.
In our experiments, we always set a relatively large $N_2$, say 10000, to mitigate the negative effects caused by the SNIS bias and weight decay.

Up to now, we have shown the bound on the excess risk in the analysis step. Next, we generalize our result to the forecast step.
\begin{proposition}[An upper bound on the forecast error]\label{prop:forecastrisk}
  Assume Assumptions \ref{ass:1}--\ref{ass:4} hold. If $\mathcal{M}_t$ is $\gamma$-Lipschitz, then 
  \begin{equation}
    \mathrm{MMD}(\hat P_{t+1}\|P_{t+1})=\mathrm{MMD}(\mathcal{M}_{t\#}\hat Q_t*p_{\eta_{t+1}}\|\mathcal{M}_{t\#}Q_t*p_{\eta_{t+1}})\leq\beta\gamma W_2(\hat Q_t\|Q_t),
  \end{equation} where ``$*$" denotes the convolution operator.
\end{proposition}

A straightforward condition to ensure that $\mathcal{M}_t$ is $\gamma$-Lipschitz is to set $\mathcal{M}_t\equiv\text{id}$, as considered in \citet{lange2024bellman}.
Here, we extend this condition to a more general case under the Lipschitz continuity assumption. Notably, the MMD metric may be weaker than the $W_2$ metric, especially when the test function space is relatively small.
Therefore, with assumptions \ref{ass:1}-\ref{ass:5}, we cannot ensure the convergence of $W_2(\hat Q_t\|Q_t)$.
Nevertheless, we establish an alternative bound for our kernel ATPF, which is detailed below.

\begin{proposition}[An alternative upper bound on the forecast error]\label{prop:forecastriskkernel}
  Assume Assumptions \ref{ass:1}--\ref{ass:4} hold and we adopt the unit ball of the RKHS for $\mathcal{F}_D$ as described in Section \ref{subsec:ker}. If the contraction condition holds, i.e., $\|\mathcal{F}_D\circ\mathcal{M}_t\|_{\mathcal{H}}\leq \Omega\|\mathcal{F}_D\|_{\mathcal{H}}$ with an absolute constant $\Omega$, then
  \begin{equation}
    \mathrm{MMD}(\mathcal{M}_{t\#}\hat Q_t*p_{\eta_{t+1}}\|\mathcal{M}_{t\#}Q_t*p_{\eta_{t+1}})\leq\Omega\cdot\mathrm{MMD}(\hat Q_t\|Q_t)+2\beta\cdot\trace(\Var(\eta_{t+1}))^{\frac{1}{2}}.
  \end{equation}
\end{proposition}

The additional term $\trace(\Var(\eta_{t+1}))^{\frac{1}{2}}$ is due to the evolution noise, which cannot be neglected. Similar results can also be found in Proposition 2 of \citet{lange2024bellman}.

\subsection{Convergence Rate Analysis}
In this section, we derive upper bounds for the covering numbers involved in Theorem \ref{thm:simplefinalupperbound}, which allows us to establish convergence rates of our method.
When $\mathcal{F}_D$ denotes an RKHS, Lemma D.2 in \citet{yang2020function} provides two absolute constants $C_3$ and $C_4$ such that
\begin{equation}\label{eq:kernelcovering}\log\mathscr{N}_\infty\left(\varepsilon,\mathcal{F}_D\right)\leq \begin{cases}
  C_3\gamma[\log(B/\varepsilon)+C_4]&(\gamma\text{-finite spectrum})\\C_3[\log(B/\varepsilon)+C_4]^{1+1/\gamma}&(\gamma\text{-exponential decay})\\C_3(B/\varepsilon)^{2/[\gamma(1-2\tau)-1]}[\log(B/\varepsilon)+C_4]&(\gamma\text{-polynomial decay})
\end{cases}\end{equation}
where these three cases correspond to different eigenvalue decay rates.
In this paper, we mainly focus on the linear kernel $K(x,y)=x^Ty+c$, which corresponds to the first case with $\gamma=d$, and the Gaussian kernel $\displaystyle K(x,y)=\exp\left(-\frac{\|x-y\|_2^2}{2w^2}\right)$, which corresponds to the second case with $\gamma=\frac{1}{d}$.

We now focus on $\log\mathcal{N}_\infty(\varepsilon,\mathcal{G}_t)$.
Specifically, here we adopt the neural network parameterization method in \citet{suzuki2019adaptivity}, where $\mathcal{G}_t=\Phi(L,W,S,R):=\{(A^{(L)}\text{ReLU}(\cdot)+b^{(L)})\circ\dots\circ(A^{(1)}x+b^{(1)})| A^{(i)}\in\mathbb{R}^{W_i\times W_{i+1}},b^{(i)}\in \mathbb{R}^{W_{i+1}},\sum_{i=1}^L(\|A^{(i)}\|_0+\|b^{(i)}\|_0)\leq S,\max_i\|A^{(i)}\|_\infty\vee \|b^{(i)}\|_\infty\leq R\}$.
Note that results for this fully-connected neural networks can be easily translated into other architectures (e.g., \citet{petersen2018optimal,ramesh2022hierarchical}).
By Lemma 3 in \citet{suzuki2019adaptivity}, there exists an absolute constant $C_5$ such that
\begin{equation}\label{eq:suzuki}
\log\mathscr{N}_\infty(\varepsilon,\mathcal{G}_t,\|\cdot\|_{L_\infty[-M,M]^d})\leq 2C_5SL\log(\varepsilon^{-1}L\|W\|_\infty(R\vee 1)M).
\end{equation}
Note that this bound is only valid for distributions with compact support.
It is necessary to extend this result to sub-Gaussian cases since the noises involved in SSMs are usually not bounded but sub-Gaussian.
We summarize it in the following corollary.

\begin{corollary}[Bound on log-covering number in sub-Gaussian cases]\label{cor:log-covering}
  If we assume that for each $t$, all mentioned distributions are sub-Gaussian with uniform sub-Gaussian parameter $\sigma_t$, then there exists an absolute constant $C_7$ such that the following inequality holds with probability at least $\displaystyle 1-\exp\left(-\frac{N_1\varepsilon^2}{2048 B^2}\right)$:
  \begin{equation}\label{eq:sub-gaussin-covering}
    \log\mathscr{N}_1\left(\frac{\varepsilon}{32\beta},\mathcal{G}_t,X^{1:N_1}\right)
    \leq 2C_7SL\log\left(\beta L\|W\|_\infty(R\vee 1)\sigma_t\sqrt{\frac{\log(2dN_1)}{\varepsilon^2}+\frac{N_1}{2048B^2}}\right).
  \end{equation}
\end{corollary}

Note that the growth rate of the log-covering number is $O(\log(N_1))$.
Thus, the convergence rate in sub-Gaussian cases will be slower by a factor $\log(N_1)$ than that in bounded cases.

After analyzing the covering number of these two spaces, we can now derive the convergence rate for our proposed kernel ATPF methods with some standard neural network approximation theorems \citep{suzuki2019adaptivity,chakraborty2024statistical}.
Following these works, we consider the Besov space $B_{p,q}^s(\Omega)$ (see Definition 4 in \citet{suzuki2019adaptivity}) and assume the following assumption.

\begin{assumption} \label{ass:6} There exists an oracle $G'_t\in B_{p,q}^s(\Omega)$ such that $G'_{t\#}\hat P_t=Q_t$.
\end{assumption}

Since high-probability bounds can be directly derived by plugging these log-covering number bounds into \eqref{eq:PACbound}, in the following we focus more on the entropy integral, which further leads to the bound on $\mathbb{E}\mathrm{MMD}(\hat G_{t\#}\hat{P}_t\|Q_t)$.
We observe that the Gaussian kernel may lack scalability in high-dimensional settings, as its approximation error is significantly affected by the curse of dimensionality.
Specifically, we can establish the following convergence rate.

\begin{proposition}[Convergence rate of the Gaussian kernel MMD]\label{prop:gaussconvergence} Assume Assumptions \ref{ass:1}--\ref{ass:6} hold and that all mentioned distributions are sub-Gaussian. Let $s>\frac{d}{p}$. If we adopt the unit ball of the Gaussian RKHS for $\mathcal{F}_D$ and parameterize $\mathcal{G}_t$ by the NN set $\Phi(L,W,S,R)$, then the following convergence rate holds: \begin{equation}\label{eq:MMDGaussianrate}
  \EE\mathrm{MMD}(\hat G_{t\#} \hat P_t\|Q_t)\lesssim(\log N_1)^{\frac{s}{2}\vee 1}N_1^{-\frac{s}{2(s+d)}}+N_2^{-\frac{1}{2}}.
\end{equation}
\end{proposition}

Moreover, when the linear kernel is adopted, we have the following convergence rate for the posterior mean estimate.

\begin{proposition}[Convergence rate of the posterior mean]\label{prop:meanconvergence}
  Assume Assumptions \ref{ass:1}--\ref{ass:5} hold and that all mentioned distributions are sub-Gaussian. If we adopt the unit ball of the linear RKHS for $\mathcal{F}_D$ and parameterize $\mathcal{G}_t$ by the NN set $\Phi(L,W,S,R)$, the convergence rate of the posterior mean is\begin{equation}
    \|\EE_{X\sim\hat G_{t}\#\hat P_t}X-\EE_{Y\sim Q_t}Y\|_2\lesssim d\left(\frac{\log N_1}{\sqrt{N_1}}+\frac{1}{\sqrt{N_2}}\right)
  \end{equation}
  if $R\geq\|\EE_{Y\sim Q_t}Y\|_\infty$. Here, $\EE_{X\sim\hat G_{t}\#\hat P_t}X$ is the estimated posterior mean and $\EE_{Y\sim Q_t}Y$ is the true posterior mean.
\end{proposition}

\section{Experiments}
\label{sec:experi}
In this section, we compare our proposed ATPF with several data assimilation algorithms: EnKF \citep{katzfuss2020ensemble}, ETPF \citep{reich2013nonparametric}, NLEAF \citep{lei2011moment}, PF \citep{gordon1993novel}, MPF \citep{nakano2007merging} and PFGR \citep{xiong2006note}. Following \citet{lei2011moment}, we use NLEAF1 and NLEAF2 for Lorenz63 system and NLEAF1 and NLEAF1q for Lorenz96 system, where ``1'' stands for matching the first moment, ``2'' stands for matching the second moment and ``1q'' stands for quadratic regression.

We use the following settings for different approaches in our experiments.
For ATPF methods, we simply set a time-homogeneous regularization parameter $\lambda_t\equiv\lambda$ for each $t$.
For KATPF with Gaussian kernels, we set the bandwidth parameter $w$ to be the median of all $L_2$ distances between each pair of particles \citep{liu2016stein}.
For other existing methods, we just follow the standard settings used in previous works.
For PF series, we take the resampling step after each reweighting step. In addition, we adopt the inflation trick (see Appendix \ref{app:AppendixC} or \citet{lei2011moment}) for all filtering algorithms appearing in this section.
We use the likelihood of the data as a criterion to select the optimal inflation parameters.
All data sets and code are available at \url{https://github.com/WQgcx/ATPF.git}.

\subsection{Nonlinear Simulation}
In this section, we test the approximation accuracy of different methods on a nonlinear model where the EnKF method would fail to provide exact posterior estimate.
More specifically, we consider the following simple 1-dimension 2-stage state-space model:
\begin{equation}
    p(x_0)\sim\mathcal{N}(0,1),\quad p(x_1|x_0)\sim\mathcal{N}(x_0^2+\log(x_0^2+1),0.1^2),\quad p(y_t|x_t)\sim \mathcal{N}(x_t,1)\text{ for }t=0,1.
\end{equation}

Although the observation $y_t$ conditioned on the latent state $x_t$ follows a standard Gaussian distribution, the true posterior of $x_1$ is non-Gaussian due to the nonlinear evolution from $x_0$ to $x_1$. Mathematically speaking, this is because $p(x_1|y_0,y_1)\propto p(x_1|y_0)p(y_1|x_1)$; while $p(y_1|x_1)$ is Gaussian, $p(x_1|y_0)$ is not.
Now assume that we already have the observations $y_0=y_1=0$, our goal is to sample from the posterior $p(x_1|y_0,y_1)$ whose ground truth value can be obtained via numerical quadratures.

The regularization parameter $\lambda$ corresponding to the vanilla ATPF (i.e., use the alternative optimization algorithm), Gaussian KATPF and linear KATPF is set to be 0.001, 0.0 and 0.2, respectively. Experimental results are shown in Figure \ref{fig:simulation}.
We use 10000 particles to demonstrate the approximation accuracy of the fitted posteriors from different methods. The blue lines represent the density of the true posterior $p(x_1|y_0,y_1)$ and the orange histograms represent the estimated posterior densities from different methods.
\begin{figure}[t]
  \centering
  \includegraphics[width=0.85\textwidth]{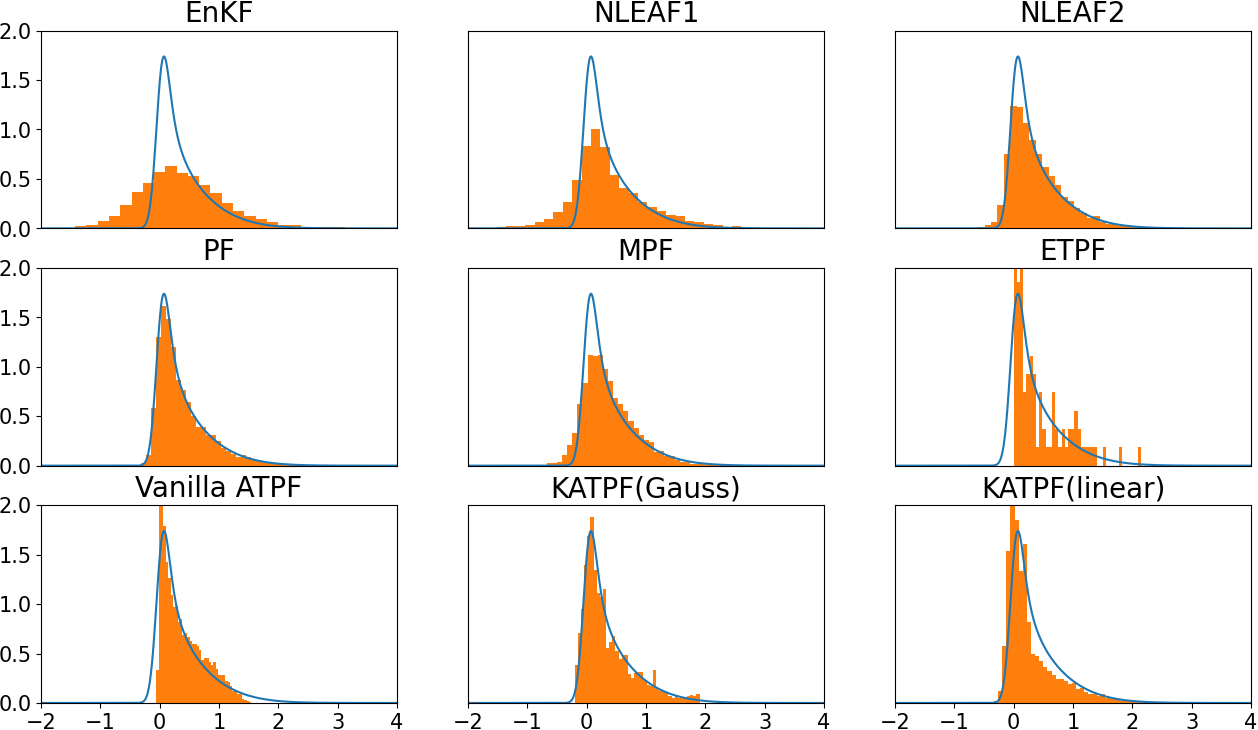}
  \caption{Experimental results of a simple two-stage state-space model.}
\label{fig:simulation}
\end{figure}

From Figure \ref{fig:simulation}, we see that the true posterior is not a symmetric Gaussian, but instead exhibits a right-skewed distribution.
Therefore, the core challenge in training and filtering is to accurately capture the shape of density around the peak.
We see that the EnKF (top left panel) is limited to a Gaussian approximation which fails to capture the skewness.
The NLEAF methods (top middle and top right panels) perform better, as they can capture the peak shape by matching either the first or second moments.
The vanilla PF (middle left panel) is the best estimator in this case, as it is asymptotically unbiased and well suits this low-dimensional scenario.
MPF (center panel) enhances particle diversity and stability, though it does so at the price of some statistical consistency by allowing more outliers.
ETPF (middle right panel), which relies on a linear programming algorithm, tends to produce some indented protuberances, possibly due to the lack of randomness in its design.
As for our proposed methods, they all capture the general shape of the target distribution well.
More specifically, the vanilla ATPF (bottom left panel) provide particles that concentrate on the mode, and the linear KATPF (bottom right) tends to underestimate the skewness of the density and would benefit from additional $W_2$ regularization to correct this.
In this low-dimensional case, the Gaussian kernel performs particularly well.
The corresponding RKHS manages to provide a rich set of test functions, leading to the best performance among all ATPF variants (bottom middle panel). 
\subsection{Filtering the Rainfall Data}
In this section, we consider the threshold observation model for rainfall (e.g., \citet{sanso1999venezuelan}).
The observations are modeled as $y_{t,l}=z_{t,l}^{\theta_t}1_{z_{t,l}>0}$ for $l=1,\dots,d$ and some $\theta_t>1$, and $z_t|x_t\sim\mathcal{N}(H_tx_t,\sigma_t^2)$.
Here, $z_t$ is an unknown auxiliary variable which connects the latent state $x_t$ and the actual observation $y_t$.
This kind of two-stage observation rule defines a rain-fall type distribution, i.e., a mixture of a positive right-skewed distribution, with a point mass at zero (e.g., \citet{sanso1999venezuelan,katzfuss2020ensemble}).
The parameter $\theta_t$, which controls the skewness, is either known or assigned with a pre-specified prior.
The whole rainfall model can be mathematically depicted as\begin{equation}\begin{aligned}
  x_t|x_{t-1}\sim\mathcal{N}(\mathcal{M}_{t-1}&(x_{t-1}),Q_t),\quad z_t|x_t\sim\mathcal{N}(H_tx_t,\sigma_t^2),\quad \theta_t\sim p(\theta_t),\\
  y_{t,l}|z_{t,l},&\theta_t= z_{t,l}^{\theta_t}1_{z_{t,l}>0},\text{ for }l=1,\dots,d.\\
\end{aligned}\end{equation}

Following \citet{katzfuss2020ensemble}, we focus on the analysis step, that is, how to transport the prior particles to the posterior ones. 
We first simulate $100$ true state vectors on a one-dimensional spatial domain $[1,100]$ according to a Gaussian Process (GP) prior with the following exponential kernel:\begin{equation}
  (x_1,\dots,x_d)\sim\mathcal{N}_d(\mu,\Sigma) \text{ with }\Sigma_{ij}=\alpha_1\exp\left(-\frac{|i-j|}{\alpha_2}\right),
\end{equation} where $\alpha_1$ is the scale parameter and $\alpha_2$ is the width parameter.
In our experiments, we set $d=100$, $\mu=(0,\dots,0)^T$, $\alpha_1=1$ and $\alpha_2=20$.
As for the rainfall model itself, we set $H_t\equiv I_{100\times 100}$ and $\sigma_t\equiv 0.4$.
And for simplicity, we set $\theta_t\equiv 3$, i.e., $p(\theta_t)=\delta(\theta_t-3)$.

As in \citet{katzfuss2020ensemble}, we truncate all particles in the posterior estimates at $0$ as negative rainfall values are not physically meaningful.
We compute RMSEs over the true rain and treat the estimate obtained by the first-order Langevin method \citep{welling2011bayesian} (running 10000 steps for 500 independent chains and collecting the final samples) as the exact posterior.
Experimental results are presented in Figure \ref{fig:rainfall}, with corresponding RMSE values shown in Table \ref{tb:rainfall}.
We use 100000 particles for the PF and 500 particles for the other filtering methods.
The 95\% credit intervals are constructed based on the ensemble.

\begin{table}[t]
  \centering
  \caption{Experimental results (RMSE) of the rainfall model.}
  \vspace{0.5em}
  \label{tb:rainfall}
  \resizebox{0.9\linewidth}{!}{
  \begin{tabular}{ccccccc}
  \toprule
    MCMC & EnKF  & NLEAF1 & NLEAF2 & PF & KATPF(Gauss) & KATPF(linear)\\
    \midrule
    0.1245 & 0.2422 & 0.1717& 0.1752 & 0.1433 & 0.1209 & 0.1394 \\
    \bottomrule
  \end{tabular}
  }
\end{table}

From Table \ref{tb:rainfall}, we see that our ATPF methods outperform their teacher (i.e., PF) with much fewer particles.
This is partly due to the explicit regularization of $W_2$ distance and implicit regularization of neural network parameterization and the corresponding stochastic optimization algorithms \citep{ma2018implicit,wei2020implicit}.
While the PF relies solely on a mixture of delta functions to approximate the density, our KATPF method introduces smoother approximations, akin to other kernel-based methods, which results in improved performance.
This trend is consistent across the subsequent experiments.

\begin{figure}[t]
  \centering
  \includegraphics[width=0.9\textwidth]{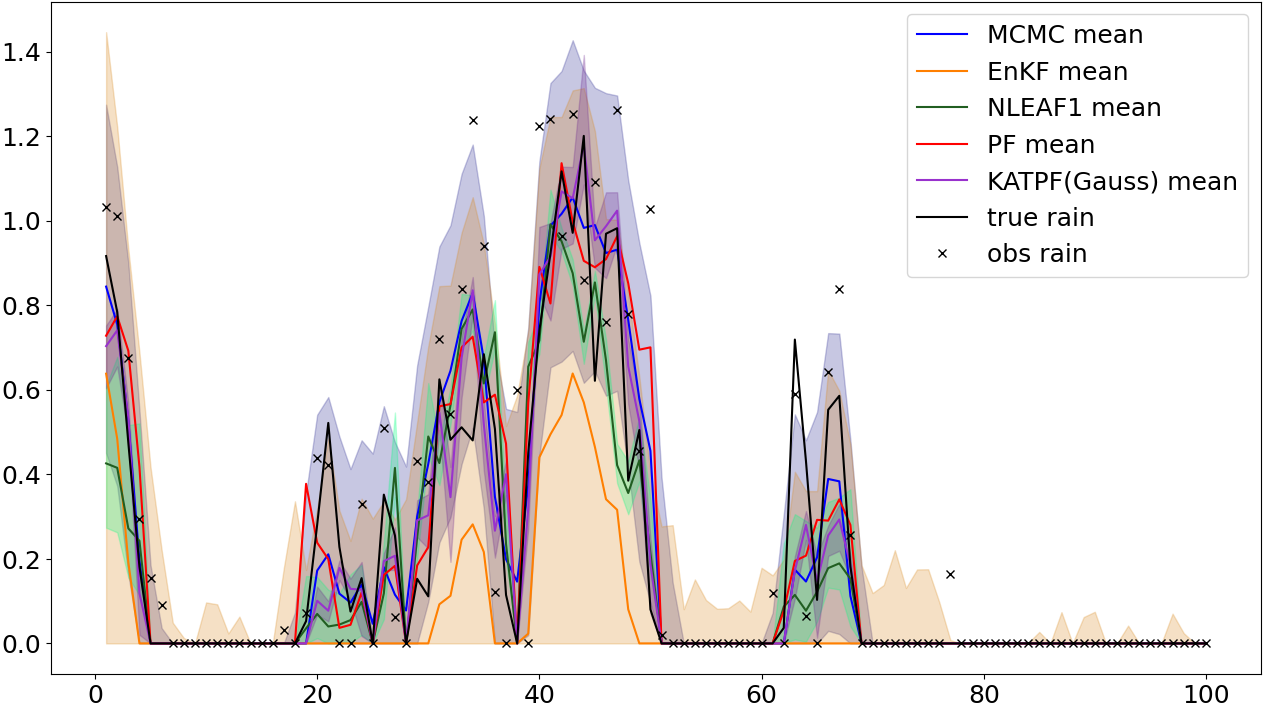}
  \caption{Experimental results of the rainfall model. The shaded areas represent the corresponding 95\% confidence intervals (CIs).}
  \label{fig:rainfall}
\end{figure}
From Figure \ref{fig:rainfall}, we observe that the EnKF tends to significantly underestimate the true rainfall while NLEAF-type methods are less robust to outliers, as evidenced by their narrower confidence intervals (CIs) that sometimes fail to cover the true data.
In fact, except for the EnKF, nearly all baseline filtering methods produce narrower and bumpier CIs than the exact posterior during the peak period, indicating that prediction with this rainfall model is challenging due to the strong randomness assumption and truncated observations.
On the other hand, our proposed methods strike a good balance between the accuracy of point estimates and the uncertainty in the posterior approximations, thanks to the adversarial framework.
This results in more accurate estimates than the EnKF (see Table \ref{tb:rainfall}) and more reasonable CIs than NLEAF methods (see Figure \ref{fig:rainfall}).
Notably, the Gaussian KATPF outperforms all the other filtering algorithms in terms of RMSE, even though it was not trained by minimizing the square loss.

\subsection{The Lorenz63 Model}
\label{subsec:lorenz63}
As a chaotic system, the Lorenz63 model is described by the following set of ordinary differential equations:
\begin{equation}
  \begin{aligned}
    \frac{\text{d}x_\tau(1)}{\text{d}\tau}&=-\sigma x_{\tau}(1)+\sigma x_{\tau}(2),\\
    \frac{\text{d}x_\tau(2)}{\text{d}\tau}&=-x_\tau(1)x_\tau(3)+\rho x_\tau(1)-x_\tau(2),\\
    \frac{\text{d}x_\tau(3)}{\text{d}\tau}&=x_\tau(1)x_\tau(2)-\beta x_\tau(3),
  \end{aligned}
\end{equation}
where $\tau$ is the continuous time index, $x_\tau$ is a three-dimensional latent state variable representing a simplified model of a heated fluid's flow, with parameters $\beta=\frac{8}{3}$, $\rho=28$, and $\sigma=10$. Note that this model does not include evolution noise, which implies $\mathbb{P}(\eta_t=0)=1$ in \eqref{eq:SSM}.

In the simulation, we use standard Euler discretization with a small time step $\delta t=0.02$ to numerically integrate this continuous-time Lorenz63 model.
The observation rule is set as $y_t|x_t\sim\mathcal{N}(x_t,I_{3\times 3})$ and data are generated by observing $T=100$ or $T=2000$ samples from the system.
The goal is to evaluate the performance of different methods in both short-term and long-term evolution scenarios.
In addition, the time interval between consecutive observations is set to be $\Delta t=0.02$, $0.04$, $0.06$ or $0.08$, corresponding to 1, 2, 3 or 4 Euler discretization steps between $x_{t-1}$ and $x_t$ (or $y_{t-1}$ and $y_t$), respectively.
Here, larger observation gaps typically introduce more nonlinearities in the evolution.
We use 1000 particles in all experiments.

\begin{table}[t]\footnotesize
  \centering\caption{Experimental results (RMSE) of the Lorenz63 system.}
  \vspace{0.5em}
  \label{tb:lorenz63}
  \begin{tabular}{cccccccccc}
  \toprule
  &EnKF&NLEAF1&NLEAF2&PF&MPF&PFGR&\makecell{KATPF\\(Gauss)}&\makecell{KATPF\\(linear)}&\makecell{KATPF-JO\\(linear)}\\\midrule
    \makecell{$T = 100$\\$\Delta t=0.02$}& 0.1924&0.1862&0.1907&0.1719&0.2295&0.2110&0.1946&0.1891&0.1398 \\\midrule
    \makecell{$T = 2000$\\$\Delta t=0.02$}&0.1142&0.1133&0.1063&0.1131&0.1287&0.1476&0.1111&0.1087&0.0843\\\midrule
    \makecell{$T = 2000$\\$\Delta t=0.04$}&0.1870&0.1716&0.1475&0.1465&0.1970&0.1874&0.1608&0.1544&0.1157\\\midrule
    \makecell{$T = 2000$\\$\Delta t=0.06$}&0.1824&0.1769&0.1445&0.1503&0.1968&0.1759&0.1769&0.1506&0.0973\\\midrule
    \makecell{$T = 2000$\\$\Delta t=0.08$}&0.2293&0.2140&0.1751&0.1825&0.2343&0.2081&0.1858&0.1829&0.1305\\
    \bottomrule
  \end{tabular}
\end{table}

Table \ref{tb:lorenz63} shows the RMSE between the true values and the ensemble mean \citep{lei2011moment}.
We see that our KATPF methods consistently outperform the other baselines, particularly when the nonlinearity increases.
Specifically, when the observation interval $\Delta t$ is small (i.e., the first two rows), the evolution model is nearly linear due to the Euler discretization, and all methods perform similarly.
As $\Delta t$ increases, the filtering task becomes more challenging due to the increasing nonlinearity of the evolution model, leading to a general decline in performance across all methods.
However, KATPF methods exhibit a slower rate of performance degradation compared to the others.
From Figure \ref{fig:lorenz63step4}, we can see that as non-linearity increases, the particle trajectories become more chaotic, making it harder to accurately identify them.
Despite this, our method maintains strong performance.
The results also reveal that the joint optimization procedure (ATPF-JO) outperforms the sequential training approach.
See more results on the effect of the OT regularization in Appendix \ref{app:AppendixB} (Table \ref{tb:lorenz63_wo_OT}).
\subsection{The Lorenz96 Model}
\label{subsec:lorenz96}
The Lorenz96 model is another commonly used testbed for evaluating data assimilation methods. The 40-dimensional state vector evolves according to the following system of ordinary differential equations:\begin{equation}\frac{\text{d}x_\tau(i)}{\text{d}\tau}=[x_\tau(i+1)-x_\tau(i-2)]x_\tau(i-1)-x_\tau(i)+8\text{  for }i=1,\dots,40,\end{equation} where periodic boundary conditions are applied: $x_\tau(0)=x_\tau(40),x_\tau(-1)=x_\tau(39)$ and $x_\tau(41)=x_\tau(1)$.
This model is widely used for studying chaotic behavior and turbulence in atmospheric and climate systems.
Also, this model does not include evolution noise.

Due to the complexity of the Lorenz96 model, we set the discretization step size to $\delta=0.01$ in the simulation.
To increase the difficulty of fitting, we adopt the hard case setup from \citet{lei2011moment} where the observation rule is given by $y_t=H_tx_t+\varepsilon_t$ \eqref{eq:lineobs} with $H_t=(e_1,e_3,\dots,e_{39})_{40\times 20}$ and $\varepsilon_t\sim\mathcal{N}(0,I_{20\times 20})$.
Here, $e_i=(0,\dots,0,\underbrace{1}_{i\text{-th}},0,\dots,0)$ is the standard basis vector with all zero entries except the $i$th position. This means that only 20 out of the 40 coordinates of the state vector are observed at each time step, with each observation contaminated by standard Gaussian noise.
We generate data by simulating $T=100$ or $T=1000$ samples from this model with $\Delta t=0.01$ (i.e., only one discretization step between $x_{t-1}$ and $x_t$). To further enhance nonlinearity, we also consider a case with $\delta=0.02$, $\Delta t=0.2$ and $T=10$, similar to \citet{katzfuss2020ensemble}.
To better handle the high-dimensional nature of the system, we use localization for the NLEAF methods (denoted as ``w/ loc.'') and the tapering method for the EnKF (denoted as ``w/tap.'').
The extra hyperparameters involved are chosen according to \citet{lei2011moment}.
For our proposed ATPF methods, we use MPF as a guidance mechanism, which helps alleviate weight degeneracy in this high-dimensional setting.
All filtering algorithms adopt the inflation trick introduced in Appendix \ref{app:AppendixC}.
For the EnKF, NLEAF series, and our proposed KATPF, we used 500 particles, which we found to be sufficient in our experiments.
For PF series, we use 30000 particles (and $10^6$ particles for the third case) to ensure numerical stability and statistical consistency.

\begin{table}[t]
\footnotesize
  \centering
  \caption{Experimental results (RMSE) of the Lorenz96 model.}
  \label{tb:lorenz96}
  \vspace{0.5em}
  \begin{tabular}{cccccccccc}\toprule
    &\makecell{EnKF\\(w/tap.)}&\makecell{NLEAF1\\(w/loc.)}&\makecell{NLEAF1q\\(w/loc.)}&PF&MPF&\makecell{KATPF\\(Gauss)}&\makecell{KATPF\\(linear)}&\makecell{KATPF-JO\\(linear)}\\\midrule
    \makecell{$T = 100$\\$\Delta t= 0.01$}&0.4212&0.3895&0.4620&0.5478&0.4457&0.4250&0.4133&0.3927\\\midrule
    \makecell{$T=1000$\\$\Delta t=0.01$}&0.2125&0.2133&0.2374&0.2733&0.2170&0.2155&0.1920&0.1600\\\midrule
    \makecell{$T=10$\\$\Delta t=0.2$}&0.8320&0.8036&0.9043&0.8498&0.8079&0.8649&0.7860&0.7723\\
    \bottomrule
  \end{tabular}
\end{table}

Table \ref{tb:lorenz96} shows the RMSEs of different methods across three different scenarios.
We see that KATPF methods perform well in all three cases, even outperforming their teacher method (i.e., MPF) with significantly fewer particles.
Although NLEAF1 performs well in short-term settings (when $T$ is small), KATPF methods exhibit stronger capacities, especially in long-term settings.
Again, with joint optimization, the KATPF-JO performs the best among all KATPF variants.
\section{Conclusion and Discussion}
\label{sec:conc}
In this paper, we propose the adversarial transform particle filter (ATPF), a new filtering method for data assimilation in nonlinear and non-Gaussian models.
ATPF combines the particle diversity of EnKFs with the statistical consistency of PFs within an adversarial framework that minimizes the maximum mean discrepancy (MMD) to match the posterior distribution.
We provide theoretical guarantees for ATPF's consistency and convergence, and validate its accuracy and effectiveness through extensive experiments.
Additionally, optimal transport based regularization can be employed to enhance performance.
We also show that existing techniques, such as the inflation trick, can be seamlessly incorporated into our method.

Despite its advantages, fitting high dimensional models directly with ATPF remains challenging as it relies on particle filters (or their variants) to estimate posterior expectations, and the importance sampling required for this becomes difficult in high-dimensional settings (see the property of the $\chi^2$ divergence in Theorem \ref{thm:SNISbias}).
However, unlike PFs that use importance sampling to estimate the whole posterior distribution, ATPF only uses importance sampling to estimate the posterior expectation of test functions.
This is similar to NLEAF \citep{lei2011moment}, which lends itself more easily to localization for dimension reduction.
Another remedy would be to introduce a sequence of auxiliary distributions that interpolate between the prior and the posterior (e.g., the annealing path \citep{geyer1991markov}).
By learning transformations to map between successive distributions using ATPF, we could more easily achieve better approximations, as successive distributions can be made closer to each other.
We leave the investigation of more efficient ATPF methods for high-dimensional models to future work.

\section*{Acknowledgments}

This research was supported by National Natural Science Foundation of China grants 12292980.
The research of Cheng Zhang was supported in part by National Natural Science Foundation of China (grant no. 12201014 and grant no. 12292983), the National Engineering Laboratory for Big Data Analysis and Applications, the Key Laboratory of Mathematics and Its Applications (LMAM), and the Fundamental Research Funds for the Central Universities, Peking University.
The research of Wei Lin was supported in part by National Natural Science Foundation of China (grant no. 12171012 and grant no. 12292981).

\newpage
\bibliographystyle{nameyear}
\bibliography{Bibliography-MM-MC}

\begin{thebibliography}{61}
\providecommand{\natexlab}[1]{#1}
\providecommand{\url}[1]{\texttt{#1}}
\expandafter\ifx\csname urlstyle\endcsname\relax
  \providecommand{\doi}[1]{doi: #1}\else
  \providecommand{\doi}{doi: \begingroup \urlstyle{rm}\Url}\fi

\bibitem[Agapiou et~al.(2017)Agapiou, Papaspiliopoulos, Sanz-Alonso, and
  Stuart]{agapiou2017importance}
Sergios Agapiou, Omiros Papaspiliopoulos, Daniel Sanz-Alonso, and Andrew~M
  Stuart.
\newblock Importance sampling: Intrinsic dimension and computational cost.
\newblock \emph{Statistical Science}, pp.\  405--431, 2017.

\bibitem[Akritas(1986)]{akritas1986empirical}
Michael~G Akritas.
\newblock Empirical processes associated with v-statistics and a class of
  estimators under random censoring.
\newblock \emph{The Annals of Statistics}, pp.\  619--637, 1986.

\bibitem[Alvarez et~al.(2012)Alvarez, Rosasco, Lawrence,
  et~al.]{alvarez2012kernels}
Mauricio~A Alvarez, Lorenzo Rosasco, Neil~D Lawrence, et~al.
\newblock Kernels for vector-valued functions: A review.
\newblock \emph{Foundations and Trends{\textregistered} in Machine Learning},
  4\penalty0 (3):\penalty0 195--266, 2012.

\bibitem[Anderson(2001)]{anderson2001ensemble}
Jeffrey~L Anderson.
\newblock An ensemble adjustment kalman filter for data assimilation.
\newblock \emph{Monthly weather review}, 129\penalty0 (12):\penalty0
  2884--2903, 2001.

\bibitem[Anderson(2007)]{anderson2007adaptive}
Jeffrey~L Anderson.
\newblock An adaptive covariance inflation error correction algorithm for
  ensemble filters.
\newblock \emph{Tellus A: Dynamic meteorology and oceanography}, 59\penalty0
  (2):\penalty0 210--224, 2007.

\bibitem[Anderson(2010)]{anderson2010nonGaussian}
Jeffrey~L Anderson.
\newblock A non-gaussian ensemble filter update for data assimilation.
\newblock \emph{Monthly Weather Review}, 138:\penalty0 4186--4198, 2010.

\bibitem[Arjovsky et~al.(2017)Arjovsky, Chintala, and
  Bottou]{arjovsky2017wasserstein}
Martin Arjovsky, Soumith Chintala, and L{\'e}on Bottou.
\newblock Wasserstein generative adversarial networks.
\newblock In \emph{International conference on machine learning}, pp.\
  214--223. PMLR, 2017.

\bibitem[Arulampalam et~al.(2002)Arulampalam, Maskell, Gordon, and
  Clapp]{arulampalam2002tutorial}
M~Sanjeev Arulampalam, Simon Maskell, Neil Gordon, and Tim Clapp.
\newblock A tutorial on particle filters for online nonlinear/non-gaussian
  bayesian tracking.
\newblock \emph{IEEE Transactions on signal processing}, 50\penalty0
  (2):\penalty0 174--188, 2002.

\bibitem[Bertino et~al.(2003)Bertino, Evensen, and
  Wackernagel]{bertino2003sequential}
Laurent Bertino, Geir Evensen, and Hans Wackernagel.
\newblock Sequential data assimilation techniques in oceanography.
\newblock \emph{International Statistical Review}, 71\penalty0 (2):\penalty0
  223--241, 2003.

\bibitem[Bishop et~al.(2001)Bishop, Etherton, and Majumdar]{bishop2001adaptive}
Craig~H Bishop, Brian~J Etherton, and Sharanya~J Majumdar.
\newblock Adaptive sampling with the ensemble transform kalman filter. part i:
  Theoretical aspects.
\newblock \emph{Monthly weather review}, 129\penalty0 (3):\penalty0 420--436,
  2001.

\bibitem[Carlin et~al.(1992)Carlin, Polson, and Stoffer]{carlin1992}
B.~P. Carlin, N.~G. Polson, and D.~S. Stoffer.
\newblock A monte carlo approach to nonnormal and nonlinear state-space
  modeling.
\newblock \emph{Journal of the American Statistical Association}, 87:\penalty0
  493--500, 1992.

\bibitem[Chakraborty \& Bartlett(2024)Chakraborty and
  Bartlett]{chakraborty2024statistical}
Saptarshi Chakraborty and Peter~L Bartlett.
\newblock On the statistical properties of generative adversarial models for
  low intrinsic data dimension.
\newblock \emph{arXiv preprint arXiv:2401.15801}, 2024.

\bibitem[Evensen(1994)]{evensen1994sequential}
Geir Evensen.
\newblock Sequential data assimilation with a nonlinear quasi-geostrophic model
  using monte carlo methods to forecast error statistics.
\newblock \emph{Journal of Geophysical Research: Oceans}, 99\penalty0
  (C5):\penalty0 10143--10162, 1994.

\bibitem[Evensen(2003)]{evensen2003ensemble}
Geir Evensen.
\newblock The ensemble kalman filter: Theoretical formulation and practical
  implementation.
\newblock \emph{Ocean dynamics}, 53:\penalty0 343--367, 2003.

\bibitem[Fan et~al.(2021)Fan, Zhang, Taghvaei, and Chen]{fan2021variational}
Jiaojiao Fan, Qinsheng Zhang, Amirhossein Taghvaei, and Yongxin Chen.
\newblock Variational wasserstein gradient flow.
\newblock \emph{arXiv preprint arXiv:2112.02424}, 2021.

\bibitem[Fletcher(2022)]{fletcher2022data}
Steven~J Fletcher.
\newblock \emph{Data assimilation for the geosciences: From theory to
  application}.
\newblock Elsevier, 2022.

\bibitem[Furrer et~al.(2006)Furrer, Genton, and Nychka]{furrer2006covariance}
Reinhard Furrer, Marc~G Genton, and Douglas Nychka.
\newblock Covariance tapering for interpolation of large spatial datasets.
\newblock \emph{Journal of Computational and Graphical Statistics}, 15\penalty0
  (3):\penalty0 502--523, 2006.

\bibitem[Gamerman(1998)]{Gamerman1998}
D.~Gamerman.
\newblock Markov chain monte carlo for dynamics generalized linear models.
\newblock \emph{Biometrika}, 85:\penalty0 215--227, 1998.

\bibitem[Geyer(1991)]{geyer1991markov}
C.~J. Geyer.
\newblock Markov chain {M}onte {C}arlo maximum likelihood.
\newblock \emph{Interface Proceedings}, 1991.

\bibitem[Goodfellow et~al.(2014)Goodfellow, Pouget-Abadie, Mirza, Xu,
  Warde-Farley, Ozair, Courville, and Bengio]{goodfellow2014generative}
Ian Goodfellow, Jean Pouget-Abadie, Mehdi Mirza, Bing Xu, David Warde-Farley,
  Sherjil Ozair, Aaron Courville, and Yoshua Bengio.
\newblock Generative adversarial nets.
\newblock \emph{Advances in neural information processing systems}, 27, 2014.

\bibitem[Gordon et~al.(1993)Gordon, Salmond, and Smith]{gordon1993novel}
Neil~J Gordon, David~J Salmond, and Adrian~FM Smith.
\newblock Novel approach to nonlinear/non-gaussian bayesian state estimation.
\newblock In \emph{IEE proceedings F (radar and signal processing)}, volume
  140, pp.\  107--113. IET, 1993.

\bibitem[Gy{\"o}rfi et~al.(2002)Gy{\"o}rfi, Kohler, Krzyżak, and
  Walk]{distfreebook}
L{\'a}szl{\'o} Gy{\"o}rfi, Michael Kohler, Adam Krzyżak, and Harro Walk.
\newblock \emph{A Distribution-Free Theory of Nonparametric Regression}.
\newblock 2002.
\newblock URL \url{https://api.semanticscholar.org/CorpusID:43315484}.

\bibitem[Hesterberg(1988)]{hesterberg1988advances}
Timothy~Classen Hesterberg.
\newblock \emph{Advances in importance sampling}.
\newblock Stanford University, 1988.

\bibitem[Hunt et~al.(2007)Hunt, Kostelich, and Szunyogh]{hunt2007efficient}
Brian~R Hunt, Eric~J Kostelich, and Istvan Szunyogh.
\newblock Efficient data assimilation for spatiotemporal chaos: A local
  ensemble transform kalman filter.
\newblock \emph{Physica D: Nonlinear Phenomena}, 230\penalty0 (1-2):\penalty0
  112--126, 2007.

\bibitem[Jordan et~al.(1998)Jordan, Kinderlehrer, and
  Otto]{jordan1998variational}
Richard Jordan, David Kinderlehrer, and Felix Otto.
\newblock The variational formulation of the fokker--planck equation.
\newblock \emph{SIAM journal on mathematical analysis}, 29\penalty0
  (1):\penalty0 1--17, 1998.

\bibitem[Kalman(1960)]{kalman1960}
R.~E. Kalman.
\newblock A new approach to linear filtering and prediction problems.
\newblock \emph{Journal of Basic Engineering}, 82:\penalty0 35--45, 1960.

\bibitem[Katzfuss et~al.(2016)Katzfuss, Stroud, and
  Wikle]{katzfuss2016understanding}
Matthias Katzfuss, Jonathan~R Stroud, and Christopher~K Wikle.
\newblock Understanding the ensemble kalman filter.
\newblock \emph{The American Statistician}, 70\penalty0 (4):\penalty0 350--357,
  2016.

\bibitem[Katzfuss et~al.(2020)Katzfuss, Stroud, and
  Wikle]{katzfuss2020ensemble}
Matthias Katzfuss, Jonathan~R Stroud, and Christopher~K Wikle.
\newblock Ensemble kalman methods for high-dimensional hierarchical dynamic
  space-time models.
\newblock \emph{Journal of the American Statistical Association}, 115\penalty0
  (530):\penalty0 866--885, 2020.

\bibitem[Kim et~al.(2018)Kim, Bang, et~al.]{kim2018introduction}
Youngjoo Kim, Hyochoong Bang, et~al.
\newblock Introduction to kalman filter and its applications.
\newblock \emph{Introduction and Implementations of the Kalman Filter},
  1:\penalty0 1--16, 2018.

\bibitem[Kong(1992)]{kong1992note}
Augustine Kong.
\newblock A note on importance sampling using standardized weights.
\newblock \emph{University of Chicago, Dept. of Statistics, Tech. Rep},
  348:\penalty0 14, 1992.

\bibitem[K{\"u}nsch(2005)]{kunsch2005recursive}
Hans~R K{\"u}nsch.
\newblock Recursive monte carlo filters: algorithms and theoretical analysis.
\newblock \emph{Annals of Statistics}, pp.\  1983--2021, 2005.

\bibitem[Lange(2024)]{lange2024bellman}
Rutger-Jan Lange.
\newblock Bellman filtering and smoothing for state--space models.
\newblock \emph{Journal of Econometrics}, 238\penalty0 (2):\penalty0 105632,
  2024.

\bibitem[Lei \& Bickel(2011)Lei and Bickel]{lei2011moment}
Jing Lei and Peter Bickel.
\newblock A moment matching ensemble filter for nonlinear non-gaussian data
  assimilation.
\newblock \emph{Monthly Weather Review}, 139\penalty0 (12):\penalty0
  3964--3973, 2011.

\bibitem[Liang(2021)]{liang2021well}
Tengyuan Liang.
\newblock How well generative adversarial networks learn distributions.
\newblock \emph{Journal of Machine Learning Research}, 22\penalty0
  (228):\penalty0 1--41, 2021.

\bibitem[Liu \& Wang(2016)Liu and Wang]{liu2016stein}
Qiang Liu and Dilin Wang.
\newblock Stein variational gradient descent: A general purpose bayesian
  inference algorithm.
\newblock \emph{Advances in neural information processing systems}, 29, 2016.

\bibitem[Ma et~al.(2018)Ma, Wang, Chi, and Chen]{ma2018implicit}
Cong Ma, Kaizheng Wang, Yuejie Chi, and Yuxin Chen.
\newblock Implicit regularization in nonconvex statistical estimation: Gradient
  descent converges linearly for phase retrieval and matrix completion.
\newblock In \emph{International Conference on Machine Learning}, pp.\
  3345--3354. PMLR, 2018.

\bibitem[Magazinov \& Peled(2022)Magazinov and
  Peled]{magazinov2022concentration}
Alexander Magazinov and Ron Peled.
\newblock Concentration inequalities for log-concave distributions with
  applications to random surface fluctuations.
\newblock \emph{The Annals of Probability}, 50\penalty0 (2):\penalty0 735--770,
  2022.

\bibitem[Montemerlo et~al.(2003)Montemerlo, Thrun, Koller, Wegbreit,
  et~al.]{montemerlo2003fastslam}
Michael Montemerlo, Sebastian Thrun, Daphne Koller, Ben Wegbreit, et~al.
\newblock Fastslam 2.0: An improved particle filtering algorithm for
  simultaneous localization and mapping that provably converges.
\newblock In \emph{IJCAI}, volume~3, pp.\  1151--1156, 2003.

\bibitem[Murphy \& Russell(2001)Murphy and Russell]{murphy2001rao}
Kevin Murphy and Stuart Russell.
\newblock Rao-blackwellised particle filtering for dynamic bayesian networks.
\newblock In \emph{Sequential Monte Carlo methods in practice}, pp.\  499--515.
  Springer, 2001.

\bibitem[Musso et~al.(2001)Musso, Oudjane, and Le~Gland]{musso2001improving}
Christian Musso, Nadia Oudjane, and Francois Le~Gland.
\newblock Improving regularised particle filters.
\newblock In \emph{Sequential Monte Carlo methods in practice}, pp.\  247--271.
  Springer, 2001.

\bibitem[Nakano et~al.(2007)Nakano, Ueno, and Higuchi]{nakano2007merging}
Shinya Nakano, Genta Ueno, and Tomoyuki Higuchi.
\newblock Merging particle filter for sequential data assimilation.
\newblock \emph{Nonlinear Processes in Geophysics}, 14\penalty0 (4):\penalty0
  395--408, 2007.

\bibitem[Oko et~al.(2023)Oko, Akiyama, and Suzuki]{oko2023diffusion}
Kazusato Oko, Shunta Akiyama, and Taiji Suzuki.
\newblock Diffusion models are minimax optimal distribution estimators.
\newblock In \emph{International Conference on Machine Learning}, pp.\
  26517--26582. PMLR, 2023.

\bibitem[Petersen \& Voigtlaender(2018)Petersen and
  Voigtlaender]{petersen2018optimal}
Philipp Petersen and Felix Voigtlaender.
\newblock Optimal approximation of piecewise smooth functions using deep relu
  neural networks.
\newblock \emph{Neural Networks}, 108:\penalty0 296--330, 2018.

\bibitem[Poterjoy(2016)]{poterjoy2016localized}
Jonathan Poterjoy.
\newblock A localized particle filter for high-dimensional nonlinear systems.
\newblock \emph{Monthly Weather Review}, 144\penalty0 (1):\penalty0 59--76,
  2016.

\bibitem[Rabier(2005)]{rabier2005overview}
Florence Rabier.
\newblock Overview of global data assimilation developments in numerical
  weather-prediction centres.
\newblock \emph{Quarterly Journal of the Royal Meteorological Society: A
  journal of the atmospheric sciences, applied meteorology and physical
  oceanography}, 131\penalty0 (613):\penalty0 3215--3233, 2005.

\bibitem[Ramesh et~al.(2022)Ramesh, Dhariwal, Nichol, Chu, and
  Chen]{ramesh2022hierarchical}
Aditya Ramesh, Prafulla Dhariwal, Alex Nichol, Casey Chu, and Mark Chen.
\newblock Hierarchical text-conditional image generation with clip latents.
\newblock \emph{arXiv preprint arXiv:2204.06125}, 1\penalty0 (2):\penalty0 3,
  2022.

\bibitem[Reich(2013)]{reich2013nonparametric}
Sebastian Reich.
\newblock A nonparametric ensemble transform method for bayesian inference.
\newblock \emph{SIAM Journal on Scientific Computing}, 35\penalty0
  (4):\penalty0 A2013--A2024, 2013.

\bibitem[Sanso \& Guenni(1999)Sanso and Guenni]{sanso1999venezuelan}
Bruno Sanso and Lelys Guenni.
\newblock Venezuelan rainfall data analysed by using a bayesian space--time
  model.
\newblock \emph{Journal of the Royal Statistical Society: Series C (Applied
  Statistics)}, 48\penalty0 (3):\penalty0 345--362, 1999.

\bibitem[Saumard \& Wellner(2014)Saumard and Wellner]{saumard2014log}
Adrien Saumard and Jon~A Wellner.
\newblock Log-concavity and strong log-concavity: a review.
\newblock \emph{Statistics surveys}, 8:\penalty0 45, 2014.

\bibitem[Serfling(2009)]{serfling2009approximation}
Robert~J Serfling.
\newblock \emph{Approximation theorems of mathematical statistics}.
\newblock John Wiley \& Sons, 2009.

\bibitem[Shephard \& Pitt(1997)Shephard and Pitt]{Shephard1997}
N.~Shephard and M.~Pitt.
\newblock Likelihood analysis of non-gaussian measurement time series.
\newblock \emph{Biometrika}, 84:\penalty0 653--667, 1997.

\bibitem[Smola et~al.(2006)Smola, Gretton, and Borgwardt]{smola2006maximum}
Alexander~J Smola, A~Gretton, and K~Borgwardt.
\newblock Maximum mean discrepancy.
\newblock In \emph{13th international conference, ICONIP}, pp.\  3--6, 2006.

\bibitem[Snyder et~al.(2008)Snyder, Bengtsson, Bickel, and
  Anderson]{snyder2008obstacles}
Chris Snyder, Thomas Bengtsson, Peter Bickel, and Jeff Anderson.
\newblock Obstacles to high-dimensional particle filtering.
\newblock \emph{Monthly Weather Review}, 136\penalty0 (12):\penalty0
  4629--4640, 2008.

\bibitem[Sun et~al.(2024)Sun, Wang, Zheng, and Chen]{sun2024high}
Hao-Xuan Sun, Shouxia Wang, Xiaogu Zheng, and Song~Xi Chen.
\newblock High-dimensional ensemble kalman filter with localization, inflation,
  and iterative updates.
\newblock \emph{Quarterly Journal of the Royal Meteorological Society},
  150\penalty0 (765):\penalty0 4870--4884, 2024.

\bibitem[Suzuki(2019)]{suzuki2019adaptivity}
Taiji Suzuki.
\newblock Adaptivity of deep relu network for learning in besov and mixed
  smooth besov spaces: optimal rate and curse of dimensionality.
\newblock In \emph{International Conference on Learning Representations},
  volume~7, 2019.

\bibitem[Wainwright(2019)]{hdpbook}
Martin Wainwright.
\newblock \emph{High-Dimensional Statistics: A Non-Asymptotic Viewpoint}.
\newblock 2019.
\newblock ISBN 9781108498029.
\newblock \doi{10.1017/9781108627771}.

\bibitem[Wei et~al.(2020)Wei, Kakade, and Ma]{wei2020implicit}
Colin Wei, Sham Kakade, and Tengyu Ma.
\newblock The implicit and explicit regularization effects of dropout.
\newblock In \emph{International Conference on Machine Learning}, pp.\
  10181--10192. PMLR, 2020.

\bibitem[Welling \& Teh(2011)Welling and Teh]{welling2011bayesian}
Max Welling and Yee~W Teh.
\newblock Bayesian learning via stochastic gradient langevin dynamics.
\newblock In \emph{Proceedings of the 28th international conference on machine
  learning (ICML-11)}, pp.\  681--688. Citeseer, 2011.

\bibitem[Xiong et~al.(2006)Xiong, Navon, and Uzunoglu]{xiong2006note}
Xiaozhen Xiong, Ionel~Michael Navon, and Bahri Uzunoglu.
\newblock A note on the particle filter with posterior gaussian resampling.
\newblock \emph{Tellus A: Dynamic Meteorology and Oceanography}, 58\penalty0
  (4):\penalty0 456--460, 2006.

\bibitem[Xu et~al.(2024)Xu, Cheng, and Xie]{xu2024normalizing}
Chen Xu, Xiuyuan Cheng, and Yao Xie.
\newblock Normalizing flow neural networks by jko scheme.
\newblock \emph{Advances in Neural Information Processing Systems}, 36, 2024.

\bibitem[Yang et~al.(2020)Yang, Jin, Wang, Wang, and Jordan]{yang2020function}
Zhuoran Yang, Chi Jin, Zhaoran Wang, Mengdi Wang, and Michael~I Jordan.
\newblock On function approximation in reinforcement learning: Optimism in the
  face of large state spaces.
\newblock \emph{arXiv preprint arXiv:2011.04622}, 2020.

\end{thebibliography}

\newpage
\appendix
\begin{center}
{\Large\bf Supplementary Material for ``Adversarial Transform Particle Filters''}
\end{center}
\section{Proofs}
\label{app:AppendixA}
\begin{definition}[RKHS, \citet{alvarez2012kernels}]
 A bivariate, symmetric and positive definite function $K(\cdot,\cdot)$ defined on the domain $E$ is called a reproducing kernel for a Hilbert space $\mathscr{H}$ if it satisfies the following two properties:
\begin{enumerate}[(a)]
\item $\forall x\in E$, $K(\cdot,x)\in\mathscr{H}$;
\item (Reproducing property) $\forall f\in\mathscr{H}$ and $x\in E$, $f(x)=\langle f,K(\cdot,x)\rangle_{\mathscr{H}}$.
\end{enumerate}
When such a reproducing kernel exists, $\mathscr{H}$ is called a reproducing kernel Hilbert space (RKHS).
\end{definition}

We recall that in Theorem \ref{thm:boundrisk} of Section \ref{sec:theory}, the excess risk is decomposed into three terms. In this section, we will detailedly show the corresponding bounds for these three terms and finally put them together to prove our main result (i.e., Theorem \ref{thm:finalupperbound}).

For the first term, we directly cite the following existing theorem:
\begin{theorem}[Upper bound on the SNIS bias \citep{agapiou2017importance}]\label{thm:SNISbias}

  Assume Assumptions \ref{ass:1}--\ref{ass:3} hold. Let $\chi^2(Q_t\|P_t)$ denote the $\chi^2$ divergence between $Q_t$ and $P_t$. Then\begin{equation}
    \sup_{\|f\|_\infty\leq B}\left|\EE_{x'\sim P_t}\left[\sum_{j=1}^{N_2}\frac{w_t(x'_j)}{\sum_{j=1}^{N_2}w_t(x'_j)}f(x'_j)\right]-\EE_{y\sim {Q_t}}f(y)\right|\leq\frac{12B(\chi^2(Q_t\|P_t)+1)}{N_2}.
  \end{equation}
\end{theorem}

For the second term, we first need a lemma to ensure Lipschitzity of the SNIS statistics:

\begin{lemma}\label{lem:lip}
  Assume Assumptions \ref{ass:1}--\ref{ass:4} hold. Let $F_t(X)=\sum_{j=1}^{N_2}\frac{w_t(x_j)}{\sum_{j=1}^{N_2}w_t(x_j)}f(x_j)$ with $X=(x_1,\dots,x_{N_2})\in\mathbb{R}^{N_2\times d}$ and $x_{m+1}=(X_{md+1},\dots,X_{(m+1)d})$ for $m=0,1,\dots,d-1$. Under the concentration condition, say $\left|\frac{\sum_{j=1}^{N_2}w_t(x_j)}{N_2}-\EE w_t\right|\leq\frac{1}{2}\EE w_t$, $F_t(X)$ is $\sqrt{\frac{40}{N_2}}\frac{B\beta}{\EE w_t}$-Lipschitz.
\end{lemma}
\begin{proof}[Lemma \ref{lem:lip}]
  For $s=1,\dots,d$,
\begin{align*}
  \frac{\partial F_t}{\partial X_s}&=\frac{\left(\frac{\partial w_t(x_1)}{\partial x_{1,k}}f(x_1)+\frac{\partial f(x_1)}{\partial x_{1,k}}w_t(x_1)\right)\left(\sum_{j=1}^{N_2}w_t(x_j)\right)-\frac{\partial w_t(x_1)}{\partial x_{1,k}}\left(\sum_{j=1}^{N_2}w_t(x_j)f(x_j)\right)}{\left(\sum_{j=1}^{N_2}w_t(x_j)\right)^2}\\
  &=\frac{\partial f(x_1)}{\partial x_{1,k}}\frac{w_t(x_1)}{\sum_{j=1}^{N_2}w_t(x_j)}+\frac{\partial w_t(x_1)}{\partial x_{1,k}}\frac{\sum_{j=1}^{N_2}\left(w_t(x_j)(f(x_1)-f(x_j))\right)}{\left(\sum_{j=1}^{N_2}w_t(x_j)\right)^2}\\
  \Rightarrow\frac{\partial F_t}{\partial x_1}&=\frac{w_t(x_1)}{\sum_{j=1}^{N_2}w_t(x_j)}\nabla f(x_1)+\frac{\sum_{j=1}^{N_2}\left(w_t(x_j)(f(x_1)-f(x_j))\right)}{\left(\sum_{j=1}^{N_2}w_t(x_j)\right)^2}\nabla w_t\\
  \Rightarrow \left\|\frac{\partial F_t}{\partial x_1}\right\|_2^2&\leq 2\frac{w_t^2(x_1)}{\left(\sum_{j=1}^{N_2}w_t(x_j)\right)^2}\|\nabla f(x)\|_2^2+2\left(\frac{\sum_{j=1}^{N_2}\left(w_t(x_j)(f(x_1)-f(x_j))\right)}{\left(\sum_{j=1}^{N_2}w_t(x_j)\right)^2}\right)^2\|\nabla w_t\|_2^2\\
  &\leq 2B^2\beta^2\frac{1}{\left(\sum_{j=1}^{N_2}w_t(x_j)\right)^2}+8B^2\beta^2\frac{1}{\left(\sum_{j=1}^{N_2}w_t(x_j)\right)^2}=\frac{10B^2\beta^2}{\left(\sum_{j=1}^{N_2}w_t(x_j)\right)^2}.
\end{align*}

Similarly, we can derive the same upper bounds for $x_2,\dots,x_{N_2}$. Therefore,\[\left\|\frac{\partial F_t}{\partial X}\right\|_2^2=\sum_{j=1}^{N_2}\left\|\frac{\partial F_t}{\partial x_j}\right\|_2^2\leq \frac{10B^2\beta^2N_2}{\left(\sum_{j=1}^{N_2}w_t(x_j)\right)^2}=\frac{10B^2\beta^2}{N_2}\left(\frac{N_2}{\sum_{j=1}^{N_2}w_t(x_j)}\right)^2.\]

Under the concentration condition, we can conclude that \[\left\|\frac{\partial F_t}{\partial X}\right\|_2^2\leq\frac{10B^2\beta^2}{N_2}\left(\frac{N_2}{\sum_{j=1}^{N_2}w_t(x_j)}\right)^2\leq \frac{40B^2\beta^2}{N_2(\EE w_t)^2}\Rightarrow \|F_t\|_{\text{Lip}}\leq \sqrt{\frac{40}{N_2}}\frac{B\beta}{\EE w_t}.\]
\end{proof}

Then we can show the following bound for the weighted empirical process:

\begin{theorem}[Upper bound 1 with smoothness and convexity conditions]\label{thm:bound2} Assume Assumptions \ref{ass:1}--\ref{ass:5} hold. Let $\varepsilon>0$. Let $\mathscr{N}_\infty(\varepsilon,\mathcal{F}_D)$ denote the $L_\infty$ $\varepsilon$-covering number of $\mathcal{F}_D$. Then\begin{equation}\label{eq:bound2}\begin{aligned}&\mathbb{P}\left(\sup_{f\in\mathcal{F}_D}\left|\sum_{j=1}^{N_2}\frac{w_t(x'_j)}{\sum_{j=1}^{N_2}w_t(x'_j)}f(x'_j)-\EE_{x'\sim P_t}\left[\sum_{j=1}^{N_2}\frac{w_t(x'_j)}{\sum_{j=1}^{N_2}w_t(x'_j)}f(x'_j)\right]\right|>\varepsilon\right)\\ \leq\,& 2\mathscr{N}_\infty
  \left(\frac{\varepsilon}{3},\mathcal{F}_D\right)\exp\left(-\frac{\gamma_tN_2(\EE w_t)^2\varepsilon^2}{1440B^2\beta^2}\right)+2\exp\left(-\frac{N_2(\EE w_t)^2}{2B^2}\right).\end{aligned}\end{equation}
\end{theorem}

\begin{proof}[Theorem \ref{thm:bound2}]
 The first step is a straightforward application of the Theorem 3.16 in \citet{hdpbook}. Assumption 5 imposes the strong log-concavity with parameter $\gamma_t$, and Lemma \ref{lem:lip} indicates the Lipschitzity with parameter $L=\sqrt{\frac{40}{N_2}}\frac{B\beta}{\mathbb{E}w_t}$. Then\[\footnotesize\begin{aligned}&\mathbb{P}\left(\left|\sum_{j=1}^{N_2}\frac{w_t(x'_j)}{\sum_{j=1}^{N_2}w_t(x'_j)}f(x'_j)-\EE_{x'\sim P_t}\left[\sum_{j=1}^{N_2}\frac{w_t(x'_j)}{\sum_{j=1}^{N_2}w_t(x'_j)}f(x'_j)\right]\right|>\frac{\varepsilon}{3}\,\,\,\Bigg|\,\underbrace{\left|\frac{\sum_{j=1}^{N_2}w_t(x_j)}{N_2}-\EE w_t\right|\leq\frac{1}{2}\EE w_t}_{\text{under this concentration condition}}\right)\\ \leq\,\,& 2\exp\left(-\frac{\gamma_tN_2(\EE w_t)^2\varepsilon^2}{1440B^2\beta^2}\right).\end{aligned}\]
 
When $N_2$ is relatively large, this concentration condition holds with probability approaching $1$ by the large deviation theories, i.e.
\[\mathbb{P}\left(\left|\frac{\sum_{j=1}^{N_2}w_t(x_j)}{N_2}-\EE w_t\right|\leq \frac{\EE w_t}{2}\right)>1-2\exp\left(-\frac{N_2(\EE w_t)^2}{2B^2}\right).\]
 
 Now induce a $L_\infty$ $\displaystyle\frac{\varepsilon}{3}$-covering $\mathcal{F}_{D,\frac{\varepsilon}{3}}$ of $\mathcal{F}_D$. Therefore, for every $f\in\mathcal{F}_D$, there exists at least an $f_0\in\mathcal{F}_{D,\frac{\varepsilon}{3}}$ such that $\|f-f_0\|_\infty<\frac{\varepsilon}{3}$. In this way,\[\footnotesize\begin{aligned}
    &\left\{\sup_{f\in\mathcal{F}_D}\left|\sum_{j=1}^{N_2}\frac{w_t(x'_j)}{\sum_{j=1}^{N_2}w_t(x'_j)}f(x'_j)-\EE_{x'\sim P_t}\left[\sum_{j=1}^{N_2}\frac{w_t(x'_j)}{\sum_{j=1}^{N_2}w_t(x'_j)}f(x'_j)\right]\right|>\varepsilon\right\}\\
    =\,&\left\{\exists f\in\mathcal{F}_D:\left|\sum_{j=1}^{N_2}\frac{w_t(x'_j)}{\sum_{j=1}^{N_2}w_t(x'_j)}f(x'_j)-\EE_{x'\sim P_t}\left[\sum_{j=1}^{N_2}\frac{w_t(x'_j)}{\sum_{j=1}^{N_2}w_t(x'_j)}f(x'_j)\right]\right|>\varepsilon\right\}\\
    \subset\,&\Bigg\{\exists f\in\mathcal{F}_D,f_0\in\mathcal{F}_{D,\frac{\varepsilon}{3}}:\left|\sum_{j=1}^{N_2}\frac{w_t(x'_j)}{\sum_{j=1}^{N_2}w_t(x'_j)}(f-f_0)(x'_j)\right|\\
    &\qquad\qquad\qquad\qquad\quad+\left|\sum_{j=1}^{N_2}\frac{w_t(x'_j)}{\sum_{j=1}^{N_2}w_t(x'_j)}f_0(x'_j)-\EE_{x'\sim P_t}\left[\sum_{j=1}^{N_2}\frac{w_t(x'_j)}{\sum_{j=1}^{N_2}w_t(x'_j)}f_0(x'_j)\right]\right|
    \\
    &\qquad\qquad\qquad\qquad\quad+\left|\EE_{x'\sim P_t}\left[\sum_{j=1}^{N_2}\frac{w_t(x'_j)}{\sum_{j=1}^{N_2}w_t(x'_j)}(f_0-f)(x'_j)\right]\right|>\varepsilon\Bigg\}\\
    \subset\,&\left\{\exists f\in\mathcal{F}_D,f_0\in\mathcal{F}_{D,\frac{\varepsilon}{3}}:\frac{\varepsilon}{3}+\left|\sum_{j=1}^{N_2}\frac{w_t(x'_j)}{\sum_{j=1}^{N_2}w_t(x'_j)}f_0(x'_j)-\EE_{x'\sim P_t}\left[\sum_{j=1}^{N_2}\frac{w_t(x'_j)}{\sum_{j=1}^{N_2}w_t(x'_j)}f_0(x'_j)\right]\right|+\frac{\varepsilon}{3}>\varepsilon\right\}\\
    =\,&\left\{\exists f_0\in\mathcal{F}_{D,\frac{\varepsilon}{3}}:\left|\sum_{j=1}^{N_2}\frac{w_t(x'_j)}{\sum_{j=1}^{N_2}w_t(x'_j)}f_0(x'_j)-\EE_{x'\sim P_t}\left[\sum_{j=1}^{N_2}\frac{w_t(x'_j)}{\sum_{j=1}^{N_2}w_t(x'_j)}f_0(x'_j)\right]\right|>\frac{\varepsilon}{3}\right\}.
  \end{aligned}\] That's to say,
  \[\footnotesize\begin{aligned}
    &\mathbb{P}\left(\sup_{f\in\mathcal{F}_D}\left|\sum_{j=1}^{N_2}\frac{w_t(x'_j)}{\sum_{j=1}^{N_2}w_t(x'_j)}f(x'_j)-\EE_{x'\sim P_t}\left[\sum_{j=1}^{N_2}\frac{w_t(x'_j)}{\sum_{j=1}^{N_2}w_t(x'_j)}f(x'_j)\right]\right|>\varepsilon\right)\\
    \leq\,&\mathbb{P}\left(\sup_{f\in\mathcal{F}_{D,\frac{\varepsilon}{3}}}\left|\sum_{j=1}^{N_2}\frac{w_t(x'_j)}{\sum_{j=1}^{N_2}w_t(x'_j)}f(x'_j)-\EE_{x'\sim P_t}\left[\sum_{j=1}^{N_2}\frac{w_t(x'_j)}{\sum_{j=1}^{N_2}w_t(x'_j)}f(x'_j)\right]\right|>\frac{\varepsilon}{3}\right)\\
    \leq\,&\mathbb{P}\left(\sup_{f\in\mathcal{F}_{D,\frac{\varepsilon}{3}}}\left|\sum_{j=1}^{N_2}\frac{w_t(x'_j)}{\sum_{j=1}^{N_2}w_t(x'_j)}f(x'_j)-\EE_{x'\sim P_t}\left[\sum_{j=1}^{N_2}\frac{w_t(x'_j)}{\sum_{j=1}^{N_2}w_t(x'_j)}f(x'_j)\right]\right|>\frac{\varepsilon}{3},\left|\frac{\sum_{j=1}^{N_2}w_t(x_j)}{N_2}-\EE w_t\right|\leq\frac{1}{2}\EE w_t\right)\\
    &\qquad\qquad\qquad\qquad\qquad\qquad\qquad\qquad\qquad\qquad\qquad\qquad\qquad\qquad+\mathbb{P}\left(\left|\frac{\sum_{j=1}^{N_2}w_t(x_j)}{N_2}-\EE w_t\right|>\frac{1}{2}\EE w_t\right)\\
    \leq\,&\underbrace{\mathbb{P}\left(\sup_{f\in\mathcal{F}_{D,\frac{\varepsilon}{3}}}\left|\sum_{j=1}^{N_2}\frac{w_t(x'_j)}{\sum_{j=1}^{N_2}w_t(x'_j)}f(x'_j)-\EE_{x'\sim P_t}\left[\sum_{j=1}^{N_2}\frac{w_t(x'_j)}{\sum_{j=1}^{N_2}w_t(x'_j)}f(x'_j)\right]\right|>\frac{\varepsilon}{3}\,\,\,\Bigg|\,\left|\frac{\sum_{j=1}^{N_2}w_t(x_j)}{N_2}-\EE w_t\right|\leq\frac{1}{2}\EE w_t\right)}_{\text{conditional probability, }\mathbb{P}(A|B)\geq\mathbb{P}(A|B)\mathbb{P}(B)=\mathbb{P}(A,B)}\\
    &\qquad\qquad\qquad\qquad\qquad\qquad\qquad\qquad\qquad\qquad\qquad\qquad\qquad\qquad+\mathbb{P}\left(\left|\frac{\sum_{j=1}^{N_2}w_t(x_j)}{N_2}-\EE w_t\right|>\frac{1}{2}\EE w_t\right)\\
    \leq\,&2\mathscr{N}_\infty
    \left(\frac{\varepsilon}{3},\mathcal{F}_D\right)\exp\left(-\frac{\gamma_tN_2(\EE w_t)^2\varepsilon^2}{1440B^2\beta^2}\right)+2\exp\left(-\frac{N_2(\EE w_t)^2}{2B^2}\right).
  \end{aligned}\] where the last inequality holds by the union bound and Hoeffding's inequality.
\end{proof}

We can also replace assumption 5 with the following assumption, which imposes convexity on the discriminator space $\mathcal{F}_D$, not on the distribution $P_t$:

\begin{assumption}[Convexity]\label{ass:7} All test functions in $\mathcal{F}_D$ are convex.
\end{assumption}

\begin{theorem}[Upper bound 2 with smoothness and convexity conditions]\label{thm:bound_2}
  Assume Assumptions \ref{ass:1}--\ref{ass:4} and \ref{ass:7} hold. Then \eqref{eq:bound2} also holds by substituting $\gamma_t$ with 2. What's more, this substitution holds for all following equations including the term $\gamma_t$.
\end{theorem}

\begin{proof}[Theorem \ref{thm:bound_2}]
  Use the Theorem 3.24 in \citet{hdpbook} to replace the Theorem 3.16 in the first step of the proof in Theorem \ref{thm:bound2}. Then the remaining argument is completely similar.
\end{proof}

Theorems \ref{thm:bound2} and \ref{thm:bound_2} are very standard results for the empirical process theories, except for an extra convergence rate of the concentration condition. Without assumptions 4 and 5, we can also show the following convergence rate:

\begin{theorem}[Upper bound without smoothness and convexity conditions]\label{thm:w/oconvex}
  Assume Assumptions \ref{ass:1}--\ref{ass:3} hold. Other notations keep the same as before. Then\begin{equation}\begin{aligned}&\mathbb{P}\left(\sup_{f\in\mathcal{F}_D}\left|\sum_{j=1}^{N_2}\frac{w_t(x'_j)}{\sum_{j=1}^{N_2}w_t(x'_j)}f(x'_j)-\EE_{x'\sim P_t}\left[\sum_{j=1}^{N_2}\frac{w_t(x'_j)}{\sum_{j=1}^{N_2}w_t(x'_j)}f(x'_j)\right]\right|>\varepsilon\right)\\ \leq\,& \mathscr{N}_\infty
    \left(\frac{\varepsilon}{3},\mathcal{F}_D\right)\frac{36B^2(\chi^2(Q_t\|P_t)+1)}{N_2\varepsilon^2}.\end{aligned}\end{equation}
\end{theorem}

\begin{proof}[Theorem \ref{thm:w/oconvex}]
  By Theorem 2.1 in \citet{agapiou2017importance},
 \[\sup_{\|f\|_\infty\leq B}\EE\left(\sum_{j=1}^{N_2}\frac{w_t(x'_j)}{\sum_{j=1}^{N_2}w_t(x'_j)}f(x'_j)-\EE_{y\sim {Q_t}}f(y)\right)^2\leq\frac{4B^2(\chi^2(Q_t\|P_t)+1)}{N_2}.\]

  Thus, for any $f\in\mathcal{F}$, by Markov's inequality, \[\mathbb{P}\left(\left|\sum_{j=1}^{N_2}\frac{w_t(x'_j)}{\sum_{j=1}^{N_2}w_t(x'_j)}f(x'_j)-\EE_{y\sim {Q_t}}f(y)\right|>\frac{\varepsilon}{3}\right)\leq \frac{36B^2(\chi^2(Q_t\|P_t)+1)}{N_2\varepsilon^2}.\]

  Then by inducing an $L_\infty$ $\displaystyle\frac{\varepsilon}{3}$-covering of $\mathcal{F}_D$ we can get the desired result.
\end{proof}

We can find this rate is much slower than that in \eqref{eq:bound2}, as mentioned in Section \ref{sec:theory}, but the convergence still holds.

The third term represents the empirical process of our target function, so by the uniform law of large numbers (e.g., Theorem 9.1 in \citet{distfreebook}),
\begin{equation}\small\label{eq:epforh}\mathbb{P}\left(\sup_{h\in\mathcal{F}_D\circ\mathcal{G}_t}\left|\frac{1}{N_1}\sum_{i=1}^{N_1}h(x_i)-\EE_{x\sim \hat P_t}h(x)\right|>\varepsilon\right)\leq 8\EE_{X\sim\hat P_t}\mathscr{N}_1\left(\frac{\varepsilon}{8},\mathcal{F}_D\circ\mathcal{G}_t,X^{1:N_1}\right)\exp\left(-\frac{N_1\varepsilon^2}{512B^2}\right),\end{equation} where $\mathscr{N}_1(\frac{\varepsilon}{8},\mathcal{F}_D\circ\mathcal{G}_t,X^{1:N_1})$ is the empirical $l_1$ $\frac{\varepsilon}{8}$-covering number for the function space $\mathcal{F}_D\circ\mathcal{G}_t:=\{f\circ g:f\in\mathcal{F}_D,g\in\mathcal{G}_t\}$ with the empirical $l_1$ metric $\|f\|_1^{N_1}:=\frac{1}{N_1}\sum_{j=1}^{N_1}|f(x_j)|$.

Now we are all set for the proofs of the final bound in Theorem \ref{thm:simplefinalupperbound}, and the following Theorem \ref{thm:finalupperbound} is a complete and detailed version of Theorem \ref{thm:simplefinalupperbound}.

\begin{theorem}[Upper bound on the excess risk]\label{thm:finalupperbound} Assume Assumptions \ref{ass:1}--\ref{ass:5} hold. Let $\varepsilon>0$ and $\chi^2(Q_t\|P_t)$ denote the $\chi^2$ divergence between $Q_t$ and $P_t$. Let $\mathscr{N}_k(\varepsilon,\cdot)$ denote the $L_k$ $\varepsilon$-covering number of a Banach space and $\mathscr{N}_k(\varepsilon,\cdot,X^{1:n})$ denote the empirical $l_k$ $\varepsilon$-covering number with the empirical $l_k$ metric. Other notations remain the same. Then\begin{equation}\label{eq:supbound}\begin{aligned}
  &\mathbb{P}\left(\sup_{G_t\in\mathcal{G}_t}\left|\mathrm{MMD}(G_{t\#}\hat{P}_t\|Q_t)-\widehat{\mathrm{MMD}}(G_{t\#}\hat{P}_t\|Q_t)\right|>\varepsilon+\frac{12B(\chi^2(Q_t\|P_t)+1)}{N_2}\right)\\
  \leq\,& 2\mathscr{N}_\infty\left(\frac{\varepsilon}{6},\mathcal{F}_D\right)\exp\left(-\frac{\gamma_tN_2(\mathbb{E}w_t)^2\varepsilon^2}{5760B^2\beta^2}\right)+2\exp\left(-\frac{N_2(\mathbb{E}w_t)^2}{2B^2}\right)\\
  &\qquad\qquad\qquad\qquad\qquad+8\mathbb{E}_{X\sim\hat P_t}\mathscr{N}_1\left(\frac{\varepsilon}{16},\mathcal{F}_D\circ\mathcal{G}_t,X^{1:N_1}\right)\exp\left(-\frac{N_1\varepsilon^2}{2048B^2}\right)\\
  \leq\,& 2\mathscr{N}_\infty\left(\frac{\varepsilon}{6},\mathcal{F}_D\right)\exp\left(-\frac{\gamma_tN_2(\mathbb{E}w_t)^2\varepsilon^2}{5760B^2\beta^2}\right)+2\exp\left(-\frac{N_2(\mathbb{E}w_t)^2}{2B^2}\right)\\
  &\qquad\qquad\qquad\qquad\qquad+8\mathscr{N}_\infty\left(\frac{\varepsilon}{32},\mathcal{F}_D\right)\mathbb{E}_{X\sim\hat P_t}\mathscr{N}_1\left(\frac{\varepsilon}{32\beta},\mathcal{G}_t,X^{1:N_1}\right)\exp\left(-\frac{N_1\varepsilon^2}{2048B^2}\right).\end{aligned}
\end{equation}
  In other words, for any $\delta>0$, with probability at least $1-\delta$, the following inequality holds simultaneously and uniformly for all $\hat Q_t=G_{t\#}\hat P_t\in\mathcal{G}_{t\#}\hat P_t$ when $N_1\geq N_1(\delta)$ and $N_2\geq N_2(\delta)$:
  \begin{equation}\footnotesize\begin{aligned}&\mathrm{MMD}(\hat Q_t\|Q_t)\leq \widehat{\mathrm{MMD}}(\hat Q_t\|Q_t)\\ &+\frac{12B(\chi^2(Q_t\|P_t)+1)}{N_2}+\left(\frac{64B\beta}{\sqrt{\gamma_t}(\EE w_t)}\sqrt{\frac{\log\mathcal{N}_\infty(\frac{\varepsilon(\delta)}{6},\mathcal{F}_D)}{N_2}}+64B\sqrt{\frac{\log\mathcal{N}_\infty(\frac{\varepsilon(\delta)}{32},\mathcal{F}_D)+\log\mathbb{E}\mathcal{N}_1(\frac{\varepsilon(\delta)}{32\beta},\mathcal{G}_t,X^{1:N_1})}{N_1}}\right)\sqrt{\log\frac{3}{\delta}}.\end{aligned}\end{equation}
     This further implies\begin{equation}\footnotesize\begin{aligned}
    &\underbrace{\mathrm{MMD}(\hat G_{t\#}\hat P_t\|Q_t)}_{\text{empirical estimate }}\leq\underbrace{\inf_{G_t\in\mathcal{G}_t}\mathrm{MMD}(G_{t\#}\hat P_t\|Q_t)}_{\text{the best estimate}}\\
    &\qquad\qquad+\frac{24B(\chi^2(Q_t\|P_t)+1)}{N_2}+\left(\frac{128B\beta}{\sqrt{\gamma_t}(\EE w_t)}\sqrt{\frac{\log\mathscr{N}_\infty(\mathcal{F}_D)}{N_2}}+128B\sqrt{\frac{\log\mathscr{N}_\infty(\mathcal{F}_D)+\log\EE\mathscr{N}_1^\beta(\mathcal{G}_t)}{N_1}}\right)\sqrt{\log\frac{3}{\delta}},\end{aligned}
  \end{equation}where $\mathscr{N}_\infty(\mathcal{F}_D)$ and $\mathscr{N}_1^\beta(\mathcal{G}_t)$ are short for the corresponding terms.
\end{theorem}

\begin{proof}[Theorem \ref{thm:finalupperbound}]
  The first inequality in \eqref{eq:supbound} holds by combining Theorem \ref{thm:boundrisk}, Theorem \ref{thm:bound2}, and \eqref{eq:epforh} together:
\begin{align*}
    &\mathbb{P}\left(\sup_{G_t\in\mathcal{G}_t}\left|\mathrm{MMD}(G_{t\#}\hat{P}_t\|Q_t)-\widehat{\mathrm{MMD}}(G_{t\#}\hat{P}_t\|Q_t)\right|>\varepsilon+\frac{12B(\chi^2(Q_t\|P_t)+1)}{N_2}\right)\\
    \leq\,&\mathbb{P}\Bigg(\sup_{f\in\mathcal{F}_D}\left|\EE_{y\sim {Q_t}}f(y)-\EE_{x'\sim P_t}\left[\sum_{j=1}^{N_2}\frac{w_t(x'_j)}{\sum_{j=1}^{N_2}w_t(x'_j)}f(x'_j)\right]\right|\\ 
    &\qquad\qquad+\sup_{f\in\mathcal{F}_D}\left|\EE_{x'\sim P_t}\left[\sum_{j=1}^{N_2}\frac{w_t(x'_j)}{\sum_{j=1}^{N_2}w_t(x'_j)}f(x'_j)\right]-\sum_{j=1}^{N_2}\frac{w_t(x'_j)}{\sum_{j=1}^{N_2}w_t(x'_j)}f(x'_j)\right|\\
    &\qquad\qquad+\sup_{h\in\mathcal{F}_D\circ\mathcal{G}_t}\left|\frac{1}{N_1}\sum_{i=1}^{N_1}h(x_i)-\EE_{x\sim \hat P_t}h(x)\right|>\frac{12B(\chi^2(Q_t\|P_t)+1)}{N_2}+\frac{\varepsilon}{2}+\frac{\varepsilon}{2}\Bigg)\\
    \leq\,&\mathbb{P}\left(\sup_{f\in\mathcal{F}_D}\left|\EE_{x'\sim P_t}\left[\sum_{j=1}^{N_2}\frac{w_t(x'_j)}{\sum_{j=1}^{N_2}w_t(x'_j)}f(x'_j)\right]-\sum_{j=1}^{N_2}\frac{w_t(x'_j)}{\sum_{j=1}^{N_2}w_t(x'_j)}f(x'_j)\right|>\frac{12B(\chi^2(Q_t\|P_t)+1)}{N_2}\right)\\
    +\,&\mathbb{P}\left(\sup_{f\in\mathcal{F}_D}\left|\EE_{x'\sim P_t}\left[\sum_{j=1}^{N_2}\frac{w_t(x'_j)}{\sum_{j=1}^{N_2}w_t(x'_j)}f(x'_j)\right]-\sum_{j=1}^{N_2}\frac{w_t(x'_j)}{\sum_{j=1}^{N_2}w_t(x'_j)}f(x'_j)\right|>\frac{\varepsilon}{2}\right)\\
    +\,&\mathbb{P}\left(\sup_{h\in\mathcal{F}_D\circ\mathcal{G}_t}\left|\frac{1}{N_1}\sum_{i=1}^{N_1}h(x_i)-\EE_{x\sim \hat P_t}h(x)\right|>\frac{\varepsilon}{2}\right)\\
    \leq\,&0+2\mathscr{N}_\infty\left(\frac{\varepsilon}{6},\mathcal{F}_D\right)\exp\left(-\frac{\gamma_tN_2(\mathbb{E}w_t)^2\varepsilon^2}{5760B^2\beta^2}\right)+2\exp\left(-\frac{N_2(\mathbb{E}w_t)^2}{2B^2}\right)\\
    +\,& 8\mathbb{E}_{X\sim\hat P_t}\mathscr{N}_1\left(\frac{\varepsilon}{16},\mathcal{F}_D\circ\mathcal{G}_t,X^{1:N_1}\right)\exp\left(-\frac{N_1\varepsilon^2}{2048B^2}\right).
  \end{align*}

  For the second inequality, we can induce an $L_\infty$ $\displaystyle \frac{\varepsilon}{32}$-covering of $\mathcal{F}_D$ and an empirical $L_1$ $\displaystyle\frac{\varepsilon}{32\beta}$-$\|\cdot\|_2$ covering of $\mathcal{G}_t$. For every $f\in\mathcal{F}_D$ and $g\in\mathcal{G}_t$, there exist an $\tilde{f}\in\mathcal{F}_{\frac{\varepsilon}{32}}$ and a $\tilde{g}\in\mathcal{G}_{t,\frac{\varepsilon}{32\beta}}$ such that $\displaystyle\|f-\tilde{f}\|_\infty<\frac{\varepsilon}{32}$ and $\displaystyle\frac{1}{N_1}\sum_{i=1}^{N_1}\|g(x_i)-\tilde{g}(x_i)\|_2<\frac{\varepsilon}{32\beta}$. Then by the Lipschitz assumption, \begin{align*}\frac{1}{N_1}\sum_{i=1}^{N_1}f(g(x_i))&\leq \frac{1}{N_1}\sum_{i=1}^{N_1}[f(\tilde g(x_i))+\beta\|g(x_i)-\tilde g(x_i)\|_2]\\
    &\leq \frac{1}{N_1}\sum_{i=1}^{N_1}f(\tilde g(x_i))+\frac{\varepsilon}{32}\\ &\leq \frac{1}{N_1}\sum_{i=1}^{N_1}\tilde f(\tilde g(x_i))+\frac{\varepsilon}{16}.\end{align*} Similarly, $\displaystyle\frac{1}{N_1}\sum_{i=1}^{N_1}\tilde f(\tilde g(x_i))-\frac{\varepsilon}{16}\leq\frac{1}{N_1}\sum_{i=1}^{N_1}f(g(x_i))$. Thus,\[\EE_{X\sim\hat P_t}\mathscr{N}_1\Bigg(\frac{\varepsilon}{16},\mathcal{F}_D\circ\mathcal{G}_t,X^{1:N_1}\Bigg)\leq\mathscr{N}_\infty\left(\frac{\varepsilon}{32},\mathcal{F}_D\right)\EE_{X\sim\hat P_t}\mathscr{N}_1\left(\frac{\varepsilon}{32\beta},\mathcal{G}_t,X^{1:N_1}\right).\]

  For the third inequality, let $\varepsilon=\varepsilon(\delta)$ be the solution of \begin{equation}\small
      \label{eq:covering_eps_delta}
  \varepsilon=\left(\frac{64B\beta}{\sqrt{\gamma_t}(\EE w_t)}\sqrt{\frac{\log\mathcal{N}_\infty(\frac{\varepsilon}{6},\mathcal{F}_D)}{N_2}}+64B\sqrt{\frac{\log\mathcal{N}_\infty(\frac{\varepsilon}{32},\mathcal{F}_D)+\log\mathbb{E}\mathcal{N}_1(\frac{\varepsilon}{32\beta},\mathcal{G}_t,X^{1:N_1})}{N_1}}\right)\sqrt{\log\frac{3}{\delta}}.\end{equation} Since the left-hand side is an increasing function with regard to $\varepsilon$ while the right-hand side is a decreasing one, this equation will consistently have exactly one solution for any $N_1,N_2$ and $\delta$. By plugging this equation into the second inequality and letting $N_2$ be large enough such that $\displaystyle 2\exp\left(-\frac{N_2(\EE w_t)^2}{2B^2}\right)<\frac{\delta}{3}$, we can immediately derive the third inequality.
  
  For the last inequality, use the definition of $\hat G_t$:\begin{align*}
    \mathrm{MMD}(\hat G_{t\#}\hat P_t\|Q_t)&\leq \widehat{\mathrm{MMD}}(\hat G_{t\#}\hat P_t\|Q_t)+\text{excess\_term}\\ &\leq\inf_{G_t\in\mathcal{G}_{t}}\widehat{\mathrm{MMD}}(G_{t\#}\hat P_t\|Q_t)+\text{excess\_term}\\ &\leq\inf_{G_t\in\mathcal{G}_t}\mathrm{MMD}(G_{t\#}\hat P_t\|Q_t)+2\times\text{excess\_term}.
  \end{align*}
\end{proof}

\begin{proof}[Proposition \ref{prop:KMMD}]
\begin{align*}
  d_{\mathcal{F}_D}(P\|Q)&=\sup_{f\in\mathcal{F}_D}|\mathbb{E}_{x\sim P}f(x)-\mathbb{E}_{y\sim Q}f(y)|=\sup_{\|f\|_{\mathcal{H}}\leq 1}\left|\mathbb{E}_{x\sim P}f(x)-\mathbb{E}_{y\sim Q}f(y)\right|^2\\
  &=\sup_{\|f\|_{\mathcal{H}}\leq 1}\left|\mathbb{E}_{x\sim P}\langle f,\phi(x)\rangle_{\mathcal{H}}-\mathbb{E}_{y\sim Q}\langle f,\phi(y)\rangle_{\mathcal{H}}\right|^2\\
  &=\sup_{\|f\|_{\mathcal{H}}\leq 1}\left|\langle f,[\mathbb{E}_{x\sim P}\phi(G_t(x))-\mathbb{E}_{y\sim Q}\phi(y)]\rangle_{\mathcal{H}}\right|^2\\
  &=\|\mathbb{E}_{x\sim P}\phi(x)-\mathbb{E}_{y\sim Q}\phi(y)\|_{\mathcal{H}}^2\quad\left(\text{since }f^*=\frac{\mathbb{E}_{x\sim P}\phi(x)-\mathbb{E}_{y\sim Q}\phi(y)}{\|\mathbb{E}_{x\sim P}\phi(x)-\mathbb{E}_{y\sim Q}\phi(y)\|_{\mathcal{H}}}\right)\\
  &=\mathbb{E}_{x,x'\sim P}K(x,x')-2\mathbb{E}_{x\sim P,y\sim Q}K(G_t(x),y)+\mathbb{E}_{y\sim Q,y'\sim Q}K(y,y').
\end{align*}
\end{proof}

\begin{proof}[Theorem \ref{thm:boundrisk}]
  All following inequalities are directly derived by the absolute value inequality, i.e., $|a|+|b|\geq |a+b|$.
  \begin{align*}&\quad\,\,\mathrm{MMD}(\hat Q_t\|Q_t)-\widehat{\mathrm{MMD}}(\hat Q_t\|Q_t)\\
    &=\sup_{f\in\mathcal{F}_D}\left|\EE_{x\sim \hat P_t}f(G_t(x))-\EE_{x'\sim Q_t}f(x')\right|-\sup_{f\in\mathcal{F}_D}\left|\frac{1}{N_1}\sum_{i=1}^{N_1}f(G_t(x_i))-\sum_{j=1}^{N_2}\frac{w_t(x'_j)}{\sum_{j=1}^{N_2}w_t(x'_j)}f(x'_j)\right|\\
    &\leq\sup_{f\in\mathcal{F}_D}\left|\EE_{y\sim {Q_t}}f(y)-\sum_{j=1}^{N_2}\frac{w_t(x'_j)}{\sum_{j=1}^{N_2}w_t(x'_j)}f(x'_j)\right|+\sup_{f\in\mathcal{F}_D}\left|\frac{1}{N_1}\sum_{i=1}^{N_1}f(G_t(x_i))-\EE_{x\sim \hat P_t}f(G_t(x))\right|\\
    &\leq \sup_{f\in\mathcal{F}_D}\left|\EE_{y\sim {Q_t}}f(y)-\EE_{x'\sim P_t}\left[\sum_{j=1}^{N_2}\frac{w_t(x'_j)}{\sum_{j=1}^{N_2}w_t(x'_j)}f(x'_j)\right]\right|\\ &\quad+\left|\EE_{x'\sim P_t}\left[\sum_{j=1}^{N_2}\frac{w_t(x'_j)}{\sum_{j=1}^{N_2}w_t(x'_j)}f(x'_j)\right]-\sum_{j=1}^{N_2}\frac{w_t(x'_j)}{\sum_{j=1}^{N_2}w_t(x'_j)}f(x'_j)\right|\\ &\quad +\sup_{h\in\mathcal{F}_D\circ\mathcal{G}_t}\left|\frac{1}{N_1}\sum_{i=1}^{N_1}h(x_i)-\EE_{x\sim \hat P_t}h(x)\right|.
  \end{align*}
\end{proof}

\begin{proof}[Proposition \ref{prop:forecastrisk}]
  Consider an optimal coupling between $\hat Q_t$ and $Q_t$ and denote it as $\Gamma(\hat Q_t,Q_t)$. The corresponding marginals are denoted as $\tilde{\hat{Q}}_t$ and $\tilde{Q}_t$, respectively. Note that $\tilde{\hat{Q}}_t\stackrel{\text{d}}{=}\hat Q_t$ and $\tilde{Q}_t\stackrel{\text{d}}{=}Q_t$. Then by the definition,\[\footnotesize\begin{aligned}
    &\mathrm{MMD}(\mathcal{M}_t(\hat Q_t)+p_{\eta_{t+1}}\|\mathcal{M}_t(Q_t)+p_{\eta_{t+1}})\\
    =\,&\sup_{f\in\mathcal{F}_D}\left|\mathbb{E}_{X\sim \hat Q_t,\eta_{t+1}^1\sim p_{\eta_{t+1}}}f(\mathcal{M}_t(X)+\eta_{t+1}^1)-\mathbb{E}_{Y\sim Q_t,\eta_{t+1}^2\sim p_{\eta_{t+1}}}f(\mathcal{M}_t(Y)+\eta_{t+1}^2)\right|\\
    =\,&\sup_{f\in\mathcal{F}_D}\left|\mathbb{E}_{X\sim \tilde{\hat{Q}}_t,\eta_{t+1}\sim p_{\eta_{t+1}}}f(\mathcal{M}_t(X)+\eta_{t+1})-\mathbb{E}_{Y\sim \tilde{Q}_t,\eta_{t+1}\sim p_{\eta_{t+1}}}f(\mathcal{M}_t(Y)+\eta_{t+1})\right|\\
    =\,&\sup_{f\in\mathcal{F}_D}\left|\mathbb{E}_{X\sim \tilde{\hat{Q}}_t,Y\sim\tilde{Q}_t,\eta_{t+1}\sim p_{\eta_{t+1}}}f(\mathcal{M}_t(Y)+\eta_{t+1}+(\mathcal{M}_t(X)-\mathcal{M}_t(Y)))-\mathbb{E}_{Y\sim \tilde{Q}_t,\eta_{t+1}\sim p_{\eta_{t+1}}}f(\mathcal{M}_t(Y)+\eta_{t+1})\right|\\
    \leq\,&\underbrace{\sup_{f\in\mathcal{F}_D}\left|\mathbb{E}_{X\sim \tilde{\hat{Q}}_t,Y\sim\tilde{Q}_t,\eta_{t+1}\sim p_{\eta_{t+1}}}f(\mathcal{M}_t(Y)+\eta_{t+1})-\mathbb{E}_{Y\sim \tilde{Q}_t,\eta_{t+1}\sim p_{\eta_{t+1}}}f(\mathcal{M}_t(Y)+\eta_{t+1})\right|}_{=\,0}\\
    &\qquad\qquad\qquad\qquad\qquad\qquad\qquad+\beta\EE_{X\sim \tilde{\hat{Q}}_t,Y\sim\tilde{Q}_t}\|\mathcal{M}_t(X)-\mathcal{M}_t(Y)\|_2\qquad(\text{by Lipschitzity of }f)\\
    \leq\,&\beta\gamma\EE_{X\sim \tilde{\hat{Q}}_t, Y\sim\tilde{Q}_t}\|X-Y\|_2 \qquad\qquad\,\,\,\,(\text{by Lipschitzity of }\mathcal{M}_t)\\
    \leq\,&\beta\gamma\left(\EE_{X,Y\sim\Gamma(\hat Q_t,Q_t)}\|X-Y\|_2^2\right)^{\frac{1}{2}} \qquad(\text{by Cauchy's inequality})\\
    =\,&W_2(\hat Q_t\|Q_t).
  \end{aligned}\]
\end{proof}

\begin{proof}[Proposition \ref{prop:forecastriskkernel}]
  By the definition, \[\begin{aligned}
    &\mathrm{MMD}(\mathcal{M}_t(\hat Q_t)+p_{\eta_{t+1}}\|\mathcal{M}_t(Q_t)+p_{\eta_{t+1}})\\
    =\,&\sup_{\|f\|_{\mathcal{H}}\leq 1}\left|\mathbb{E}_{X\sim \hat Q_t,\eta_{t+1}^1\sim p_{\eta_{t+1}}}f(\mathcal{M}_t(X)+\eta_{t+1}^1)-\mathbb{E}_{Y\sim Q_t,\eta_{t+1}^2\sim p_{\eta_{t+1}}}f(\mathcal{M}_t(Y)+\eta_{t+1}^2)\right|\\
    \leq \,&\sup_{\|f\|_{\mathcal{H}}\leq 1}\left|\mathbb{E}_{X\sim \hat Q_t}f(\mathcal{M}_t(X))-\mathbb{E}_{Y\sim Q_t}f(\mathcal{M}_t(Y))\right|+\beta\EE_{\eta_{t+1}^1\sim p_{\eta_{t+1}}}\|\eta_{t+1}^1\|+\beta\EE_{\eta_{t+1}^2\sim p_{\eta_{t+1}}}\|\eta_{t+1}^2\|\\
    \leq\,&\sup_{\|g\|_{\mathcal{H}}\leq \Omega}\left|\mathbb{E}_{X\sim \hat Q_t}g(X)-\mathbb{E}_{Y\sim Q_t}g(Y)\right|+2\beta\EE_{p_{\eta_{t+1}}}\|\eta_{t+1}\|\quad(g=f\circ\mathcal{M}_t)\\
    \leq\,&\Omega\sup_{\|g\|_{\mathcal{H}}\leq 1}\left|\mathbb{E}_{X\sim \hat Q_t}g(X)-\mathbb{E}_{Y\sim Q_t}g(Y)\right|+2\beta(\EE_{p_{\eta_{t+1}}}\|\eta_{t+1}\|_2^2)^{\frac{1}{2}}\\
    =\,&\Omega\cdot\mathrm{MMD}(\hat Q_t\|Q_t)+2\beta\cdot\trace(\Var(\eta_{t+1}))^{\frac{1}{2}}.
  \end{aligned}\]
\end{proof}

\begin{proof}[Corollary \ref{cor:log-covering}] By the sub-Gaussian assumption, \[\footnotesize\begin{aligned}\mathbb{P}\Bigg(\exists i\in[N_1],\|X_i\|_\infty\geq C_6\sigma_t\sqrt{\log(2dN_1)+\frac{N_1\varepsilon^2}{2048B^2}}\Bigg)&\stackrel{\text{union bound}}{\leq} N_1d\mathbb{P}\left(X_{11}\geq C_6\sigma_t\sqrt{\log(2dN_1)+\frac{N_1\varepsilon^2}{2048B^2}}\right)\\
  &\stackrel{\text{tail bound}}{\leq}\exp\left(-\frac{N_1\varepsilon^2}{2048B^2}\right)\end{aligned}\]
  where $C_6$ is an absolute constant and the term ``$\exp\left(-\frac{N_1\varepsilon^2}{2048B^2}\right)$'' matches the nonasymptotic bound on the right-hand side of \eqref{eq:supbound}. In this way, under the condition \[\mathscr{A}:=\left\{\max_{1\leq i\leq N_1}\|X_i\|_\infty\leq M:=C_6\sigma_t\sqrt{\log(2dN_1)+\frac{N_1\varepsilon^2}{2048B^2}}\right\},\] since the $l_1$ norm is dominated by the $L_\infty$ norm, 
  \begin{small}\begin{align}\log\mathbb{E}_{X\sim\hat P_t}\mathscr{N}_1\left(\frac{\varepsilon}{32\beta},\mathcal{G}_t,X^{1:N_1}\right) & \leq \log\mathscr{N}_\infty\left(\frac{\varepsilon}{32\beta},\mathcal{G}_t|_{[-M,M]^d}\right)\label{eq:expectation_transfer_to_bounded_case}\\
    &\stackrel{\eqref{eq:suzuki}}{\leq} 2C_7SL\log\left(\beta L\|W\|_\infty(R\vee 1)\sigma_t\sqrt{\frac{\log(2dN_1)}{\varepsilon^2}+\frac{N_1}{2048B^2}}\right).\end{align}\end{small}
  Therefore, \[\small\begin{aligned}
    &\qquad\mathbb{P}\left(\sup_{G_t\in\mathcal{G}_t}\left|\mathrm{MMD}(G_{t\#}\hat{P}_t\|Q_t)-\widehat{\mathrm{MMD}}(G_{t\#}\hat{P}_t\|Q_t)\right|>\varepsilon+\frac{12B(\chi^2(Q_t\|P_t)+1)}{N_2}\right)\\
    &\leq\mathbb{P}\left(\sup_{G_t\in\mathcal{G}_t}\left|\mathrm{MMD}(G_{t\#}\hat{P}_t\|Q_t)-\widehat{\mathrm{MMD}}(G_{t\#}\hat{P}_t\|Q_t)\right|>\varepsilon+\frac{12B(\chi^2(Q_t\|P_t)+1)}{N_2},\mathscr{A}\right)+\mathbb{P}(\mathscr{A}^c)\\
    &\leq\mathbb{P}\left(\sup_{G_t\in\mathcal{G}_t}\left|\mathrm{MMD}(G_{t\#}\hat{P}_t\|Q_t)-\widehat{\mathrm{MMD}}(G_{t\#}\hat{P}_t\|Q_t)\right|>\varepsilon+\frac{12B(\chi^2(Q_t\|P_t)+1)}{N_2}\Bigg|\mathscr{A}\right)+\exp\left(-\frac{N_1\varepsilon^2}{2048B^2}\right)\\
    &\leq 2\mathscr{N}_\infty\left(\frac{\varepsilon}{6},\mathcal{F}_D\right)\exp\left(-\frac{\gamma_tN_2(\mathbb{E}w_t)^2\varepsilon^2}{5760B^2\beta^2}\right)+2\exp\left(-\frac{N_2(\mathbb{E}w_t)^2}{2B^2}\right)\\
    &\qquad\qquad\qquad\qquad\qquad\qquad+9\mathscr{N}_\infty\left(\frac{\varepsilon}{32},\mathcal{F}_D\right)\mathscr{N}_\infty\left(\frac{\varepsilon}{32\beta},\mathcal{G}_t|_{[-M,M]^d}\right)\exp\left(-\frac{N_1\varepsilon^2}{2048B^2}\right),
  \end{aligned}\] where the last inequality follows from \eqref{eq:directPACbound} and \eqref{eq:expectation_transfer_to_bounded_case}.
\end{proof}

\begin{proof}[Proposition \ref{prop:gaussconvergence}] We first prove a lemma.
\begin{lemma}\label{lem:subgaussian_approximation_error}
    Let $N$ be the size of neural networks. Then under assumptions \ref{ass:1}-\ref{ass:6} and the sub-Gaussian condition, the approximation error will converge at the following rate:
    \begin{equation}\label{eq:subgaussian_approximation_error}
   \inf_{G_t\in\mathcal{G}_t}\mathrm{MMD}(\hat Q_t\|Q_t))\lesssim (\log N)^{\frac{s}{2}}N^{-\frac{s}{d}}.
   \end{equation}
   That is, the convergence rate with regard to $N$ in sub-Gaussian cases will be slower by a factor $(\log N)^{\frac{s}{2}}$ than that in bounded cases.
\end{lemma}
\begin{proof} By bounded, Lipschitzity and sub-Gaussian assumptions,
    \begin{align*}
        \inf_{G_t\in\mathcal{G}_t}\text{MMD}&(\hat Q_t\|Q_t) = \inf_{G_t\in\mathcal{G}_t}\sup_{f\in\mathcal{F}_D}\left|\mathbb{E}_{x\sim G_{t\#}\hat P_t}f(x)-\mathbb{E}_{y\sim Q_t}f(y)\right|\\
        &=\inf_{G_t\in\mathcal{G}_t}\sup_{f\in\mathcal{F}_D}\left|\mathbb{E}_{x\sim G_{t\#}\hat P_t}f(x)-\mathbb{E}_{y\sim G_{t\#}'\hat P_t}f(y)\right|\\
        &=\inf_{G_t\in\mathcal{G}_t}\sup_{f\in\mathcal{F}_D}\int f(G_t(x))-f(G_t'(x))\text{d}\hat P_t\\
        &=\inf_{G_t\in\mathcal{G}_t}\sup_{f\in\mathcal{F}_D}\left(\int_{\|x\|_\infty>H} f(G_t(x))-f(G_t'(x))\text{d}\hat P_t+\int_{\|x\|_\infty\leq H} f(G_t(x))-f(G_t'(x))\text{d}\hat P_t\right)\\
        &\lesssim 2B \hat P_t(\|x\|_\infty>H) + \inf_{G_t\in\mathcal{G}_t}\beta\int_{\|x\|_\infty\leq H}\|G_t(x)-G_t'(x)\|_2\text{d}\hat P_t\\
        &\lesssim 2Bd\exp(-\frac{H^2}{\sigma_t^2}) + \inf_{G_t\in\mathcal{G}_t}\beta\|G_t-G_t'\|_{L_\infty[-H,H]^d}.
    \end{align*}
Assume $m\in\mathbb{N}$ satisfies $\frac{d}{p}<s<\min(m,m-1+\frac{1}{p})$, $\nu=\frac{s-\delta}{2\delta}$, $\varepsilon=N^{-\frac{s}{d}-(\nu^{-1}+d^{-1})(\frac{d}{p}-s)}+\log(N)^{-1}$, $W_0=6dm(m+2)+2d$, $L=3+2\lceil\log_2(\frac{3^{d\vee m}}{\varepsilon c_{(d,m)}})+5\rceil\lceil\log_2(d\vee m)\rceil$, $W=NW_0$, $S=(L-1)W_0^2N+N$ and $R=O(N^{(\nu^{-1}+d^{-1})(1\vee(\frac{d}{p}-s))})$ where $N$ is the number of B-spline functions.
By Proposition 1 of \citet{suzuki2019adaptivity} and dividing the domain into cubes with each side length 1, the second term has the convergence rate $H^sN^{-\frac{s}{d}}$. Since this inequality holds for all $H\in\mathbb{R}_+$, then by minimizing the right-hand side over $H$ we can find the extra factor $(\log N)^{\frac{s}{2}}$. 
\end{proof}

  Then we can derive this proposition simply by (1) using the chaining method (i.e., take expectations on the empirical MMD term of \eqref{eq:PACbound} in Theorem \ref{thm:simplefinalupperbound} and apply $\EE X=\int_0^{+\infty}\mathbb{P}(X>t)\text{d}t$) or refering to Lemma 26 in \citet{liang2021well}(i.e., the standard entropy integral bound) and (2) plugging \eqref{eq:kernelcovering} with $\gamma=\frac{1}{d}$ in the second case, \eqref{eq:sub-gaussin-covering} and \eqref{eq:subgaussian_approximation_error} into the Dudley entropy integral and minimizing the right-hande side over $\delta\in[0,\frac{1}{2}]$ and $N\in\mathbb{Z}_+$. That is, \[\begin{aligned}
      &\mathbb{E}\mathrm{MMD}(\hat Q_t\|Q_t)-\inf_{G_t\in\mathcal{G}_t}\mathrm{MMD}(\hat Q_t\|Q_t)\\
      &\lesssim \,\inf_{0<\delta<\frac{1}{2}}\Bigg(4\delta+\frac{8\sqrt{2}}{\sqrt{N_1}}\int_{\delta}^{\frac{1}{2}}\sqrt{\log\mathscr{N}_\infty(\varepsilon,\mathcal{F}_D)+\log\mathscr{N}_\infty(\frac{\varepsilon}{32\beta},\mathcal{G}_t|_{[-M,M]^d})}\text{d}\varepsilon\\
      &\qquad\qquad\qquad\qquad\qquad\qquad\qquad\quad +\frac{8\sqrt{2}}{\sqrt{N_2}}\int_{\delta}^{\frac{1}{2}}\sqrt{\log\mathscr{N}_\infty(\varepsilon,\mathcal{F}_D)}\text{d}\varepsilon\Bigg)\\
      &\lesssim \inf_{0<\delta<\frac{1}{2}}\Bigg(4\delta+\frac{8\sqrt{2}}{\sqrt{N_1}}\int_{\delta}^{\frac{1}{2}}\sqrt{\log\mathscr{N}_\infty(\varepsilon,\mathcal{F}_D)}\text{d}\varepsilon+\frac{8\sqrt{2}}{\sqrt{N_1}}\int_{\delta}^{\frac{1}{2}}\sqrt{\log\mathscr{N}_\infty(\frac{\varepsilon}{32\beta},\mathcal{G}_t|_{[-M,M]^d})}\text{d}\varepsilon\\
      &\qquad\qquad\qquad\qquad\qquad\qquad\qquad\quad +\frac{8\sqrt{2}}{\sqrt{N_2}}\int_{\delta}^{\frac{1}{2}}\sqrt{\log\mathscr{N}_\infty(\varepsilon,\mathcal{F}_D)}\text{d}\varepsilon\Bigg)\\
      &\lesssim \frac{N\log(NN_1)}{\sqrt{N_1}}+\frac{1}{\sqrt{N_2}}\qquad(\text{by taking } \delta=0)\\
      \Rightarrow \,\,& \EE\mathrm{MMD}(\hat Q_t\|Q_t)\\
      &\lesssim (\log N)^{\frac{s}{2}}N^{-\frac{s}{d}}+\frac{N\log(NN_1)}{\sqrt{N_1}}+\frac{1}{\sqrt{N_2}}\\
      &\lesssim (\log N_1)^{\frac{s}{2}\vee 1}N_1^{-\frac{s}{2(s+d)}}+N_2^{-\frac{1}{2}}. \qquad(\text{by taking } N\asymp N_1^{\frac{d}{2(s+d)}})\end{aligned}\]
\end{proof}
\begin{proof}[Proposition \ref{prop:meanconvergence}]
  Under the circumstance of the linear kernel, minimizing the empirical MMD loss is equivalent to minimizing the empirical $L_2$ loss:
  \begin{equation*}\footnotesize\begin{aligned}
    \hat G_t&=\underset{G_t\in\mathcal{G}_t}{\arg\min}\frac{1}{N_1^2}\sum_{i=1}^{N_1}\sum_{j=1}^{N_1}K(G_t(x_i),G_t(x_j))-\frac{2}{N_1}\sum_{i=1}^{N_1}\sum_{j=1}^{N_2}\underbrace{\frac{w_t(x_j')}{\sum_{j=1}^{N_2}w_t(x_j')}}_{:=w_j}K(G_t(x_i),x_j')\\
    &=\underset{G_t\in\mathcal{G}_t}{\arg\min}\frac{1}{N_1^2}\sum_{i=1}^{N_1}\sum_{j=1}^{N_1}K(G_t(x_i),G_t(x_j))-\frac{2}{N_1}\sum_{i=1}^{N_1}\sum_{j=1}^{N_2}w_jK(G_t(x_i),x_j')+\sum_{i=1}^{N_2}\sum_{j=1}^{N_2}w_iw_jK(x_i',x_j')\\
    &=\underset{G_t\in\mathcal{G}_t}{\arg\min}\left(\frac{\sum_{i=1}^{N_1}G_t(x_i)}{N_1}\right)^T\left(\frac{\sum_{j=1}^{N_1}G_t(x_j)}{N_1}\right)-2\left(\frac{\sum_{i=1}^{N_1}G_t(x_i)}{N_1}\right)^T\left(\sum_{j=1}^{N_2}w_jx_j'\right)+\left(\sum_{i=1}^{N_2}w_ix_i'\right)^T\left(\sum_{j=1}^{N_2}w_jx_j'\right)\\
    &=\underset{G_t\in\mathcal{G}_t}{\arg\min}\left\|\frac{\sum_{i=1}^{N_1}G_t(x_i)}{N_1}-\sum_{j=1}^{N_2}w_jx_j'\right\|_2^2.
  \end{aligned}\end{equation*} Then by taking the following two steps: (1) mimicking the proofs of Proposition \ref{prop:gaussconvergence} in this paper and (2) plugging \eqref{eq:kernelcovering} with $\gamma=d$ in the first case and \eqref{eq:sub-gaussin-covering} into the Dudley entropy integral, we can derive the desired conclusion. Moreover, the approximation error will be zero (i.e., the term $\inf_{G_t\in\mathcal{G}_t}\mathrm{MMD}(G_{t\#}\hat P_t\|Q_t)$ vanishes) as long as the constant function $h(X)\equiv \EE_{Y\sim Q_t}Y$ belongs to the set $\Phi(L,W,S,R)$. For example, let $R\geq\|\EE_{Y\sim Q_t}Y\|_\infty$.
\end{proof}

\newpage
\section{Supplementary Figures and Tables}
\label{app:AppendixB}
\begin{figure}[H]
  \centering
  \subfigure{
    \includegraphics[width=0.32\columnwidth]{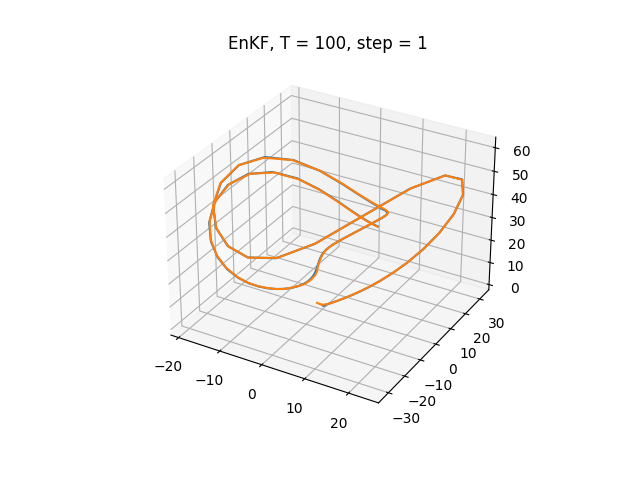}
  }\subfigure{
    \includegraphics[width=0.32\columnwidth]{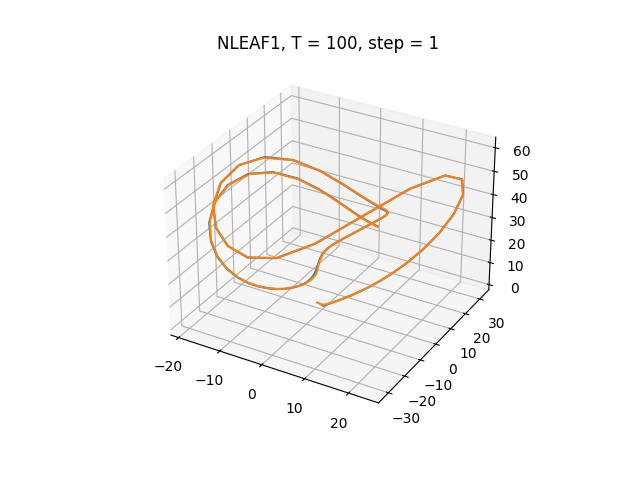}
  }\subfigure{
    \includegraphics[width=0.32\columnwidth]{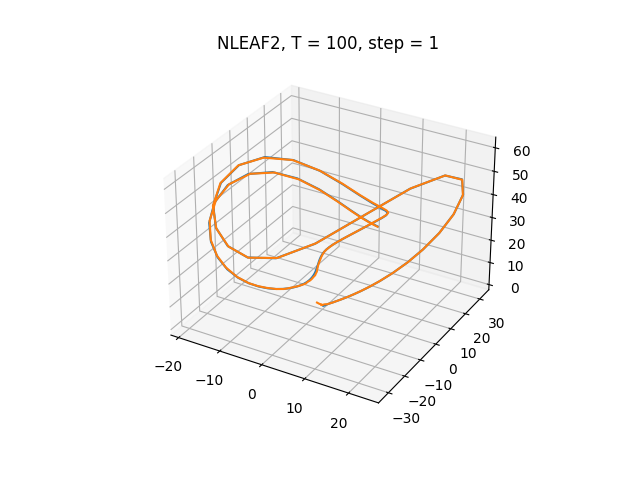}
  }\\\subfigure{
    \includegraphics[width=0.32\columnwidth]{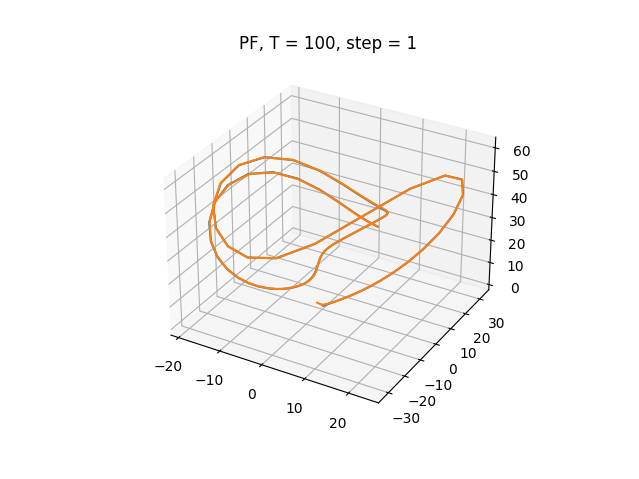}
  }\subfigure{
    \includegraphics[width=0.32\columnwidth]{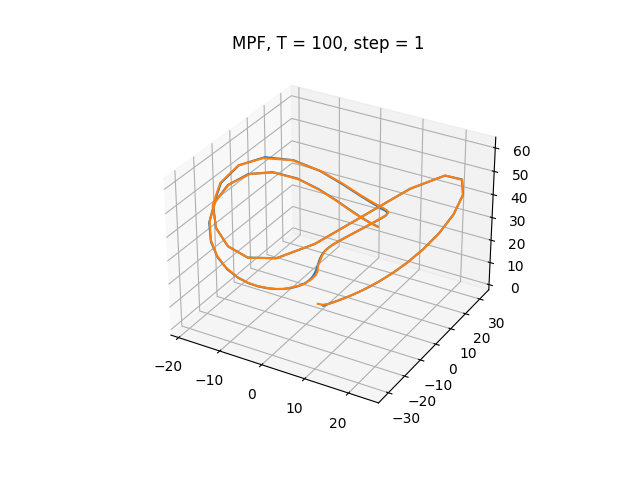}
  }\subfigure{
    \includegraphics[width=0.32\columnwidth]{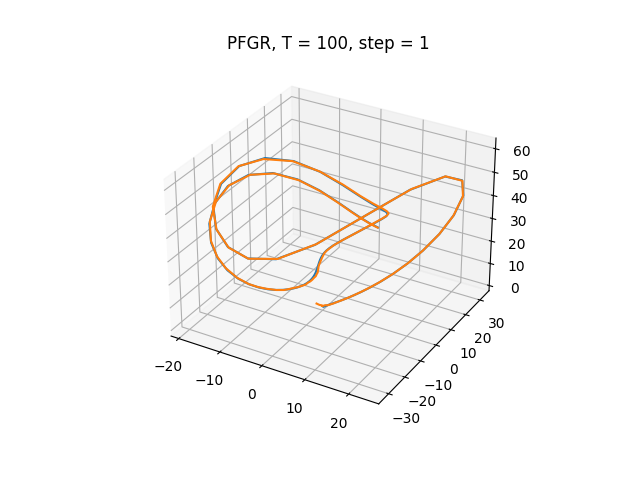}
  }\\\subfigure{
    \includegraphics[width=0.32\columnwidth]{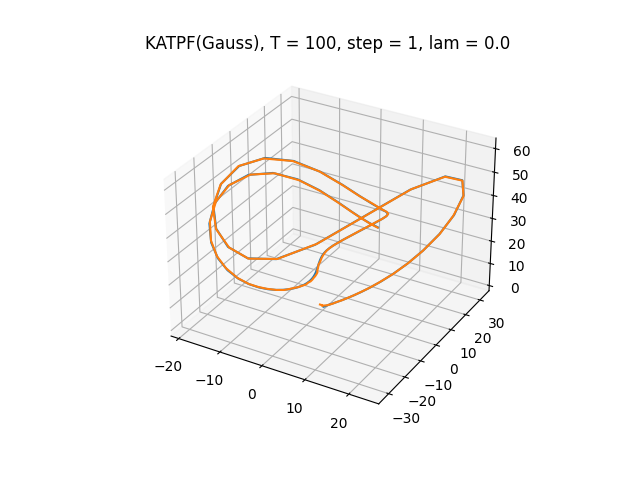}
  }\subfigure{
    \includegraphics[width=0.32\columnwidth]{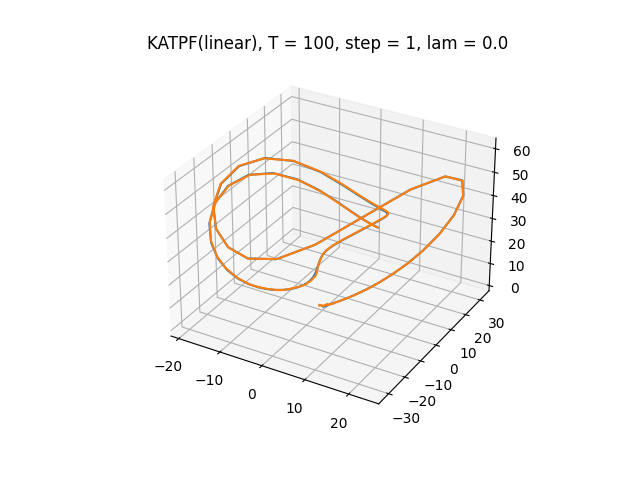}
  }\subfigure{
    \includegraphics[width=0.32\columnwidth]{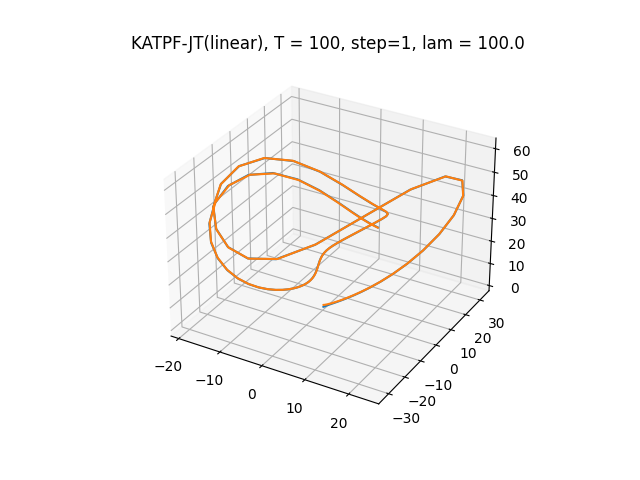}
  }
  \caption{Experimental results of the Lorenz63 system when $T=100$ and $\Delta t=0.02$.}
  \label{fig:lorenz63t100step1}
\end{figure}

\begin{figure}[H]
  \centering
  \subfigure{
    \includegraphics[width=0.32\columnwidth]{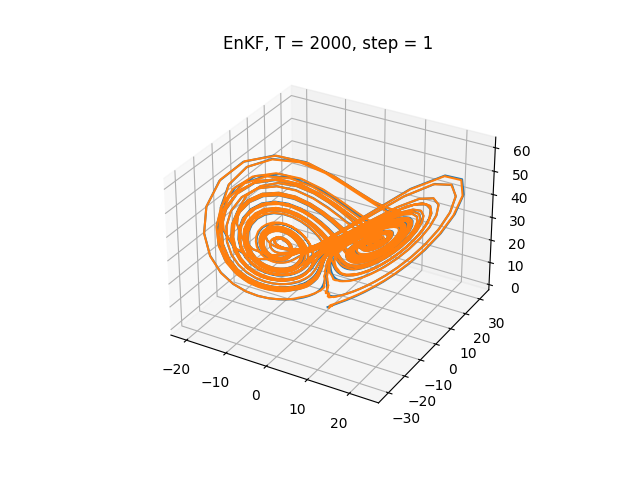}
  }\subfigure{
    \includegraphics[width=0.32\columnwidth]{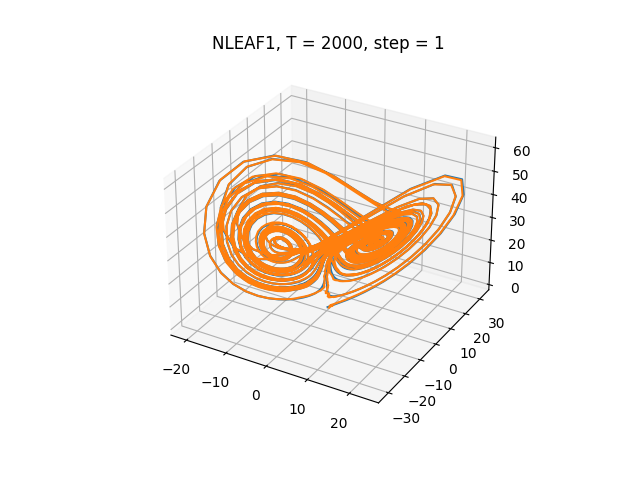}
  }\subfigure{
    \includegraphics[width=0.32\columnwidth]{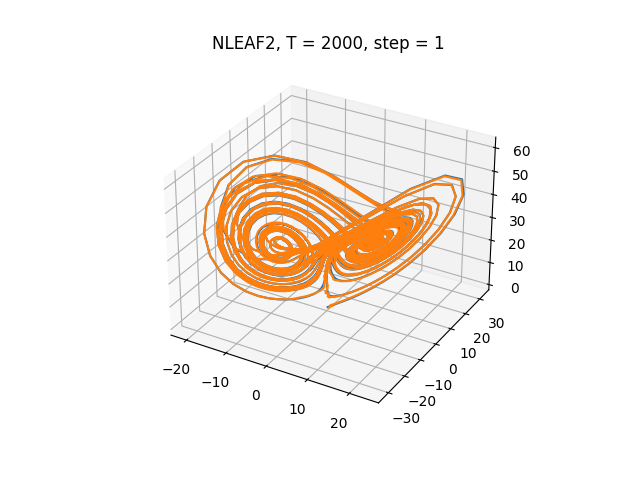}
  }\\\subfigure{
    \includegraphics[width=0.32\columnwidth]{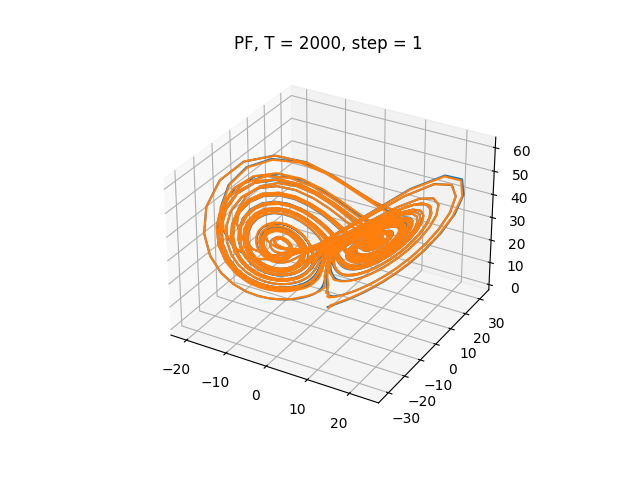}
  }\subfigure{
    \includegraphics[width=0.32\columnwidth]{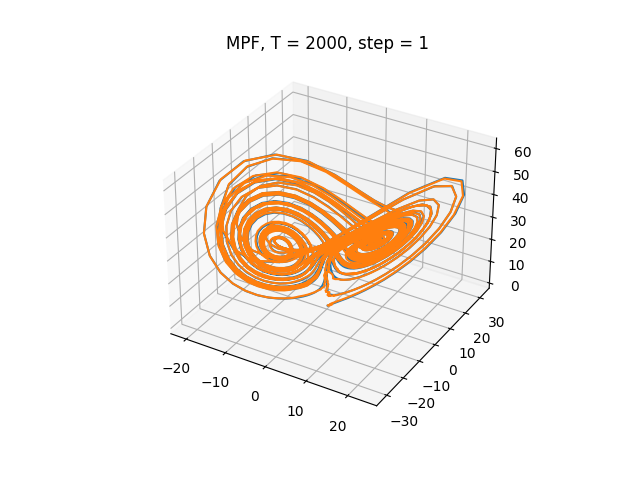}
  }\subfigure{
    \includegraphics[width=0.32\columnwidth]{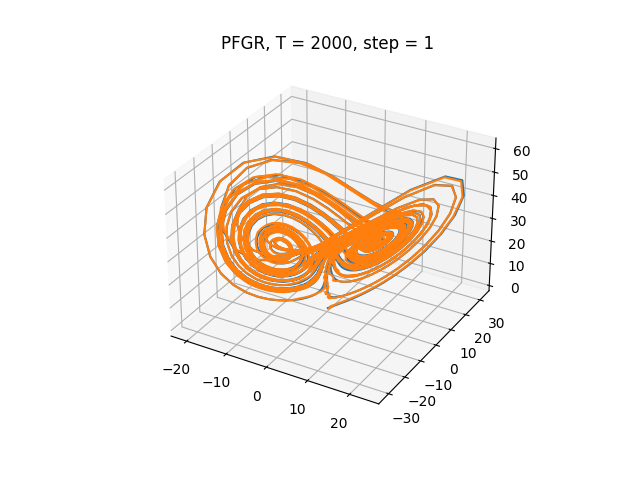}
  }\\\subfigure{
    \includegraphics[width=0.32\columnwidth]{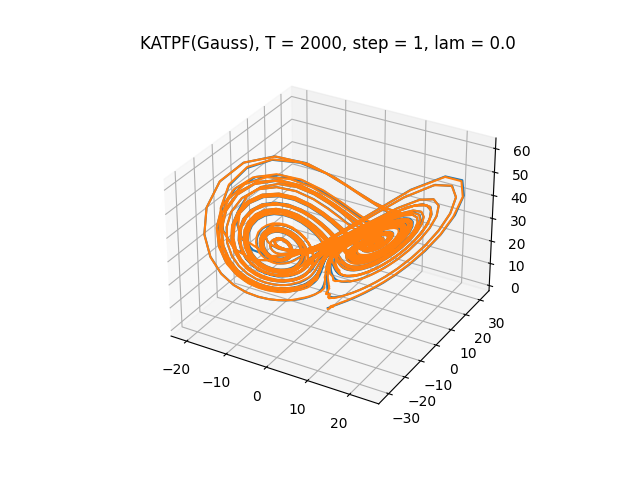}
  }\subfigure{
    \includegraphics[width=0.32\columnwidth]{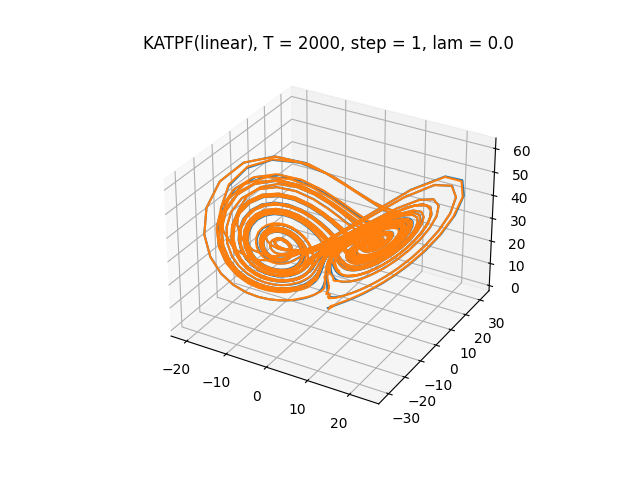}
  }\subfigure{
    \includegraphics[width=0.32\columnwidth]{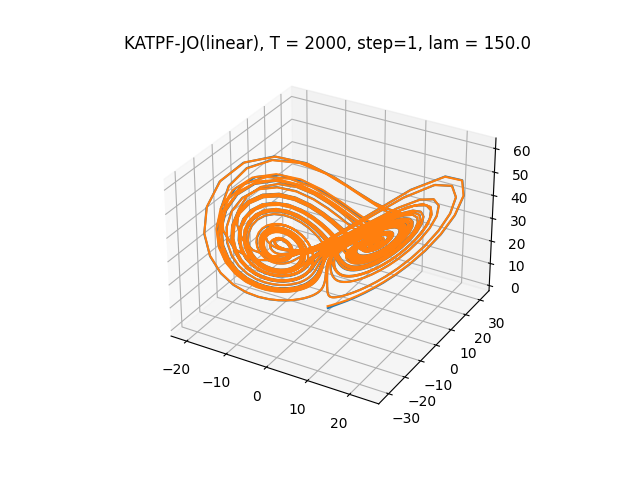}
  }
  \caption{Experimental results of the Lorenz63 system when $T=2000$ and $\Delta t=0.02$.}
  \label{fig:lorenz63step1}
\end{figure}

\begin{figure}[H]
  \centering
  \subfigure{
    \includegraphics[width=0.32\columnwidth]{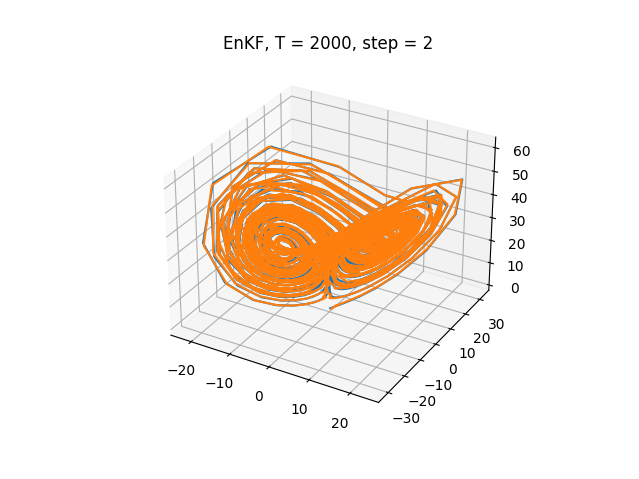}
  }\subfigure{
    \includegraphics[width=0.32\columnwidth]{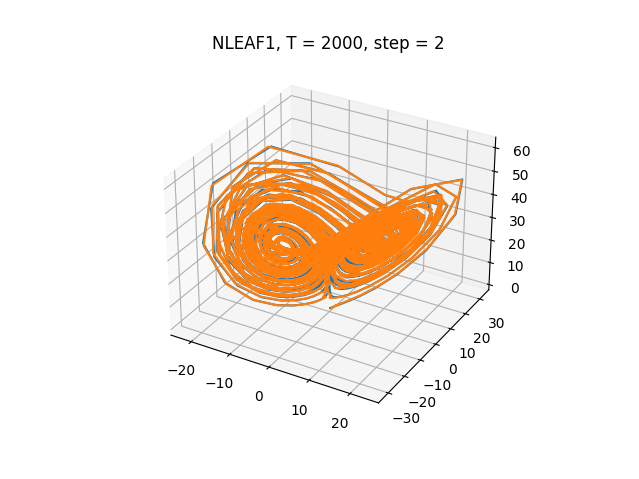}
  }\subfigure{
    \includegraphics[width=0.32\columnwidth]{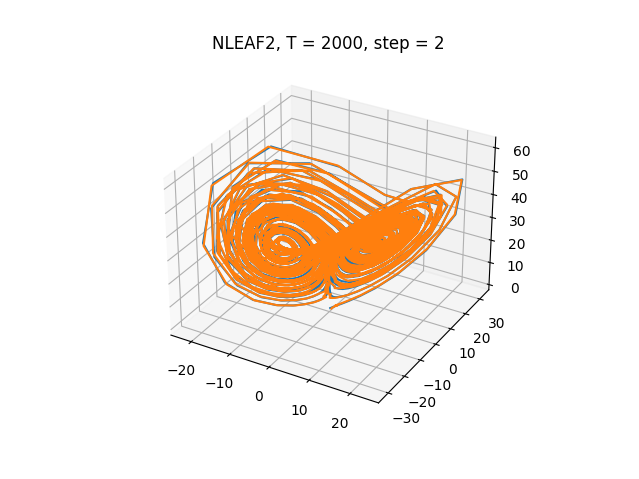}
  }\\\subfigure{
    \includegraphics[width=0.32\columnwidth]{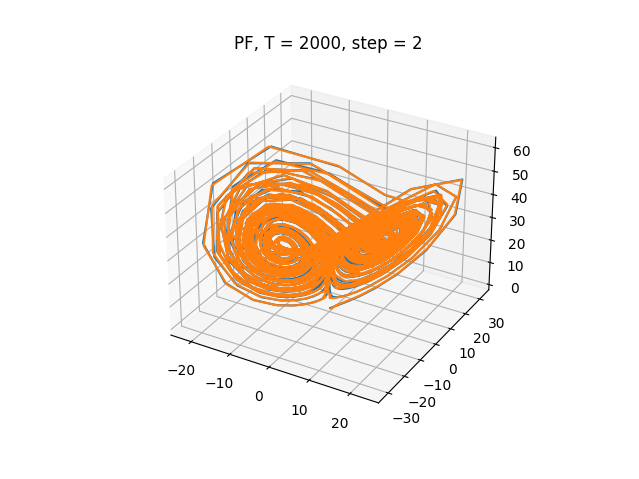}
  }\subfigure{
    \includegraphics[width=0.32\columnwidth]{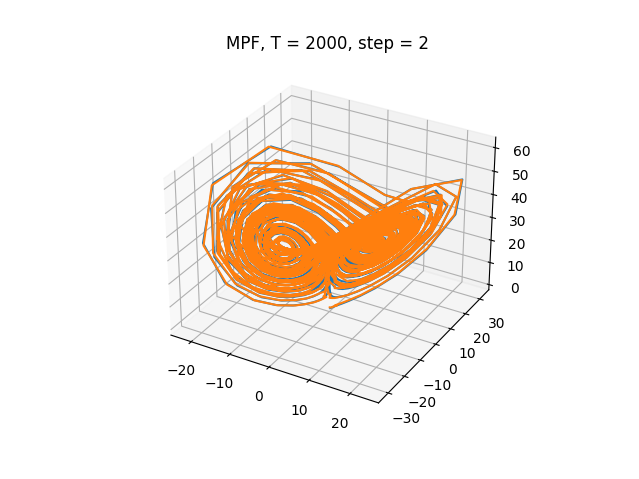}
  }\subfigure{
    \includegraphics[width=0.32\columnwidth]{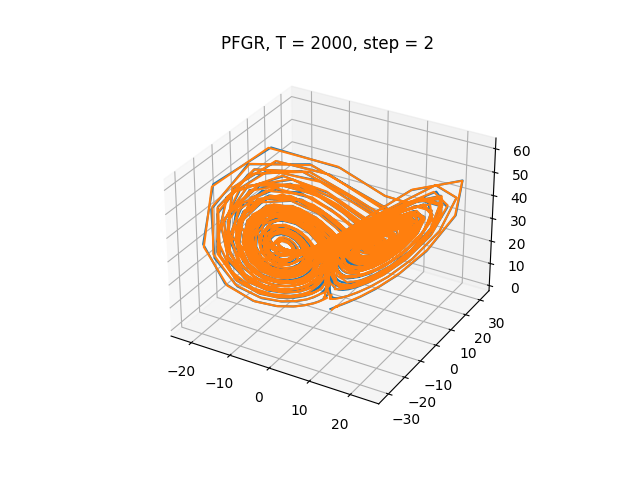}
  }\\\subfigure{
    \includegraphics[width=0.32\columnwidth]{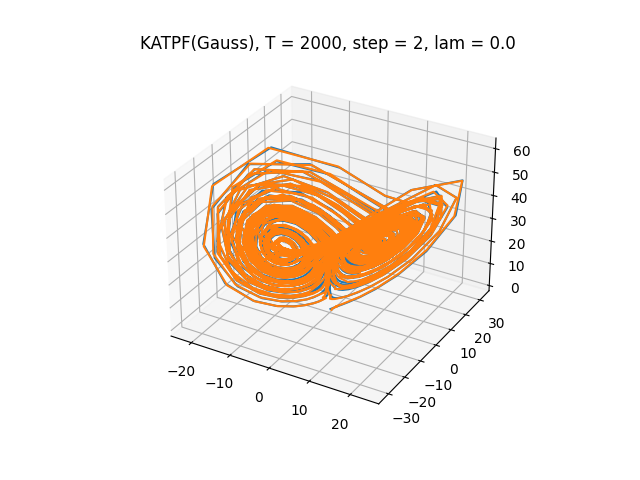}
  }\subfigure{
    \includegraphics[width=0.32\columnwidth]{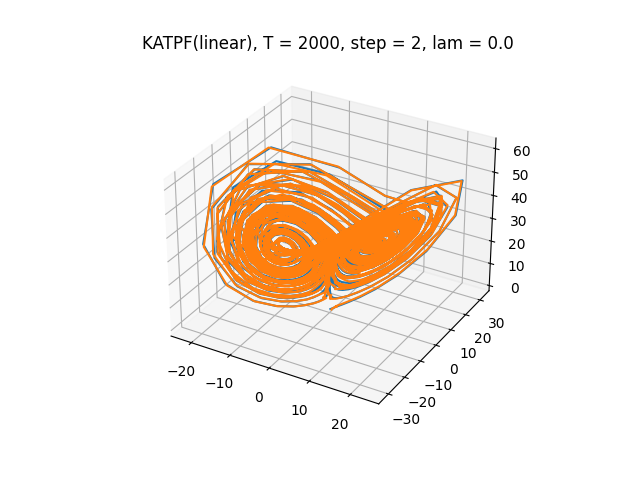}
  }\subfigure{
    \includegraphics[width=0.32\columnwidth]{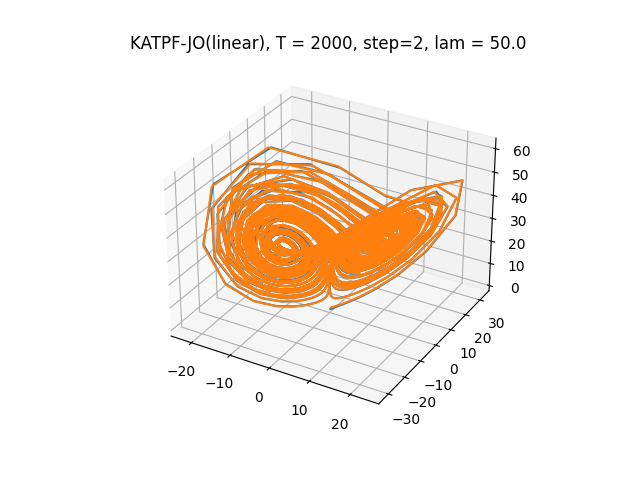}
  }
  \caption{Experimental results of the Lorenz63 system when $T=2000$ and $\Delta t=0.04$.}
  \label{fig:lorenz63step2}
\end{figure}

\begin{figure}[H]
  \centering
  \subfigure{
    \includegraphics[width=0.32\columnwidth]{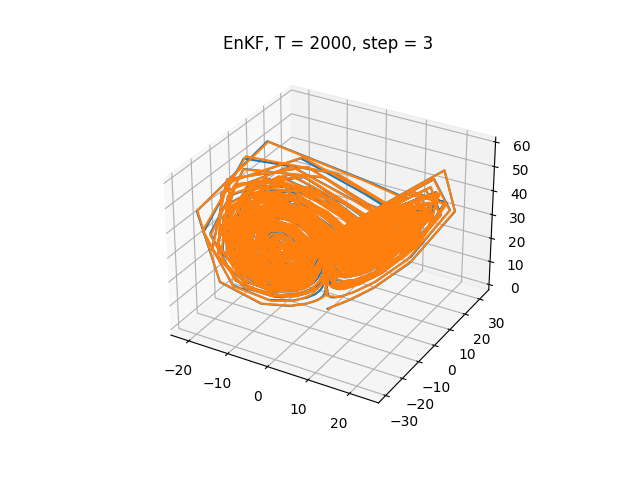}
  }\subfigure{
    \includegraphics[width=0.32\columnwidth]{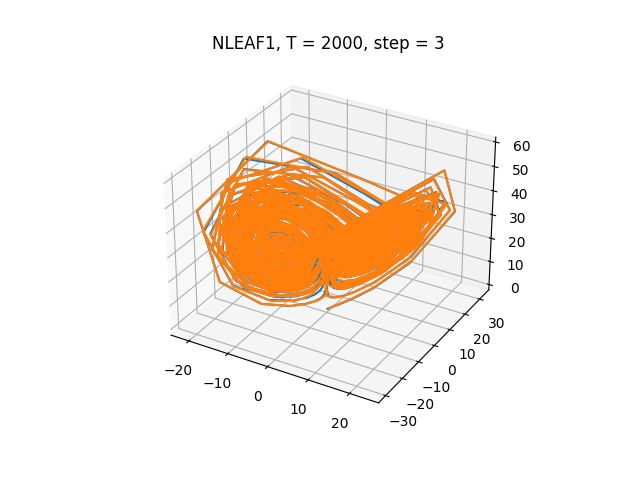}
  }\subfigure{
    \includegraphics[width=0.32\columnwidth]{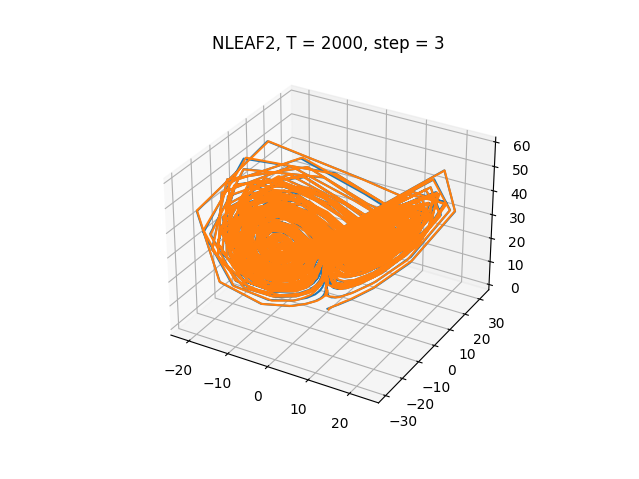}
  }\\\subfigure{
    \includegraphics[width=0.32\columnwidth]{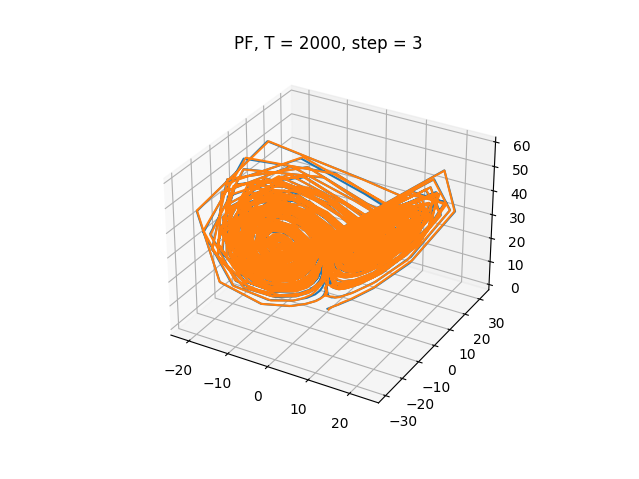}
  }\subfigure{
    \includegraphics[width=0.32\columnwidth]{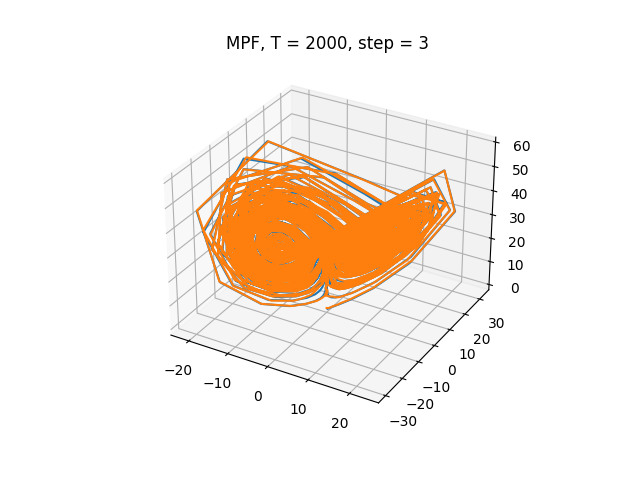}
  }\subfigure{
    \includegraphics[width=0.32\columnwidth]{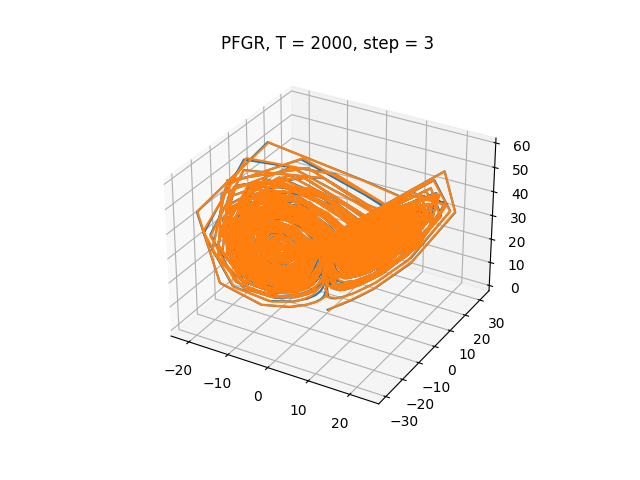}
  }\\\subfigure{
    \includegraphics[width=0.32\columnwidth]{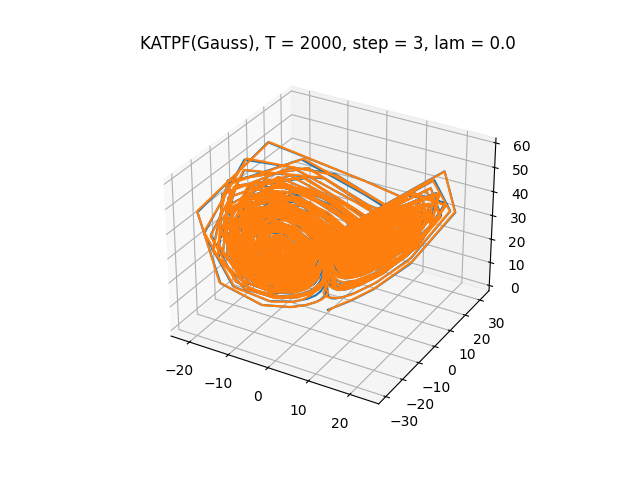}
  }\subfigure{
    \includegraphics[width=0.32\columnwidth]{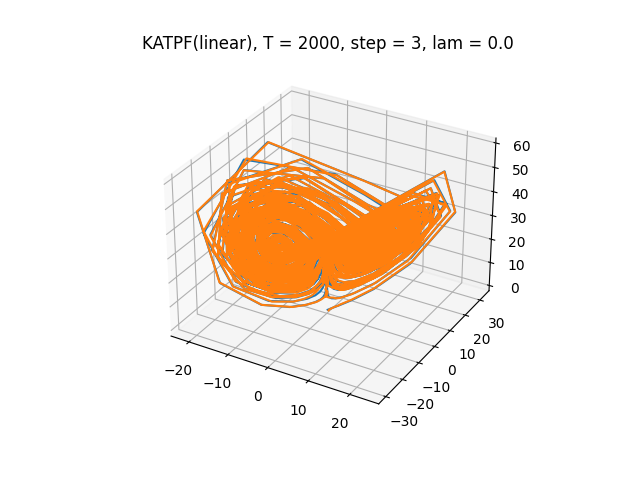}
  }\subfigure{
    \includegraphics[width=0.32\columnwidth]{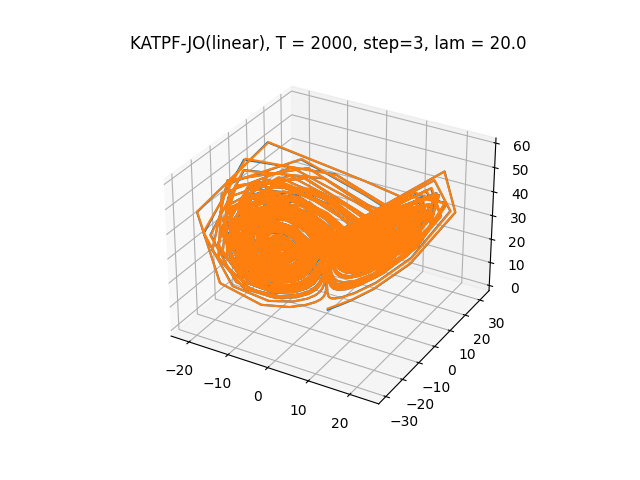}
  }
  \caption{Experimental results of the Lorenz63 system when $T=2000$ and $\Delta t=0.06$.}
  \label{fig:lorenz63step3}
\end{figure}

\begin{figure}[H]
  \centering
  \subfigure{
    \includegraphics[width=0.32\columnwidth]{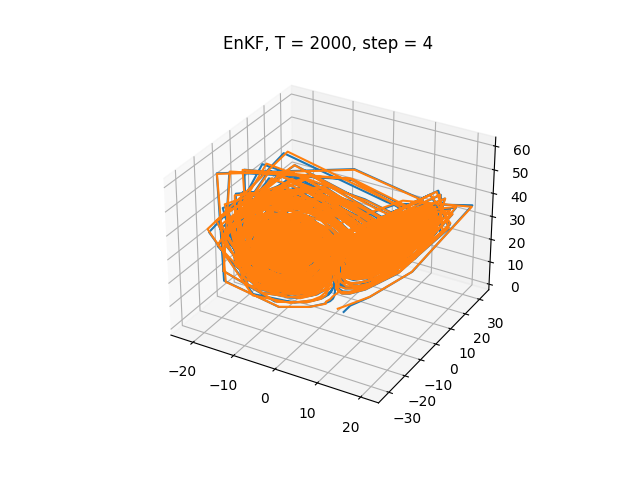}
  }\subfigure{
    \includegraphics[width=0.32\columnwidth]{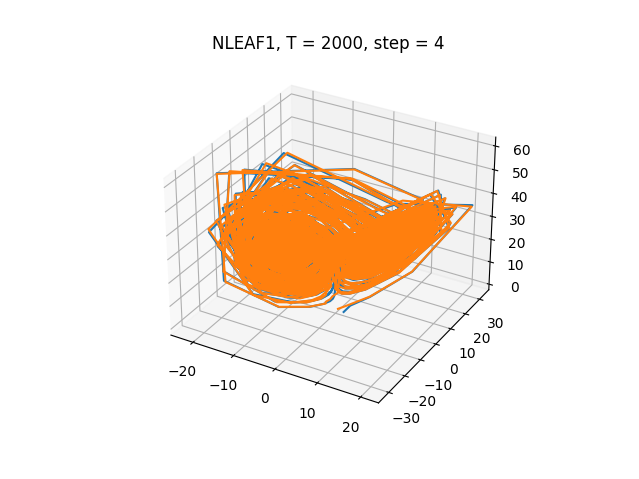}
  }\subfigure{
    \includegraphics[width=0.32\columnwidth]{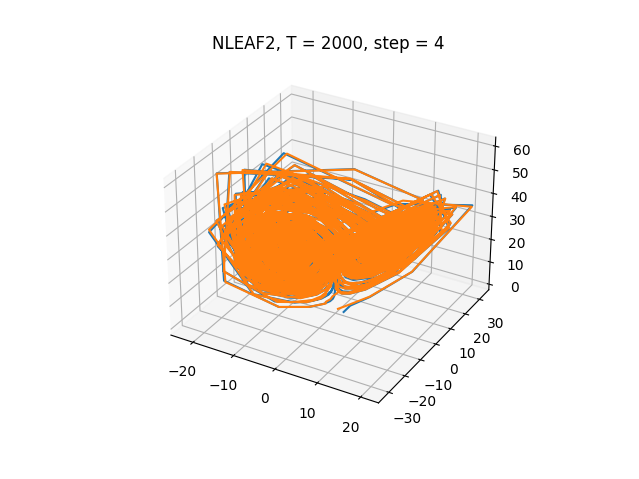}
  }\\\subfigure{
    \includegraphics[width=0.32\columnwidth]{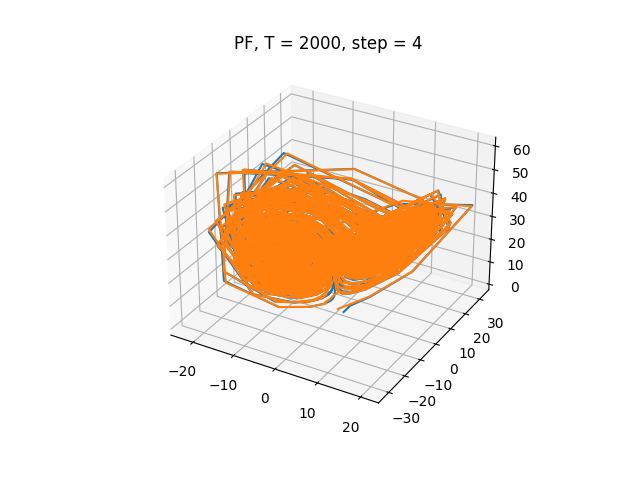}
  }\subfigure{
    \includegraphics[width=0.32\columnwidth]{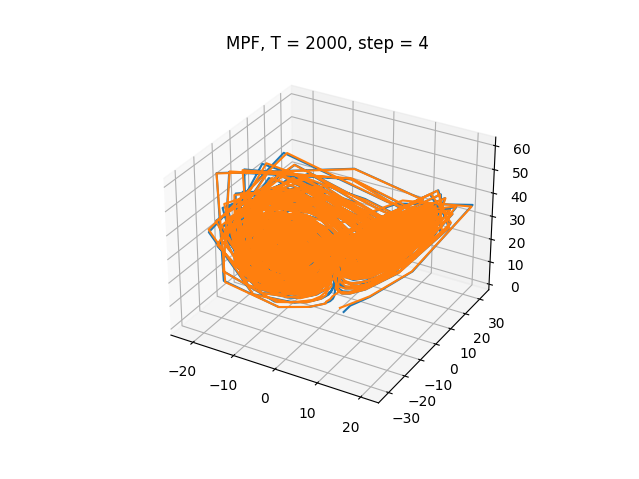}
  }\subfigure{
    \includegraphics[width=0.32\columnwidth]{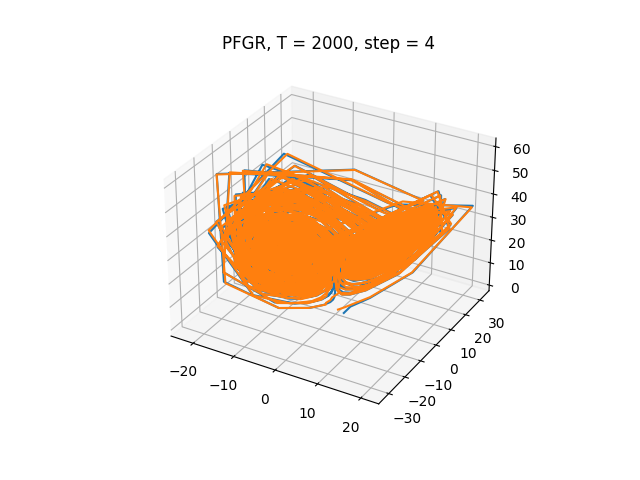}
  }\\\subfigure{
    \includegraphics[width=0.32\columnwidth]{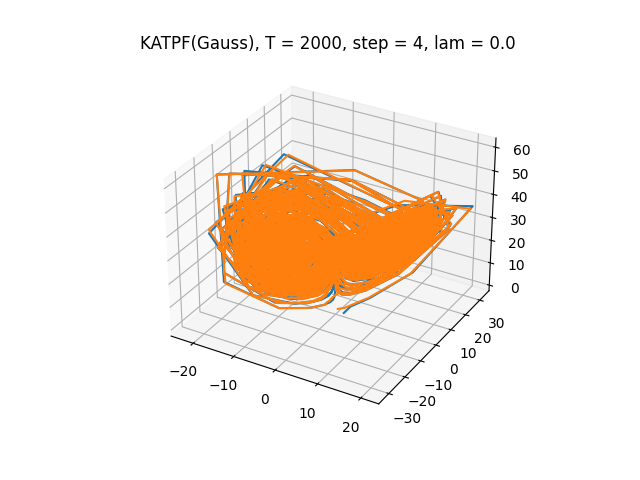}
  }\subfigure{
    \includegraphics[width=0.32\columnwidth]{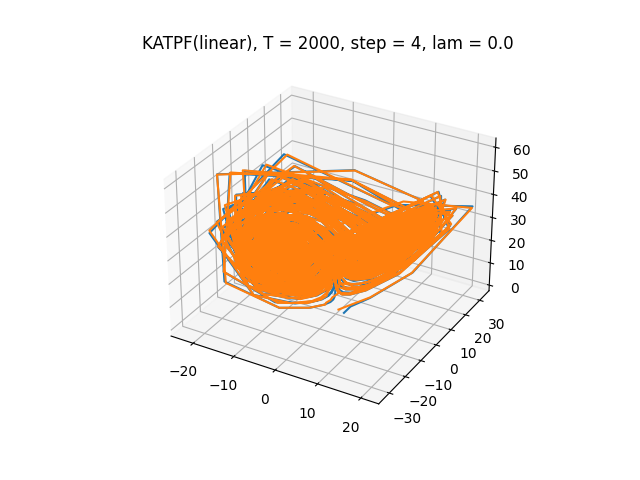}
  }\subfigure{
    \includegraphics[width=0.32\columnwidth]{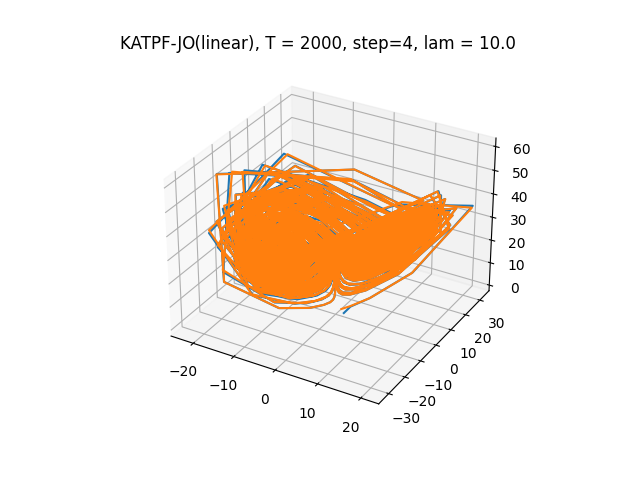}
  }
  \caption{Experimental results of the Lorenz63 system when $T=2000$ and $\Delta t=0.08$.}
  \label{fig:lorenz63step4}
\end{figure}
\newpage
\begin{table}[H]
  \centering\caption{Supplementary experimental results (RMSE) of the Lorenz63 system for ablation study. ``$\nearrow$'': the algorithmic convergence becomes worse and accordingly RMSE increases.}
  \vspace{0.5em}
  \label{tb:lorenz63_wo_OT}
  \begin{tabular}{cccccccc}
  \toprule
  
  KATPF-JO(linear)&$\lambda=0$&$\lambda=1$&$\lambda=10$&$\lambda=20$&$\lambda=50$&$\lambda^*$\\\midrule
    $T = 100,\Delta t=0.02$&0.1836&0.1652&0.1621&0.1582&0.1445&0.1398\\\midrule
    $T = 2000,\Delta t=0.02$& 0.1441&0.1335&0.1143&0.1068&0.0965&0.0843\\\midrule
    $T = 2000,\Delta t=0.04$&0.1819&0.1641&0.1380&0.1277&0.1157&0.1157\\\midrule
    $T = 2000,\Delta t=0.06$&0.1767&0.1412&0.1142&0.0973&$\nearrow$&0.0973\\\midrule
    $T = 2000,\Delta t=0.08$&0.2130&0.1795&0.1305&$\nearrow$&$\nearrow$&0.1305\\
    \bottomrule
  \end{tabular}
\end{table}

\,\\

\begin{figure}[H]
  \centering
  \subfigure{
    \includegraphics[width=0.28\columnwidth]{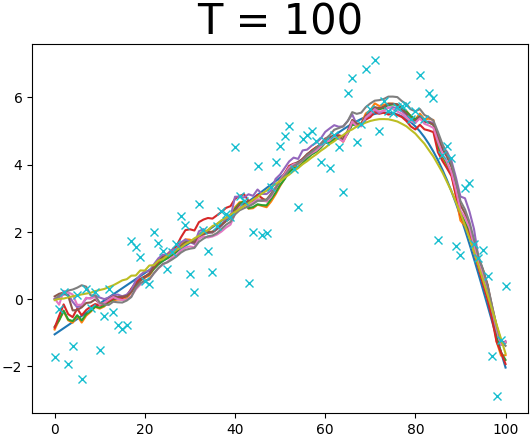}
  }\subfigure{
    \includegraphics[width=0.28\columnwidth]{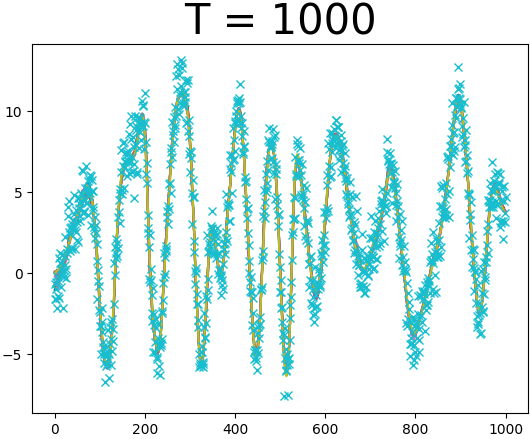}
  }\subfigure{
    \includegraphics[width=0.42\columnwidth]{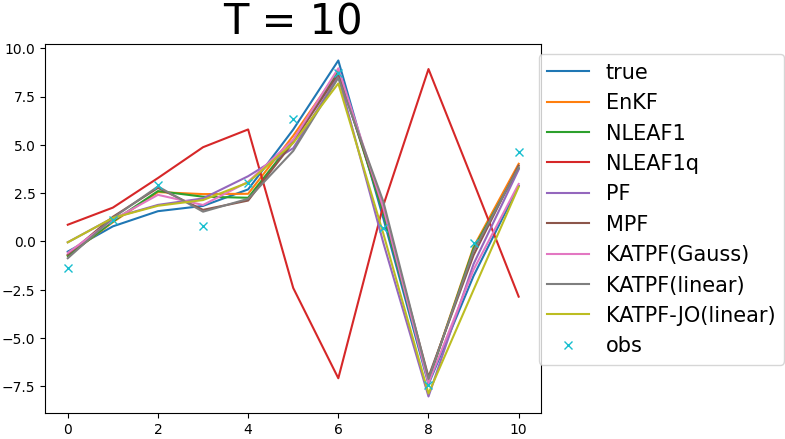}
  }
  \caption{State at location 1 over time in the Lorenz96 system.}
  \label{fig:lorenz96_loc1}
\end{figure}

\begin{figure}[H]
  \centering
  \subfigure{
    \includegraphics[width=0.28\columnwidth]{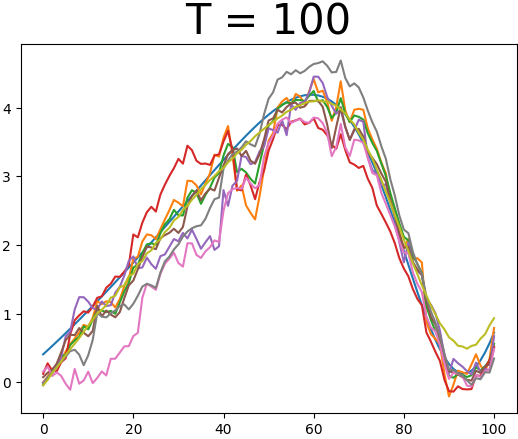}
  }\subfigure{
    \includegraphics[width=0.28\columnwidth]{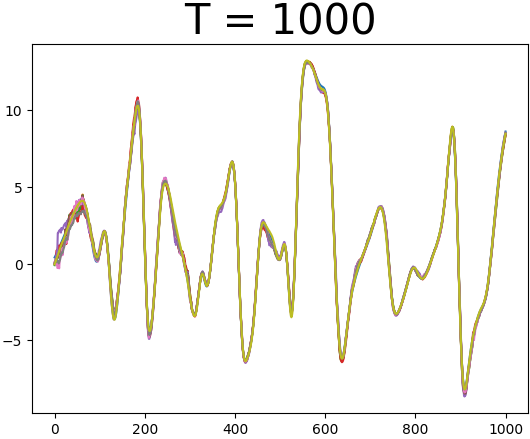}
  }\subfigure{
    \includegraphics[width=0.42\columnwidth]{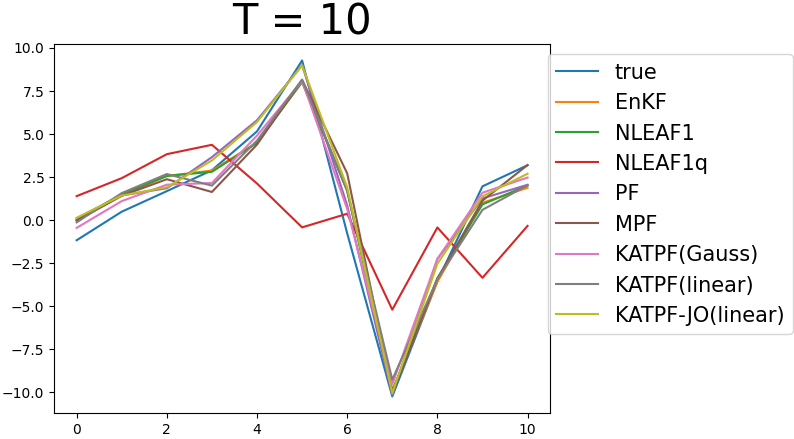}
  }
  \caption{State at location 2 over time in the Lorenz96 system (Recall that observation of this location is missing).}
  \label{fig:lorenz96_loc2}
\end{figure}

\newpage
\section{Covariance Inflation, Localization and Covariance Tapering}
\label{app:AppendixC}
In practice, covariance inflation \citep{lei2011moment}, localization \citep{hunt2007efficient} and covariance tapering \citep{furrer2006covariance} are all commonly used to adapt the ensemble methods to the high-dimensional scenarios.

Covariance inflation is usually executed after the analysis step but before the forecast step, where the particles are adjusted as follows:
\begin{equation}
  x_t^i \leftarrow \bar x_t + (1+\delta)(x_t^i-\bar x_t)
  \label{eq:inflat1}
\end{equation} with some $\delta>0$.
\citet{lei2011moment} used slightly different inflation method for PF:
\begin{equation}
  x_t^i\leftarrow x_t^i+2\delta\text{Cov}(x_t)^{\frac{1}{2}}\xi_t^i
  \label{eq:inflat2}
\end{equation}
with $\xi_t^i\stackrel{\text{i.i.d.}}{\sim}\mathcal{N}(0,1)$.
This can also be interpreted as a form of the regularized PF described in Section \ref{sec:intro} and Section \ref{subsec:pf}, with the Gaussian kernel.
Following \citet{lei2011moment}, unless otherwise specified, we use \eqref{eq:inflat2} for PFs and \eqref{eq:inflat1} for other methods after each analysis step in our experiments.

In the localization method, the overall state $x_t=(x_{t,1},\dots,x_{t,d})$ is divided into $d$ overlapping sub-states $\{x_{t,N_j}\}_{j=1}^d$ with $x_{t,N_j}=(x_{t,j-l},\dots,x_{t,j},\dots,x_{t,j+l})$ for some positive integer $l\leq\frac{d-1}{2}$ (with all indices taken modulo $d$).
Each sub-state $x_{t,N_j}$ is updated using its corresponding sub-observation $y_{t,\tilde{N}_j}:=H_t|_{N_j}x_{t,N_j}$($H_t|_{N_j}$ means the restriction of $H_t$ to the coordinate set $N_j$). Therefore, each coordinate $x_{t,j}$ is updated simultaneously in $2l+1$ sub-states $x_{t,N_{j-l}},\dots,x_{t,N_j},\dots,x_{t,N_{j+l}}$. The final update for $x_{t,j}$ is obtained by averaging the updates derived across the $2k+1$ sub-states $x_{t,N_{j-k}},\dots,x_{t,N_j},\dots,x_{t,N_{j+k}}$, where $k\leq l$.

The fundamental idea behind covariance tapering is to sparsify the covariance matrix by introducing zeros deliberately.
Let $T_\theta$ be a sparse positive definite correlation matrix (e.g., \citet{katzfuss2016understanding,furrer2006covariance}). Then the tapered covariance is computed as an element-wise product:
\begin{equation}
  \text{Cov}_{\text{tap}}(x,x')=\text{Cov}(x,x')\odot T_\theta(x,x').
  \label{eq:tapering}
\end{equation}
We can use $\widehat{\text{Cov}}_{\text{tap}}(x,x')$ instead of the sample covariance $\widehat{\text{Cov}}(x,x')$ to estimate the population covariance.
It is important to note that this technique is applicable and practical only for filtering methods that require covariance estimation, such as the EnKF.
\end{document}